\documentclass[11pt]{article}

\usepackage[margin=1in]{geometry}

\usepackage{amssymb,amsmath,amsthm,amsfonts}
\usepackage{mathtools}
\usepackage{enumitem}
\usepackage[numbers,comma,sort&compress]{natbib}
\usepackage{authblk}
\usepackage{graphicx}
\usepackage[font=small]{caption}
\usepackage[labelformat=simple]{subcaption}
\usepackage{float}
\usepackage[ruled,vlined,linesnumbered]{algorithm2e}
\usepackage{physics}
\usepackage{footnote}
\usepackage{xcolor}
\usepackage{qcircuit}

\usepackage{tikz}
\usetikzlibrary{arrows.meta}
\usetikzlibrary{positioning}

\usetikzlibrary{calc,patterns,angles,quotes}
\usepackage{siunitx}

\usepackage[colorlinks]{hyperref}

\setlist[enumerate]{itemsep=0pt}

\usepackage{thmtools}
\usepackage{thm-restate}

\declaretheorem[name=Theorem,numberwithin=section]{theorem}
\declaretheorem[name=Definition,numberwithin=section]{definition}
\declaretheorem[name=Lemma,numberwithin=section]{lemma}
\declaretheorem[name=Proposition,numberwithin=section]{proposition}
\declaretheorem[name=Corollary,numberwithin=section]{corollary}

\numberwithin{equation}{section}

\newtheorem{remark}[theorem]{Remark}

\newcommand{\eqn}[1]{(\ref{eqn:#1})}
\newcommand{\eq}[1]{(\ref{eq:#1})}

\newcommand{\thm}[1]{\hyperref[thm:#1]{Theorem~\ref*{thm:#1}}}
\newcommand{\cor}[1]{\hyperref[cor:#1]{Corollary~\ref*{cor:#1}}}
\newcommand{\defn}[1]{\hyperref[defn:#1]{Definition~\ref*{defn:#1}}}
\newcommand{\lem}[1]{\hyperref[lem:#1]{Lemma~\ref*{lem:#1}}}
\newcommand{\prop}[1]{\hyperref[prop:#1]{Proposition~\ref*{prop:#1}}}
\newcommand{\fig}[1]{\hyperref[fig:#1]{Figure~\ref*{fig:#1}}}
\newcommand{\tab}[1]{\hyperref[tab:#1]{Table~\ref*{tab:#1}}}
\newcommand{\algo}[1]{\hyperref[algo:#1]{Algorithm~\ref*{algo:#1}}}
\renewcommand{\sec}[1]{\hyperref[sec:#1]{Section~\ref*{sec:#1}}}
\newcommand{\append}[1]{\hyperref[append:#1]{Appendix~\ref*{append:#1}}}
\newcommand{\fac}[1]{\hyperref[fac:#1]{Fact~\ref*{fac:#1}}}
\newcommand{\lin}[1]{\hyperref[lin:#1]{Line~\ref*{lin:#1}}}
\newcommand{\fnote}[1]{\hyperref[fnote:#1]{Footnote~\ref*{fnote:#1}}}

\SetKwInput{KwInput}{Input}
\SetKwInput{KwOutput}{Output}

\newcommand{\specialcell}[2][c]{%
  \begin{tabular}[#1]{@{}c@{}}#2\end{tabular}}

\def\>{\rangle}
\def\<{\langle}

\newcommand{\N}{\mathbb{N}}
\newcommand{\Z}{\mathbb{Z}}
\newcommand{\R}{\mathbb{R}}
\newcommand{\C}{\mathbb{C}}

\newcommand{\E}{\mathbb{E}}

\newcommand{\B}{\mathrm{B}}
\newcommand{\K}{\mathrm{K}}
\newcommand{\cvx}{\mathrm{C}}
\renewcommand{\S}{\mathrm{S}}
\renewcommand{\O}{O}
\newcommand{\Cone}{\mathrm{C}}

\newcommand{\memcost}{C_\mathrm{MEM}}
\newcommand{\arithcost}{n^{5}+n^{4.5}/\epsilon}
\newcommand{\querycost}{n^{3}+n^{2.5}/\epsilon}

\newcommand{\sinsq}{\mathrm{sin}^{2}}
\newcommand{\median}{\mathrm{median}}
\newcommand{\Cheby}{\mathrm{CB}}

\let\var\relax

\DeclareMathOperator{\vol}{Vol}
\DeclareMathOperator{\poly}{poly}
\DeclareMathOperator{\spn}{span}

\DeclareMathOperator{\var}{Var}
\DeclareMathOperator{\st}{St}

\newcommand{\bbr}{\mathbb{R}}

\newcommand{\rmk}{\mathrm{K}}
\newcommand{\rmb}{\mathrm{B}}
\newcommand{\rms}{\mathrm{S}}
\newcommand{\rma}{\mathrm{A}}

\newcommand{\rml}{\mathrm{L}}
\newcommand{\rmp}{\mathrm{P}}
\newcommand{\rmq}{\mathrm{Q}}

\newcommand{\rmke}{\mathrm{K}_{\epsilon}}
\newcommand{\rmkse}{\mathrm{K}_{\sqrt{\epsilon}n^{1/4}}}

\renewcommand{\d}{\mathrm{d}}
\renewcommand{\emptyset}{\varnothing}

\newcommand{\range}[1]{[#1]}

\newcommand{\hd}[1]{\vspace{4mm} \noindent \textbf{#1}\ \ }

\let\oldnl\nl
\newcommand{\nonl}{\renewcommand{\nl}{\let\nl\oldnl}}

\makeatletter
\def\Ddots{\mathinner{\mkern1mu\raise\p@
\vbox{\kern7\p@\hbox{.}}\mkern2mu
\raise4\p@\hbox{.}\mkern2mu\raise7\p@\hbox{.}\mkern1mu}}
\makeatother

\begin{document}

\begin{titlepage}
\clearpage

\title{Quantum algorithm for estimating volumes of convex bodies}

\author{Shouvanik Chakrabarti\thanks{Department of Computer Science, Institute for Advanced Computer Studies, and Joint Center for Quantum Information and Computer Science, University of Maryland. Email address: \{shouv,tongyang,xwu\}@cs.umd.edu, \{amchilds,shung\}@umd.edu}\qquad Andrew M.\ Childs$^{*}$\qquad Shih-Han Hung$^{*}$

Tongyang Li$^{*}$\qquad Chunhao Wang\thanks{Department of Computer Science, University of Texas at Austin. Email address: chunhao@cs.utexas.edu}\qquad Xiaodi Wu$^{*}$}

\normalsize
\renewcommand\Authsep{  }
\renewcommand\Authands{  }

\date{}

\maketitle
\thispagestyle{empty}

\begin{abstract}
Estimating the volume of a convex body is a central problem in convex geometry and can be viewed as a continuous version of counting. We present a quantum algorithm that estimates the volume of an $n$-dimensional convex body within multiplicative error $\epsilon$ using $\tilde{O}(n^{3}+n^{2.5}/\epsilon)$ queries to a membership oracle and $\tilde{O}(n^{5}+n^{4.5}/\epsilon)$ additional arithmetic operations. For comparison, the best known classical algorithm uses $\tilde{O}(n^{4}+n^{3}/\epsilon^{2})$ queries and $\tilde{O}(n^{6}+n^{5}/\epsilon^{2})$ additional arithmetic operations. To the best of our knowledge, this is the first quantum speedup for volume estimation. Our algorithm is based on a refined framework for speeding up simulated annealing algorithms that might be of independent interest. This framework applies in the setting of ``Chebyshev cooling'', where the solution is expressed as a telescoping product of ratios, each having bounded variance. We develop several novel techniques when implementing our framework, including a theory of continuous-space quantum walks with rigorous bounds on discretization error. To complement our quantum algorithms, we also prove that volume estimation requires $\Omega(\sqrt n+1/\epsilon)$ quantum membership queries, which rules out the possibility of exponential quantum speedup in $n$ and shows optimality of our algorithm in $1/\epsilon$ up to poly-logarithmic factors.
\end{abstract}

\end{titlepage}

\newpage


\section{Introduction}\label{sec:intro}
Estimating the volume of a convex body is a central challenge in theoretical computer science. Volume estimation is a basic problem in convex geometry and can be viewed as a continuous version of counting. Furthermore, algorithms for a generalization of volume estimation---namely log-concave sampling---can be directly used to perform convex optimization, and hence can be widely applied to problems in statistics, machine learning, operations research, etc. See the survey~\cite{vempala2005geometric} for a more comprehensive introduction.

Volume estimation is a notoriously difficult problem. References \cite{barany1987computing,elekes1986geometric} proved that any \emph{deterministic algorithm} that approximates the volume of an $n$-dimensional convex body within a factor of $n^{o(n)}$ necessarily makes exponentially many queries to a membership oracle for the convex body. Furthermore, Refs.~\cite{dyer1988complexity,khachiyan1988complexity,khachiyan1989problem} showed that estimating the volume exactly (deterministically) is \#P-hard, even for explicitly described polytopes.

Surprisingly, volumes of convex bodies can be approximated efficiently by \emph{randomized algorithms}. Dyer, Frieze, and Kannan \cite{dyer1991random} gave the first polynomial-time randomized algorithm for estimating the volume of a convex body in $\R^{n}$. They present an iterative algorithm that constructs a sequence of convex bodies. The volume of the convex body of interest can be written as the telescoping product of the ratios of the volumes of consecutive convex bodies, and these ratios are estimated by Markov chain Monte Carlo (MCMC) methods via random walks inside these convex bodies. The algorithm in~\cite{dyer1991random} has complexity\footnote{Throughout the paper, $\tilde{O}$ omits factors in $\poly(\log R/r, \log 1/\epsilon, \log n)$ where $R$ and $r$ are defined in \eqn{ball-relationship}.} $\tilde{O}(n^{23})$ with multiplicative error $\epsilon=\Theta(1)$. Subsequent work~\cite{lovasz1990mixing,applegate1991sampling,dyer1991computing,lovasz1993random,kannan1997random,Lovasz99,LV06} improved the design of the iterative framework and the choice of the random walks. The state-of-the-art algorithm for estimating the volume of a general convex body~\cite{lovasz2006simulated} uses $\tilde{O}(n^{4})$ queries to the oracle for the convex body and $\tilde{O}(n^{6})$ additional arithmetic operations.

It is natural to ask whether quantum computers can solve volume estimation even faster than classical randomized algorithms. Although there are frameworks with potential quantum speedup for simulated annealing algorithms in general, with volume estimation as a possible application~\cite{wocjan2009quantum}, we are not aware of any previous quantum speedup for volume estimation. There are several reasons to develop such a result. First, quantum algorithms for volume estimation can be seen as performing a continuous version of quantum counting~\cite{brassard1998quantum,brassard2002amplitude}, a key algorithmic technique with wide applications in quantum computing. Second, quantum algorithms for volume estimation can exploit quantum MCMC methods (e.g., \cite{Richter2007,wocjan2008speedup,montanaro2015quantum}), and a successful quantum volume estimation algorithm may illuminate the application of quantum MCMC methods in other scenarios. Third, there has been recent progress on quantum algorithms for convex optimization~\cite{vanApeldoorn2018convex,chakrabarti2018quantum}, so it is natural to study the closely related task of estimating volumes of convex bodies.

\paragraph{Formulation}
Given a convex set $\K \subset \R^n$, we consider the problem of estimating its volume
\begin{align}\label{prb:volume}
\vol(\rmk):=\int_{x\in \rmk}\d x.
\end{align}
To get a basic sense about the location of $\K$, we assume that it contains the origin. Furthermore, we assume that we are given inner and outer bounds on $\rmk$, namely
\begin{align}\label{eqn:ball-relationship}
\rmb_{2}^{n}(0,r) \subseteq \rmk \subseteq \rmb_{2}^{n}(0,R),
\end{align}
where $\rmb_{2}^{n}(x,l)$ is the ball of radius $l$ in $\ell_{2}$-norm centered at $x\in\R^{n}$. Denote $D:=R/r$.

We consider the very general setting where the convex body $\rmk$ is only specified by an oracle. In particular, we have a \emph{membership oracle}\footnote{The membership oracle is commonly used in convex optimization research (see for example~\cite{grotschel2012geometric}). This model is not only general but also of practical interest. For instance, when $\rmk$ is a bounded convex polytope, the membership oracle can be efficiently implemented by checking if all its linear constraints are satisfied; see also~\cite{lee2018convergence}.} for $\rmk$ that determines whether a given $x\in\R^{n}$ belongs to $\rmk$. The efficiency of volume estimation is then measured by the number of queries to the membership oracle (i.e., the \emph{query complexity}) and the total number of other arithmetic operations.

In the quantum setting, the membership oracle is a unitary operator $O_{\rmk}$. Specifically, we have
\begin{align}\label{eqn:oracle-defn}
O_{\rmk}|x,0\>=|x,\delta[x \in \rmk]\>\qquad\forall x\in\R^{n},
\end{align}
where $\delta[P]$ is $1$ if $P$ is true and $0$ if $P$ is false.\footnote{Here $x$ can be approximated just as in the classical algorithms, such as with fixed-point numbers. Our algorithmic approach is robust under discretization (see \sec{implement-hit-and-run}), and our quantum lower bound holds even when $x$ is stored with arbitrary precision (\sec{quantum-lower}). We mostly assume for convenience that $O_{\rmk}$ operates on $x\in\R^{n}$, since this neither presents a serious obstacle nor conveys significant power.} In other words, we allow coherent superpositions of queries to the membership oracle. If the classical membership oracle can be implemented by an explicit classical circuit, then the corresponding quantum membership oracle can be implemented by a quantum circuit of about the same size. Therefore, the quantum query model provides a useful framework for understanding the quantum complexity of volume estimation.

\subsection{Contributions}\label{sec:contributions}
Our main result is a quantum algorithm for estimating volumes of convex bodies:

\begin{restatable}[Main Theorem]{theorem}{maintheorem}\label{thm:main}
Let $\K \subset \R^n$ be a convex set with $\rmb_{2}^{n}(0,r) \subseteq \rmk \subseteq \rmb_{2}^{n}(0,R)$. Assume $0<\epsilon<1/2$. Then there is a quantum algorithm that returns a value $\widetilde{\vol(\K)}$ satisfying
\begin{align}\label{eqn:volume-estimation-defn}
\frac{1}{1+\epsilon}\vol(\K)\leq\widetilde{\vol(\K)}\leq(1+\epsilon)\vol(\K)
\end{align}
with probability at least $2/3$ using $\tilde{O}(\querycost)$ quantum queries to the membership oracle $O_{\K}$ (defined in \eqn{oracle-defn}) and $\tilde{O}(\arithcost)$ additional arithmetic operations.\footnote{Arithmetic operations (e.g., addition, subtraction, multiplication, and division) can be in principle implemented by a universal set of quantum gates using the Solovay-Kitaev Theorem~\cite{dawson2006solovay} up to a small overhead. In our quantum algorithm, the number of arithmetic operations is dominated by $n$-dimensional matrix-vector products computed in superposition for rounding the convex body (see \sec{round-quantum}).}
\end{restatable}

To the best of our knowledge, this is the first quantum algorithm that achieves quantum speedup for this fundamental problem, compared to the classical state-of-the-art algorithm~\cite{lovasz2006simulated,cousins2015bypassing} that uses $\tilde{O}(n^{4}+n^{3}/\epsilon^{2})$ classical queries and $\tilde{O}(n^{6}+n^{5}/\epsilon^{2})$ additional arithmetic operations.\footnote{\label{fnote:well-rounded}This is achieved by applying~\cite{lovasz2006simulated} to preprocess the convex body to be well-rounded (i.e. $R/r=O(\sqrt{n})$) using $\tilde{O}(n^{4})$ queries and then applying~\cite{cousins2015bypassing} using $\tilde{O}(n^{3}/\epsilon^{2})$ queries to estimate the volume of the (well-rounded) convex body. The number of additional arithmetic operations has an overhead of $O(n^{2})$ due to the affine transformation.} Furthermore, our quantum algorithm not only achieves a quantum speedup in query complexity, but also in the number of arithmetic operations for executing the algorithm. This differs from previous quantum algorithms for convex optimization~\cite{vanApeldoorn2018convex,chakrabarti2018quantum} where only the query complexity is improved, but the gate complexity is the same as that of the classical state-of-the-art algorithm~\cite{lee2015faster,lee2018efficient}.

On the other hand, we prove in \sec{quantum-lower-n} that volume estimation with $\epsilon=\Theta(1)$ requires $\Omega(\sqrt{n})$ quantum queries to the membership oracle, ruling out the possibility of achieving superpolynomial quantum speedup for volume estimation. Classically, the best-known lower bound on the query complexity of volume estimation is $\tilde{\Omega}(n^{2})$ due to Rademacher and Vempala~\cite{rademacher2008dispersion}, but there are technical difficulties to lift it to a quantum lower bound (see \sec{lower-intro}). For the dependence on $1/\epsilon$, we establish a quantum query lower bound of $\Omega(1/\epsilon)$, and the same argument shows a classical query lower bound of $\Omega(1/\epsilon^{2})$ (see \sec{quantum-lower-eps}). As a result, our quantum algorithm in \thm{main} achieves a provable quadratic quantum speedup in $1/\epsilon$ and is optimal in $1/\epsilon$ up to poly-logarithmic factors.

Technically, we refine a framework for achieving quantum speedups of simulated annealing algorithms, which might be of independent interest. Our framework applies to MCMC algorithms with cooling schedules that ensure each ratio in a telescoping product has bounded variance, an approach known as \emph{Chebyshev cooling.} Furthermore, we propose several novel techniques when implementing this framework, including a theory of continuous-space quantum walks with rigorous bounds on discretization error, a quantum algorithm for nondestructive mean estimation, and a quantum algorithm with interlaced rounding and volume estimation of convex bodies (as described further in \sec{techniques} below). In principle, our techniques apply not only to the integral of the identity function (as in \thm{main}), but could also be applied to any log-concave function defined on a convex body, following the approach in~\cite{lovasz2006fast}.

We summarize our main results in \tab{main-volume}.

\begin{table}[htbp]
\centering
\resizebox{1\columnwidth}{!}{%
\begin{tabular}{|c||c|c|}
\hline
 & Classical bounds & Quantum bounds {(this paper)} \\ \hline\hline
Query complexity & $\tilde{O}(n^{4}+n^{3}/\epsilon^{2})$~\cite{lovasz2006simulated,cousins2015bypassing}, $\tilde{\Omega}(n^{2})$~\cite{rademacher2008dispersion} & $\tilde{O}(\querycost)$, $\Omega(\sqrt{n}+1/\epsilon)$ \\ \hline
Total complexity & $\tilde{O}\big((n^{2}+\memcost)\cdot (n^{4}+n^{3}/\epsilon^{2})\big)$~\cite{lovasz2006simulated,cousins2015bypassing} & $\tilde{O}\big((n^{2}+\memcost)\cdot(\querycost)\big)$ \\ \hline
\end{tabular}
}
\caption{Summary of complexities of volume estimation, where $n$ is the dimension of the convex body, $\epsilon$ is the multiplicative precision of volume estimation, and $\memcost$ is the cost of applying the membership oracle once. Total complexity refers to the cost of the of queries plus the number of additional arithmetic operations.}
\label{tab:main-volume}
\end{table}
\vspace{-2mm}

\subsection{Techniques}\label{sec:techniques}
We now summarize the key technical aspects of our work.

\subsubsection{Classical volume estimation framework}\label{sec:classical-volume-framework}

\paragraph{Volume estimation by simulated annealing}
The volume of a convex body $\K$ can be estimated using simulated annealing. Consider the value
\begin{align}\label{eqn:Za-techniques}
Z(a):=\int_{\K}e^{-a\|x\|_{2}}\,\d x,
\end{align}
where $\|x\|_{2} := \sqrt{x_{1}^{2}+\cdots+x_{n}^{2}}$ is the $\ell_{2}$-norm of $x$. On the one hand, $Z(0)=\vol(\K)$; on the other hand, because $e^{-\|x\|_{2}}$ decays exponentially fast with $\|x\|_{2}$, taking a large enough $a$ ensures that the vast majority of $Z(a)$ concentrates near 0, so it can be well approximated by integrating on a small ball centered at 0. Therefore, a natural strategy is to consider a sequence $a_{0}>a_{1}>\cdots>a_{m}$ with $a_{0}$ sufficiently large and $a_{m}$ close to 0. We consider a simulated annealing algorithm that iteratively changes $a_{i}$ to $a_{i+1}$ and estimates $\vol(\K)$ by the telescoping product
\begin{align}\label{eqn:telescoping-volume}
\vol(\K)\approx Z(a_{m})=Z(a_{0})\prod_{i=0}^{m-1}\frac{Z(a_{i+1})}{Z(a_{i})}.
\end{align}
In the $i^{\text{th}}$ step, a random walk is used to sample the distribution over $\K$ with density proportional to $e^{-a_{i}\|x\|_{2}}$. Denote one such sample by $X_{i}$, and let $V_{i}:=e^{(a_{i}-a_{i+1})\|X_{i}\|_{2}}$. Then we have
\begin{align}\label{eqn:V_i-techniques}
\E[V_{i}]=\int_{\K}e^{(a_{i}-a_{i+1})\|x\|_{2}}\frac{e^{-a_{i}\|x\|_{2}}}{Z(a_{i})}\,\d x=\int_{\K}\frac{e^{-a_{i+1}\|x\|_{2}}}{Z(a_{i})}\,\d x=\frac{Z(a_{i+1})}{Z(a_{i})}.
\end{align}
Therefore, each ratio $\frac{Z(a_{i+1})}{Z(a_{i})}$ can be estimated by taking i.i.d.\ samples $X_{i}$, computing the corresponding $V_{i}$s, and taking their average.

We can analyze this volume estimation algorithm by considering its behavior at three levels:
\begin{itemize}[leftmargin=*]
\item[1)] \textsf{High level:} The algorithm follows the simulated annealing framework described above, where the volume is estimated by a telescoping product as in \eqn{telescoping-volume}.
\item[2)] \textsf{Middle level:} The number of i.i.d.\ samples used to estimate $\E[V_{i}]$ (a ratio in the telescoping product given by \eqn{V_i-techniques}) is small. Intuitively, the annealing schedule should be slow enough that $V_{i}$ has small variance.
\item[3)] \textsf{Low level:} The random walk converges fast so that we can take each i.i.d.\ sample of $V_{i}$ efficiently.
\end{itemize}

\paragraph{Classical volume estimation algorithm}
Our approach follows the classical volume estimation algorithm in \cite{lovasz2006simulated} (see also \sec{volume-estimation-review}). At the \textsf{high level}, it is a simulated annealing algorithm that estimates the volume of an alternative convex body $\K'$ produced by the \emph{pencil construction}, which intersects a cylinder $[0,2R/r]\times \K$ and a cone $\Cone:=\{x\in\R^{n+1}:x_{0}\geq 0,\|x\|_{2}\leq x_{0}\}$. This construction shares the same intuition as above, but replaces the integral \eqn{Za-techniques} by $Z(a)=\int_{\K'}e^{-ax_{0}}\,\d x$ because it is easier to calculate while can be directly used to estimate $\vol(\K)$ when $a\approx 0$ by a standard Monte Carlo approach (see \lem{pencil-to-original}).

Without loss of generality, assume that $r=1$. Lov{\'a}sz and Vempala \cite{lovasz2006simulated} proved that if we take the sequence $a_{0}>\cdots>a_{m}$ where $a_{0}=2n$, $a_{i+1}=(1-\frac{1}{\sqrt{n}})a_{i}$, and $m=\tilde{O}(\sqrt{n})$, then $Z(a_{0})\approx \int_{\Cone}e^{-a_{0}x_{0}}\,\d x$ and
\begin{align}
\var[V_{i}^{2}]=O(1)\cdot\E[V_{i}]^{2},~\forall\,i\in\range{m},
\end{align}
i.e., the variance of $V_{i}$ is bounded by a constant multiple of the square of its expectation. Such a simulated annealing schedule is known as \emph{Chebyshev cooling} (see also \sec{nondestructive}). This establishes the \textsf{middle-level} requirement of the simulated annealing framework. Furthermore, \cite{lovasz2006simulated} proves that the product of the average of $\tilde{O}(\sqrt{n}/\epsilon^{2})$ i.i.d.\ samples of $V_{i}$ for all $i\in\range{m}$ gives an estimate of $\vol(\K')$ within multiplicative error $\epsilon$ with high success probability.

At the \textsf{low level}, Ref.~\cite{lovasz2006simulated} uses a \emph{hit-and-run walk} to sample $X_{i}$. In this walk, starting from a point $p$, we uniformly sample a line $\ell$ through $p$ and move to a random point along the chord $\ell\cap\K$ with density proportional to $e^{-ax_{0}}$ (see \sec{hit-and-run} for details). Reference~\cite{LV06} analyzes the convergence of the hit-and-run walk, proving that it converges to the distribution over $\K$ with density proportional to $e^{-ax_{0}}$ within $\tilde{O}(n^{3})$ steps, assuming that $\K$ is \emph{well-rounded} (i.e., $R/r=O(\sqrt{n})$).

Finally, Ref.~\cite{lovasz2006simulated} constructs an affine transformation that transforms a general $\K$ to be well-rounded with $\tilde{O}(n^{4})$ classical queries to its membership oracle, hence removing the constraint of the previous steps that $\K$ be well-rounded. Because the affine transformation is an $n$-dimensional matrix-vector product, this introduces an overhead of $O(n^{2})$ in the number of arithmetic operations.

Overall, the algorithm has $\tilde{O}(\sqrt{n})$ iterations, where each iteration takes $\tilde{O}(\sqrt{n}/\epsilon^{2})$ i.i.d.\ samples, and each sample takes $\tilde{O}(n^{3})$ steps of the hit-and-run walk. In total, the query complexity is
\begin{align}
\tilde{O}(\sqrt{n})\cdot\tilde{O}(\sqrt{n}/\epsilon^{2})\cdot\tilde{O}(n^{3})=\tilde{O}(n^{4}/\epsilon^{2}).
\end{align}
The number of additional arithmetic operations is $\tilde{O}(n^{4}/\epsilon^{2})\cdot O(n^{2})=\tilde{O}(n^{6}/\epsilon^{2})$ due to the affine transformation for rounding the convex body.

\subsubsection{Quantum algorithm for volume estimation}
It is natural to consider a quantum algorithm for volume estimation following the classical framework in \sec{classical-volume-framework}. A naive attempt might be to develop a quantum walk that achieves a generic quadratic speedup of the mixing time. However, this is unfortunately difficult to achieve in general. Quantum walks are unitary processes that do not converge to stationary distributions in the classical sense. As a result, alternative and indirect quantum analogues of mixing properties of Markov chains have been proposed and studied (see \sec{quantum-MCMC-literature} for more detail). None of these methods provide a direct replacement for classical mixing, and we cannot directly apply them in our context.

Instead, we adapt one of the frameworks proposed in~\cite{wocjan2008speedup}. To give a quantum speedup for volume estimation by this method, we address the following additional technical challenges:
\begin{itemize}[leftmargin=*]
\item\textbf{Quantum walks in continuous space:} Quantum walks have mainly been studied in discrete spaces~\cite{szegedy2004quantum,MNRS11}, and we need to understand how to define a quantum counterpart of the hit-and-run walk.

\item\textbf{Quantum mean estimation:} Quantum counting~\cite{brassard2002amplitude} is a general tool for estimating a probability $p\in[0,1]$ with quadratic quantum speedup compared to classical sampling. However, estimating the mean of an unbounded random variable with a quantum version of Chebyshev concentration requires more advanced tools.

\item\textbf{Rounding:} Classically, rounding a general convex body takes $\tilde{O}(n^{4})$ queries~\cite{lovasz2006simulated}, more expensive than volume estimation of a well-rounded body using $\tilde{O}(n^{3}/\epsilon^{2})$ queries~\cite{cousins2015bypassing}. To achieve an overall quantum speedup, we also need to give a fast quantum algorithm for rounding convex bodies.

\item\textbf{Error analysis of the quantum hit-and-run walk:} We must bound the error incurred when implementing the quantum walk on a digital quantum computer with finite precision. Existing classical error analyses (e.g., \cite{frieze1999log}) do not automatically cover the quantum case.
\end{itemize}

We develop several novel techniques to resolve these issues:

\paragraph{Theory of continuous-space quantum walks (\sec{continuous-q-walk-theory})}
Our first technical contribution is to develop a quantum implementation of the \textsf{low-level} framework, i.e., to replace the classical hit-and-run walk by a \emph{quantum hit-and-run walk.} However, although quantum walks in discrete spaces have been well studied (see for example \cite{szegedy2004quantum,MNRS11}), we are not aware of comparable results that can be used to analyze spectral properties and mixing times of quantum walks in continuous space. Here we describe a framework for continuous-space quantum walks that can be instantiated to give a quantum version of the hit-and-run walk. In particular, we formally define such walks and analyze their spectral properties, generalizing Szegedy's theory~\cite{szegedy2004quantum} to continuous spaces (\sec{continuous-space-qwalk-defn}). We also show a direct correspondence between the stationary distribution of a classical walk and a certain eigenvector of the corresponding quantum walk (\sec{continuous-space-qwalk-stationary}).

\paragraph{Quantum volume estimation algorithm via simulated annealing (\sec{quantum-volume-estimation})}
Having described a quantum hit-and-run walk, the next step is to understand the \textsf{high-level} simulated annealing framework. As mentioned above, it is nontrivial to directly prepare stationary states of quantum walks. In this paper, we follow a quantum MCMC framework proposed by~\cite{wocjan2008speedup} that can prepare stationary states of quantum walks by simulated annealing (see \sec{quantum-MCMC}). In this framework, we have a sequence of slowly-varying Markov chains, and the stationary state of the initial Markov chain can be efficiently prepared. In each iteration, we apply fixed-point amplitude amplification of the quantum walk operator~\cite{grover2005different} due to Grover to transform the current stationary state to the next one; compared to classical slowly-varying Markov chains, the convergence rate of such quantum procedure is \emph{quadratically better in spectral gap.}

Our \textbf{main technical contribution} is to show how to adapt the Chebyshev cooling schedule in~\cite{lovasz2006simulated} to the quantum MCMC framework in~\cite{wocjan2008speedup} using our quantum hit-and-run walk. The conductance lower bound together with the classical $\tilde{O}(n^{3})$ mixing time imply that we can perform one step of fixed-point amplitude amplification using $\tilde{O}(n^{1.5})$ queries to $O_{\K}$. Furthermore, the inner product between consecutive stationary states is a constant. These two facts ensure that the stationary state in each iteration can be prepared with $\tilde{O}(n^{1.5})$ queries to the membership oracle $O_{\K}$. The total number of iterations is still $\tilde{O}(\sqrt{n})$, as in the classical case.

\paragraph{Quantum algorithm for nondestructive mean estimation (\sec{nondestructive})}
In the next step, we consider how to estimate each ratio in the telescoping product at the \textsf{middle level}. The basic tool is quantum counting~\cite{brassard2002amplitude}, which estimates a probability $p\in[0,1]$ with error $\epsilon$ and high success probability using $O(1/\epsilon)$ quantum queries, a quadratic speedup compared to the classical complexity $O(1/\epsilon^{2})$. However, in our case we need to estimate the expectation of a random variable with bounded variance. We use the ``quantum Chebyshev inequality'' developed in \cite{hamoudi2019Chebyshev} which truncates the random variable with reasonable upper and lower bounds and then reduces to quantum counting; see \sec{quantum-Chebyshev}.\footnote{A related technique is the quantum Monte Carlo method of Montanaro~\cite{montanaro2015quantum}. Here we use~\cite{hamoudi2019Chebyshev} for two reasons: first, it has the advantage of handling multiplicative instead of additive errors, which is appropriate for estimating the telescoping ratios. Second, its quantum algorithm is based on amplitude estimation and hence can readily be made nondestructive, as discussed below.} Compared to the classical counterpart, it achieves quadratic speedup in the dependences on both variance and multiplicative error.

There is an additional technical difficulty in quantum simulated annealing: classically, it is implicitly assumed that in the $(i+1)^{\text{st}}$ iteration we have samples to the stationary distribution in the $i^{\text{th}}$ iteration. Applying existing quantum mean estimation techniques to the quantum stationary state in the $i^{\text{th}}$ iteration would ruin that state and make it hard to use in the subsequent $(i+1)^{\text{st}}$ iteration.
 To resolve this issue, we estimate the mean \emph{nondestructively} in the quantum Chebyshev inequality while keeping its quadratic speedup in the error dependence using a nondestructive amplitude estimation technique developed in \cite{harrow2019adaptive}. Nondestructive mean estimation relies on the following observation: applying amplitude estimation on a state $|\psi\>$ results with high probability in the measurement collapsing to one of two states $|\psi_{+}\>,|\psi_{-}\>$ with constant overlap with $\psi$. The algorithm repeatedly projects these states onto $|\psi\>$: if the projection is successful then the state is restored, otherwise amplitude estimation can be performed again to obtain $|\psi_{+}\>,|\psi_{-}\>$ and the projection can be repeated. Due to the constant overlap, $\poly(\log(\delta^{-1}))$ repititions suffice to ensure that at least one of the projections succeeds with probability $\delta$. It remains to implement the required projection efficiently: we show how this can be accomplished using quantum walk operators corresponding to the Markov Chains in the MCMC framework; see \sec{proof-error-analysis}.

In our quantum volume estimation algorithm, we apply the quantum Chebyshev inequality under the same compute-uncompute procedure. This gives a quadratic speedup in $\epsilon^{-1}$ when estimating the $\E[V_{i}]$ in \eqn{V_i-techniques}, so that $\tilde{O}(\sqrt{n}/\epsilon)$ copies of the stationary state suffice\footnote{It is possible to use fewer copies of the stationary state. See \fnote{copies}.} (see \lem{Chebyshev-lemma}).

\paragraph{Quantum algorithm for volume estimation with interlaced rounding (\sec{round-quantum})}
The stationary states of the quantum hit-and-run walk can be prepared with $\tilde{O}(n^{1.5})$ queries to $O_{\K}$ only when the corresponding density functions are \emph{well-rounded}, i.e., every level set with probability $\mu$ contains a ball of radius $\mu r$ and the variance of the density is bounded by $R^2$, where $R/r = O(\sqrt{n})$.\footnote{When the density function is uniform in $\K$, this definition of well-roundedness reduces to that in \fnote{well-rounded}. The definition of level sets is the same as in~\cite{lovasz2006simulated}.} It remains to show how to ensure that the convex body is well-rounded.

Classically, Ref.~\cite{lovasz2006simulated} gave a rounding algorithm that transforms a convex body to ensure that all the densities sampled in the volume estimation algorithm are well-rounded. This algorithm uses $\tilde{O}(n^{4})$ queries, via $\tilde{O}(n)$ iterations of simulated annealing. A quantization of this algorithm along the same lines as detailed above gives an algorithm with $\tilde{O}(n^{3.5})$ quantum queries.

To improve over that approach, we instead follow a classical framework for directly rounding logconcave densities~\cite{lovasz2006fast}. The rounding is interlaced with the volume estimation algorithm, so that in each iteration of the simulated annealing framework, we use some of the samples to calculate an affine transformation that makes the next stationary state well-rounded. This ensures that the quantum hit-and-run walk continues to take only $\tilde{O}(n^{1.5})$ queries for each sample. Our algorithm maintains $\tilde{O}(n)$ extra quantum states for rounding, and the quantum hit-and-run walk is used to transform them from one stationary distribution to the next. In each iteration, we use a nondestructive measurement to sample the required affine transformation. With $\tilde{O}(\sqrt{n})$ iterations this results in an additional $\tilde{O}(\sqrt{n})\cdot\tilde{O}(n)\cdot\tilde{O}(n^{1.5}) = \tilde{O}(n^3)$ cost for rounding.

We also show that this framework can be used as a preprocessing step that puts the convex body itself in well-rounded position (i.e., $\rmb_2(0,r) \subseteq \K \subseteq \rmb_2(0,R)$ with $R/r = O(\sqrt{n})$) using $\tilde{O}(n^3)$ quantum queries. Putting a convex body in well-rounded position implies that several random walks used in simulated annealing algorithms (including the hit-and-run walk) mix fast without the need for further rounding. Therefore, as an alternative, we could preprocess the convex body to be well-rounded and then apply the simulated annealing algorithm to obtain a volume estimation algorithm that uses $\tilde{O}(n^3 + n^{2.5}/\epsilon)$ quantum queries.

\paragraph{Error analysis of discretized hit-and-run walks (\sec{implement-hit-and-run})}
Although we defined quantum hit-and-run walks abstractly in \sec{continuous-q-walk-theory}, implementing a continuous-space quantum walk on a digital quantum computers will lead to discretization error, and the error analysis of classical walks in a discrete space approximating $\R^{n}$ (such as~\cite{frieze1999log}) does not automatically apply to the quantum counterpart. To ensure that discretization errors do not affect a realistic implementation of our algorithm, in \sec{implement-hit-and-run} we propose a \emph{discretized hit-and-run walk} and provide rigorous bounds on the discretization error.

\paragraph{Summary}
Our quantum volume estimation algorithm can be summarized as follows.
\begin{itemize}[leftmargin=*]
\item[1)] \textsf{High level:} The quantum algorithm follows a simulated annealing framework using a quantum MCMC method~\cite{wocjan2008speedup}, where the volume is estimated by a telescoping product (as in \eqn{telescoping-volume}); the number of iterations is $\tilde{O}(\sqrt{n})$.
\item[2)] \textsf{Middle level:} We estimate the $\E[V_{i}]$ in \eqn{V_i-techniques}, a ratio in the telescoping product, using the nondestructive version of the quantum Chebyshev inequality~\cite{hamoudi2019Chebyshev}. This takes $\tilde{O}(\sqrt{n}/\epsilon)$ implementations of the quantum hit-and-run walk operators.
\item[3)] \textsf{Low level:} If the convex body $\K$ is well-rounded (i.e., $R/r=O(\sqrt{n})$), each quantum hit-and-run walk operator can be implemented using $\tilde{O}(n^{1.5})$ queries to the membership oracle $O_{\K}$ in \eqn{oracle-defn}.
\end{itemize}

Finally, we give a quantum algorithm that interlaces rounding and volume estimation of the convex body, using an additional $\tilde{O}(n^{2.5})$ quantum queries to $O_{\K}$ in each iteration. Because the affine transformation is an $n$-dimensional matrix-vector product, it introduces an overhead of $O(n^{2})$ in the number of arithmetic operations (just as in the classical rounding algorithm).

Overall, our quantum volume estimation algorithm has $\tilde{O}(\sqrt{n})$ iterations. Each iteration implements $\tilde{O}(\sqrt{n}/\epsilon)$ quantum hit-and-run walks, and each quantum hit-and-run walk uses $\tilde{O}(n^{1.5})$ queries; there is also a cost of $\tilde{O}(n^{2.5})$ for rounding. Thus the quantum query complexity is
\begin{align}
\tilde{O}(\sqrt{n})\cdot\big(\tilde{O}(\sqrt{n}/\epsilon)\cdot\tilde{O}(n^{1.5})+\tilde{O}(n^{2.5})\big)=\tilde{O}(\querycost).
\end{align}
The number of additional arithmetic operations is $\tilde{O}(\querycost)\cdot O(n^{2})=\tilde{O}(\arithcost)$ due to the affine transformations for interlaced rounding of the convex body.

 \fig{flowchart} summarizes our techniques. The volume estimation and interlaced rounding algorithms are given as \algo{quantum-volume} and \algo{quantum-interlacing}, respectively, in \sec{quantum-volume}.

 \begin{figure}[htbp]
 \centering
 {\tikzset{%
   >={Latex[width=2mm,length=2mm]},
             base/.style = {rectangle, rounded corners, draw=black,
                            minimum width=4cm, minimum height=1cm,
                            text centered, font=\sffamily},
        result/.style = {base, fill=green!30},
        new/.style = {base, fill=blue!15},
          old/.style = {base, minimum width=2.5cm, fill=yellow!30},
 }
 \begin{tikzpicture}[node distance=2cm,
     every node/.style={fill=white, font=\sffamily, scale=0.7}, align=center]
   \node (SA-volume)     [result]          {Quantum volume estimation\\algorithm (\sec{quantum-volume-estimation})};
   \node (Continuous-q-walk)     [new, above = 0.8 cm of SA-volume]   {Continuous-space\\quantum walk (\sec{continuous-q-walk-theory})};
   \node (Discretized-q-walk)     [new, left = 2.5 cm of Continuous-q-walk]   {Discretized quantum\\hit-and-run walk (\sec{implement-hit-and-run})};
   \node (rounding)     [new, below left = 0.8 cm and 1.6 cm of SA-volume]   {Quantum convex body\\rounding algorithm (\sec{round-quantum})};
   \node (nondestructive)     [new, right = 1 cm of rounding]   {Chebyshev cooling via nondestructive\\mean estimation (\sec{nondestructive})};
   \node (fixed-Grover)     [old, right = 1 cm of nondestructive]   {Fixed-point amplitude\\amplification (\sec{quantum-MCMC})};

   \draw[->]              (Continuous-q-walk) -- node {implement} (SA-volume);
   \draw[->]              (Discretized-q-walk) -- node {implement} (Continuous-q-walk);
   \draw[-,dotted]              (SA-volume) -- node {Step 1} ($(rounding.north east)+(-0.75,0.2)$);
   \draw[-,dotted]              (SA-volume) -- node {Step 2 (simulated annealing)} ($(nondestructive.north east)+(0.5,0.2)$);
   \draw[->]              (rounding) -- (nondestructive);
   \draw[->,dashed]              ([yshift=0.2 cm]nondestructive.east) -- ([yshift=0.2 cm]fixed-Grover.west);
   \draw[->,dashed]              ([yshift=-0.2 cm]fixed-Grover.west) -- ([yshift=-0.2 cm]nondestructive.east);
   \draw[black,dotted] ($(rounding.north west)+(-0.2,0.2)$)  rectangle ($(rounding.south east)+(0.2,-0.2)$);
   \draw[black,dotted] ($(nondestructive.north west)+(-0.2,0.2)$)  rectangle ($(fixed-Grover.south east)+(0.2,-0.2)$);
 \end{tikzpicture}}
 \caption{The structure of our quantum volume estimation algorithm. The four purple frames represent the four novel techniques that we propose, the yellow frame represents the known technique from~\cite{grover2005different}, and the green frame at the center represents our quantum algorithm.}
 \label{fig:flowchart}
 \end{figure}
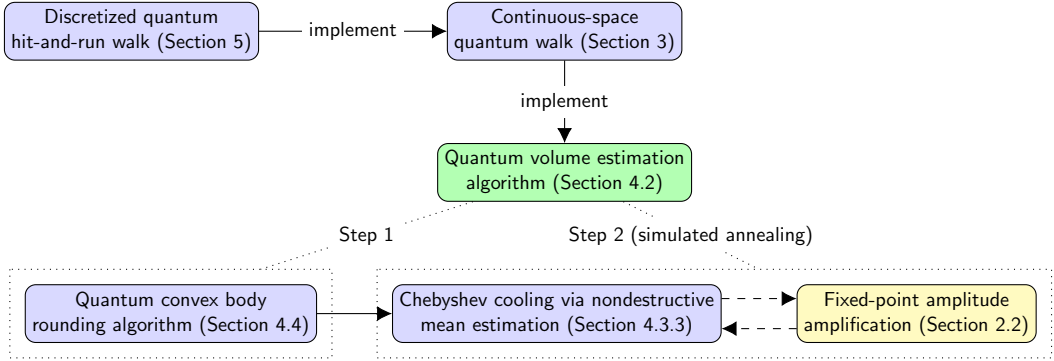

\subsubsection{Quantum lower bounds (\sec{quantum-lower})}\label{sec:lower-intro}
The classical state-of-the-art query lower bound for volume estimation is a  $\tilde{\Omega}(n^{2})$ bound for $n$-dimensional parallelopipeds~\cite{rademacher2008dispersion}.
The argument uses Yao's principle~\cite{yao1977probabilistic} to reduce the problem of estimating the volume of parallelopipeds to a corresponding average-case lower bound for deterministic algorithms.
However, in the quantum setting, it is unclear how to apply a similar argument since such a reduction to the deterministic case does not work in general.

Nevertheless, we prove that volume estimation requires $\Omega(\sqrt{n})$ quantum queries to the membership oracle, ruling out the possibility of exponential quantum speedup (see \thm{lower-bound}). We establish this by a reduction to search: for a hyper-rectangle $\K=\bigtimes_{i=1}^n [0,2^{s_i}]$ specified by a binary string $s=(s_1,\ldots,s_n)\in\{0,1\}^n$ with $|s|=0$ or $1$, we prove that a membership query to $\K$ can be simulated by a query to $s$. Thus, since $\vol(\K)=2$ if and only if $|s|=1$, the $\Omega(\sqrt{n})$ quantum lower bound on search~\cite{bennett1997strengths} applies to volume estimation.

In addition, we prove that volume estimation requires $\Omega(1/\epsilon)$ quantum queries (see \thm{lower-bound-eps}), which means that our quantum algorithm is optimal in $1/\epsilon$ up to poly-logarithmic factors. The idea is to construct a convex body whose volume estimation reduces to the Hamming distance problem with known tight quantum query complexity~\cite{nayak1999quantum}. To be more specific, we consider the $n$-dimensional unit hypercube and attach ``hyperpyramids" to its faces, such that its central axis passes through the center of the hypercube. We show that adding or deleting any hyperpyramid of volume $1/2n$ does not influence the convexity of the convex body, and calculating the volume of the body reveals the Hamming weight of a binary string that encodes the presence or absence of the hyperpyramids.

\subsection{Related work}\label{sec:related-work}
While our paper gives the first quantum algorithm for volume estimation, classical volume estimation algorithms have been well-studied, as we review in \sec{classical-volume}. Our quantum algorithm builds upon quantum analogs of Markov chain Monte Carlo methods that we review in \sec{quantum-MCMC-literature}.

\subsubsection{Classical volume estimation algorithms}\label{sec:classical-volume}
There is a rich literature on classical algorithms for estimating volumes of convex bodies (e.g., see the surveys~\cite{vempala2005geometric,lee2018KLS}). The general approach is to consider a sequence of random walks inside the convex body $\K$ whose stationary distributions converge quickly to the uniform distribution on $\K$. Applying simulated annealing to this sequence of walks (as in \sec{techniques}), the volume of $\K$ can be approximated by a telescoping product.

The first polynomial-time algorithm for volume estimation was given by~\cite{dyer1991random}. It uses a \emph{grid walk} in which the convex body $\K$ is approximated by a grid mesh $\K_{\textrm{grid}}$ of spacing $\delta$ (i.e., $\K_{\textrm{grid}}$ contains the points in $\K$ whose coordinates are integer multiples of $\delta$). The walk proceeds as follows:
\begin{enumerate}
\item Pick a grid point $y$ uniformly at random from the neighbors of the current point $x$.
\item If $y\in\K_{\textrm{grid}}$, go to $y$; else stay at $x$.
\end{enumerate}

Dyer, Frieze, and Vempala~\cite{dyer1991random} proved that for a properly chosen $\delta$, the grid walk converges to the uniform distribution on $\K_{\textrm{grid}}$ in $\tilde{O}(n^{23})$ steps, and that $\delta^{n}|\K_{\textrm{grid}}|$ is a good approximation of $\vol(\K)$ (in the sense of \eqn{volume-estimation-defn}). Subsequently, more refined analysis of the grid walk improved its cost to $\tilde{O}(n^{8})$~\cite{lovasz1990mixing,applegate1991sampling,dyer1991computing}. However, this is still inefficient in practice.

Intuitively, the grid walk converges slowly because each step only moves locally in $\K$. Subsequent work improved the complexity by considering other types of random walk. These improvements mainly use two types of walk: the \emph{hit-and-run walk} and the \emph{ball walk}. In this paper, we use the hit-and-run walk (see also \sec{hit-and-run}), which behaves as follows:
\begin{enumerate}
\item Pick a uniformly distributed random line $\ell$ through the current point $p$.
\item Move to a uniformly random point along the chord $\ell\cap\K$.
\end{enumerate}

Smith~\cite{smith1984efficient} proved that the stationary distribution of the hit-and-run walk is the uniform distribution on $\K$. Regarding the convergence of the hit-and-run walk,~\cite{Lovasz99} showed that it mixes in $\tilde{O}(n^3)$ steps from a warm start after appropriate preprocessing, and~\cite{LV06} subsequently proved that the hit-and-run walk mixes rapidly from any interior starting point (see also \thm{hit-and-run-LV06}). Under the simulated annealing framework, the hit-and-run walk gives the state-of-the-art volume estimation algorithm with query complexity $\tilde{O}(n^{4})$~\cite{lovasz2006fast,lovasz2006simulated}. Our quantum volume estimation algorithm can be viewed as a quantization of this classical hit-and-run algorithm.

Given a radius parameter $\delta$, the ball walk is defined as follows:
\begin{enumerate}
\item Pick a uniformly random point $y$ from the ball of radius $\delta$ centered at the current point $x$.
\item If $y\in\K$, go to $y$; else stay at $x$.
\end{enumerate}
Lov{\'a}sz and Simonovits~\cite{lovasz1993random} proved that the ball walk mixes in $\tilde{O}(n^6)$ steps. Kannan et al.~\cite{kannan1997random} subsequently improved the mixing time to $\tilde{O}(n^3)$ starting from a warm start, giving a total query complexity of $\tilde{O}(n^5)$ for the volume estimation problem.

The analysis of the ball walk relies on a central conjecture in convex geometry, the Kannan-Lov{\'a}sz-Simonovits (KLS) conjecture (see~\cite{lee2018KLS}). The KLS conjecture states that the Cheeger constant of any log-concave density is achieved to within a universal, dimension-independent constant factor by a hyperplane-induced subset, where the Cheeger constant is the minimum ratio between the measure of the boundary of a subset to the measure of the subset or its complement, whichever is smaller. Although this quantity is conjectured to be a constant, the best known upper bound is only $O(n^{1/4})$~\cite{lee2017eldan}. However, in the special case when the convex body is well-rounded (i.e., $R/r=O(\sqrt{n})$), a recent breakthrough by Cousins and Vempala~\cite{cousins2014cubic,cousins2015bypassing} proved the KLS conjecture for Gaussian distributions. In other words, they established a volume estimation algorithm with query complexity $\tilde{O}(n^{3})$ in the well-rounded case.

\tab{classical-volume-summary} summarizes classical algorithms for volume estimation.

\begin{table}[htbp]
\centering
\resizebox{0.7\columnwidth}{!}{%
\begin{tabular}{|c|c|c|}
\hline
Method & \specialcell{State-of-the-art\\query complexity} & Restriction on the convex body \\ \hline\hline
Grid walk & $\tilde{O}(n^{8})$~\cite{dyer1991computing} & General ($R/r=\poly(n)$) \\ \hline
Hit-and-run walk & $\tilde{O}(n^{4})$~\cite{lovasz2006fast,lovasz2006simulated} & General ($R/r=\poly(n)$) \\ \hline
Ball walk & $\tilde{O}(n^{3})$~\cite{cousins2014cubic,cousins2015bypassing} & Well-rounded ($R/r=O(\sqrt{n})$) \\ \hline
\end{tabular}
}
\caption{Summary of classical methods for estimating the volume of a convex body $\K \subset \R^n$ when $\epsilon=\Theta(1)$, where $R,r$ are the radii of the balls centered at the origin that contain and are contained by the convex body, respectively.}
\label{tab:classical-volume-summary}
\end{table}

\subsubsection{Quantum Markov chain Monte Carlo methods}\label{sec:quantum-MCMC-literature}
The performance of Markov chain Monte Carlo (MCMC) methods is determined by the rate of convergence to their stationary distributions (i.e., the mixing time). Suppose we have a reversible, ergodic Markov chain with unique stationary distribution $\pi$. Let $\pi_{k}$ denote the distribution obtained by applying the Markov chain for $k$ steps from some arbitrary initial state. It is well-known (see for example~\cite{levin2017markov}) that $O(\frac{1}{\Delta}\log(1/(\epsilon\min_{x}\pi(x))))$ steps suffice to ensure $\|\pi_{k}-\pi\|_{1}\leq\epsilon$, where $\Delta$ is the spectral gap of the Markov chain.

Many authors have studied quantum analogs of Markov chains (in both continuous \cite{FG98} and discrete \cite{ABNVW01,aharonov2001quantum,szegedy2004quantum} time) and their mixing properties. While a quantum walk is a unitary process and hence does not converge to a stationary distribution, one can define notions of quantum mixing time by choosing the number of steps at random or by adding decoherence \cite{ABNVW01,aharonov2001quantum,CCDFGS03,AR05,richter2007almost,Richter2007,chakraborty2019analog}, and compare them to the classical mixing time. Note that distribution sampled by such a process may or may not be the same as the stationary distribution $\pi$ of the corresponding classical Markov process, depending on the structure of the process and the notion of mixing. It is also natural to ask how efficiently we can prepare a quantum state close to $|\pi\>:=\sum_{x}\sqrt{\pi_x}|x\>$, which can be viewed as a ``quantum sample'' from $\pi$. However, it is unclear how to do this efficiently in general, even in cases where a corresponding classical Markov process mixes quickly; in particular, a generic quantum algorithm for this task could be used to solve graph isomorphism \cite[Section 8.4]{aharonov2003adiabatic}.

It is also possible to achieve quantum speedup of MCMC methods by not demanding speedup of the mixing time of each separate Markov chain, but only for the procedure as a whole. In particular, MCMC methods are often implemented by simulated annealing algorithms where the final output is a telescoping product of values at different temperatures. From this perspective, Somma et al.~\cite{somma2007quantum,somma2008quantum,boixo2015quantum} used quantum walks to accelerate classical simulated annealing processes by exploiting the quantum Zeno effect, using measurements implemented by phase estimation of the quantum walk operators of these Markov chains. References~\cite{temme2011quantum,yung2012quantum} also introduced how to implement Metropolis sampling on quantum computers.

Our quantum volume estimation algorithm is most closely related to work of Wocjan and Abeyesinghe~\cite{wocjan2008speedup}, which achieves complexity $\tilde{O}(1/\sqrt{\Delta})$ for preparing the final stationary distribution of a sequence of slowly varying Markov chains, where $\Delta$ is the minimum of their spectral gaps. Their quantum algorithm transits between the stationary states of consecutive Markov chains by fixed-point amplitude amplification~\cite{grover2005different}, which is implemented by amplitude estimation with $\tilde{O}(1/\sqrt{\Delta})$ implementations of the quantum walk operators of these Markov chains (see \sec{quantum-MCMC} for more details).

Our simulated annealing procedure preserves the slowly-varying property, so we adopt the framework of~\cite{wocjan2008speedup} in our algorithm for volume estimation (see \sec{proof-inner-product}). We develop several novel techniques (described in \sec{techniques}) that allow us to implement the steps of this framework efficiently. Note that the slowly-varying property also facilitates other frameworks that give efficient adiabatic~\cite{aharonov2003adiabatic} or circuit-based~\cite{orsucci2018faster} quantum algorithms for generating quantum samples of the stationary state.

Previous work has mainly applied these quantum simulated annealing algorithms to estimating partition functions of discrete systems. Given an inverse temperature $\beta>0$ and a classical Hamiltonian $H\colon\Omega\to\R$ where $\Omega$ is a finite space, the goal is to estimate the partition function
\begin{align}\label{eqn:partition-defn}
Z(\beta):=\sum_{x\in\Omega}e^{-\beta H(x)}
\end{align}
within multiplicative error $\epsilon>0$. Wocjan et al.~\cite{wocjan2009quantum} gave a quantum algorithm that achieves quadratic quantum speedup with respect to both mixing time and accuracy.

The classical algorithm that~\cite{wocjan2009quantum} quantizes uses $\tilde{O}(\log |\Omega|)$ annealing steps to ensure that each ratio $Z(\beta_{i+1})/Z(\beta_{i})$ is bounded. In fact, it is possible to relax this requirement and use a cooling schedule with only $\tilde{O}(\sqrt{\log |\Omega|})$ steps such that the variance of each ratio is bounded, so its mean can be well-approximated by Chebyshev's inequality; this is exactly the Chebyshev cooling technique~\cite{SVV009} introduced in \sec{techniques} (see also \sec{nondestructive}). Montanaro~\cite{montanaro2015quantum} improves upon~\cite{wocjan2009quantum} using Chebyshev cooling; more recently, Harrow and Wei~\cite{harrow2019adaptive} further quadratically improved the spectral gap dependence of the estimation of the partition function.

\subsection{Open questions}\label{sec:open-question}
This work leaves several natural open questions for future investigation. In particular:
\begin{itemize}
\item Can we improve the complexity of our quantum volume estimation algorithm? The gap between the current upper and lower bounds in $n$ is large; possible improvements might result from designing a shorter simulated annealing schedule, giving better analysis of the conductance of the hit-and-run walk, or even using other types of walks.

\item Can we prove better quantum query lower bounds on volume estimation? Note that classically there is an $\tilde{\Omega}(n^{2})$ query lower bound~\cite{rademacher2008dispersion}.

\item Can we give faster quantum algorithms for volume estimation in some special circumstances? For instance, volume estimation of well-rounded convex bodies only takes $\tilde{O}(n^{3})$ classical queries~\cite{cousins2015bypassing} (see also \sec{classical-volume}), and the volume of polytopes with $m$ faces can be estimated with only $\tilde{O}(mn^{2/3})$ classical queries~\cite{lee2018convergence}. Specifically, it is a natural question to ask whether the ball walk in~\cite{cousins2015bypassing} or the Riemannian Hamiltonian Monte Carlo (RHMC) method in~\cite{lee2018convergence} can be implemented by continuous-space quantum walks (and their discretizations).

\item Can we apply our simulated annealing framework to solve other problems? As a concrete example, it may be of interest to check whether our framework can recover the results of Ref.~\cite{harrow2019adaptive} on estimating the partition functions in counting problems.

\item After we have finished this paper, we became aware of a recent work~\cite{jia2020reducing} that gives a classical algorithm for rounding convex bodies with $\tilde{O}(n^{3.5})$ queries. Combined with~\cite{cousins2015bypassing}, this results in a classical algorithm for $n$-dimensional volume estimation with multiplicative error $\epsilon$ using $\tilde{O}(n^{3.5}+n^{3}/\epsilon^{2})$ queries. While our quantum algorithm still improves over this complexity, it is natural to ask whether their algorithm fits into the interlaced structure of \algo{quantum-interlacing} and achieves even better quantum query/gate complexities.
\end{itemize}

\paragraph{Organization}
We review necessary background in \sec{prelim}. We describe the theory of continuous-space quantum walks in \sec{continuous-q-walk-theory}. In \sec{quantum-volume}, we first review the classical state-of-the-art volume estimation algorithm in \sec{volume-estimation-review}, and then give our quantum algorithm for estimating volumes of well-rounded convex bodies in \sec{quantum-volume-estimation}. The proofs of our quantum algorithms are given in \sec{quantum-volume-estimation-proof}, and the quantum algorithm for rounding convex bodies is given in \sec{round-quantum}. The details of our discretized hit-and-run walk are given in \sec{implement-hit-and-run}, and we conclude with our quantum lower bound on volume estimation in \sec{quantum-lower}.


\section{Preliminaries}\label{sec:prelim}
We summarize necessary tools used in this paper as follows.
\subsection{Classical and quantum walks}\label{sec:quantum-walk-prelim}
A Markov chain over a finite state space $\Omega$ is a sequence of random variables $X_0,X_1,\ldots$ such that
for each $i\in\mathbb{N}$, the probability of transition to the next state $y \in \Omega$,
\begin{align*}
\Pr[X_{i+1}=y\mid X_i=x,X_{i-1}=x_{i-1},\ldots,X_0=x_0] = \Pr[X_{i+1}=y\mid X_i=x]=:p_{x\to y}
\end{align*}
only depends on the present state $x \in \Omega$.
The Markov chain can be represented by the transition probabilities $p_{x\to y}$ satisfying $\sum_y p_{x\to y}=1$.
For each $i\in\mathbb{N}$, we denote by $\pi_i$ the distribution over $\Omega$ with density $\pi_i(x)=\Pr[X_i=x]$.
A \emph{stationary distribution} $\pi$ satisfies $\sum_{x\in\Omega} p_{x\to y}\pi(x)=\pi(y)$.
A Markov chain is \emph{reversible} if it has a stationary distribution $\pi$ such that $\pi(x)p_{x\to y}=\pi(y)p_{y\to x}$ for all $x,y\in\Omega$. The \emph{conductance} of a reversible Markov chain is defined as
\begin{align}\label{eqn:conductance-defn-discrete}
\Phi:=\inf_{\S\subseteq\Omega}\frac{\sum_{x\in\S}\sum_{y\in\Omega/\S}\pi(x)p_{x\to y}}{\min\{\sum_{x\in\S}\pi(x),\sum_{x\in\Omega/\S}\pi(x)\}}.
\end{align}

The theory of discrete-time quantum walks has also been well developed.
Given a classical reversible Markov chain on $\Omega$ with transition probability $p$, we define a unitary operator $U_{p}$ on $\C^{|\Omega|}\otimes\C^{|\Omega|}$ such that
\begin{align}
U_{p}|x\>|0\>=|x\>|p_{x}\>,\text{ where }|p_{x}\>:=\sum_{y\in\Omega}\sqrt{p_{x\to y}}|y\>.
\end{align}
The quantum walk is then defined as \cite{szegedy2004quantum}
\begin{align}
W_{p}:=S\big(2U_{p}(I_{\Omega}\otimes|0\>\<0|)U_{p}^{\dagger}-I_{\Omega}\otimes I_{\Omega}\big),
\end{align}
where $I_{\Omega}$ is the identity map on $\C^{|\Omega|}$ and $S:=\sum_{x,y\in\Omega}|x,y\>\<y,x|=S^\dag$ is the swap gate on $\C^{|\Omega|}\otimes\C^{|\Omega|}$.

To understand the quantum walk, it is essential to analyze the spectrum of $W_p$.
First, observing that $W_p = S(2\Pi-I)$ where
$\Pi=U_p(I_\Omega\otimes\ket{0}\bra{0})U_p^\dag=\sum_{x\in\Omega}\ket{x}\bra{x}\otimes\ket{p_x}\bra{p_x}$ projects onto the span of the states $\ket{x}\otimes\ket{p_x}$,
we consider the eigenvector $\ket{\lambda}$ of $\Pi S\Pi$ with eigenvalue $\lambda$.
We have $\Pi S\Pi=\sum_{x\in\Omega}D_{xy}\ket{x}\bra{y}\otimes\ket{p_x}\bra{p_y}$ where $D_{xy}:=\sqrt{p_{x\to y}p_{y\to x}}$.
Since $W_p\ket{\lambda}=S\ket{\lambda}$ and $W_pS\ket{\lambda}=2\lambda S\ket{\lambda}-\ket{\lambda}$,
the subspace $\text{span}\{\ket{\lambda},S\ket{\lambda}\}$ is invariant under $W_p$.
The eigenvalues of $W_p$ within this subspace are $\lambda\pm i\sqrt{1-\lambda^2}=e^{\pm i\arccos\lambda}$. For more details, see~\cite{szegedy2004quantum}.

The phase gap $\arccos\lambda\geq\sqrt{2(1-\lambda)}\geq\sqrt{2\delta}$, where $\delta$ is the spectral gap of $D$. Therefore, applying phase estimation using $O(1/\sqrt{\delta})$ calls to $W_p$ suffices to distinguish the state corresponding to the stationary distribution of the classical Markov chain from the other eigenvectors.

\subsection{Quantum speedup of MCMC sampling via simulated annealing}\label{sec:quantum-MCMC}

Consider a Markov chain with spectral gap $\Delta$ and stationary distribution $\pi$. Classically, it takes $\Theta(\frac{1}{\Delta}\log(1/\epsilon\pi_{\min})))$ steps to sample from a distribution $\tilde{\pi}$ such that $\|\tilde{\pi}-\pi\|\leq\epsilon$, where $\pi_{\min}:=\min_i \pi_i$. Quantumly, \cite{wocjan2008speedup} proved the following result about a sequence of slowly varying Markov chains:

\begin{theorem}[{\cite[Theorem 2]{wocjan2008speedup}}]\label{thm:Wocjan-Abeyesinghe}
Let $p_{1},\ldots,p_{r}$ be the transition probabilities of $r$ Markov chains with stationary distributions $\pi_{1},\ldots,\pi_{r}$, spectral gaps  $\delta_{1},\ldots,\delta_{r}$, and quantum walk operators $W_{1},\ldots,\allowbreak W_{r}$, respectively; let $\Delta:=\min\{\delta_{1},\ldots,\delta_{r}\}$. Assume that $|\<\pi_{i}|\pi_{i+1}\>|^{2}\geq p$ for some $0<p<1$ and all $i\in\range{r-1}$, and assume that we can efficiently prepare the state $|\pi_{1}\>$ (where each $|\pi_i\>$ is a quantum sample defined as in \sec{quantum-MCMC-literature}). Then, for any $0<\epsilon<1$, there is a quantum algorithm that produces a quantum state $|\tilde{\pi}_{r}\>$ such that $\||\tilde{\pi}_{r}\>-|\pi_{r}\>\|\leq\epsilon$, using $\tilde{O}(r/(p\sqrt{\Delta}))$ steps of the quantum walk operators $W_{1},\ldots,W_{r}$, where the $\tilde{O}$ omits poly-logarithmic terms in $r$, $1/\epsilon$, and $1/p\sqrt{\Delta}$.\footnote{Note that this is quadratically worse in $1/p$ than the Grover's algorithm~\cite{grover1997quantum} with complexity $O(1/\sqrt{p})$. This is because we use a simple fixed-point quantum search algorithm~\cite{grover2005different} that does not require knowing $p$ in advance. Notice that there exist fixed-point quantum search algorithms that preserve the $O(1/\sqrt{p})$ speedup (e.g.,~\cite{YLC2014}, \cite[Chapter 6]{Wang2018}), but in our quantum algorithm, the simpler algorithm suffices as $p = \Theta(1)$ (see \lem{inner-product}).}
\end{theorem}

Their quantum algorithm produces the states $|\pi_{1}\>,\ldots,|\pi_{r}\>$ sequentially, and can do so rapidly if consecutive states have significant overlap and the walks mix rapidly. Intuitively, this is achieved by amplitude amplification. However, to avoid overshooting, the paper uses a variant of standard amplitude amplification, known as \emph{$\pi/3$-amplitude amplification}~\cite{grover2005different}, that we now review.

Given two states $|\psi\>$ and $|\phi\>$, we let $\Pi_{\psi}:=|\psi\>\<\psi|$, $\Pi_{\psi}^{\perp}:=I- \Pi_{\psi}$, $\Pi_{\phi}:=|\phi\>\<\phi|$, and $\Pi_{\phi}^{\perp}:=I- \Pi_{\phi}$. Define the unitaries
\begin{align}\label{eqn:pi/3-AM-AM-Ri}
R_{\psi}:=\omega\Pi_{\psi} + \Pi_{\psi}^{\perp},\quad R_{\phi}:=\omega\Pi_{\phi} + \Pi_{\phi}^{\perp}\qquad\text{where}\quad\omega=e^{i\frac{\pi}{3}}.
\end{align}
Given $|\<\psi|\phi\>|^{2}\geq p$, it can be shown that $|\<\phi|R_{\psi}R_{\phi}|\psi\>|^{2}\geq 1-(1-p)^{3}$. Recursively, one can establish the following:

\begin{lemma}[{\cite[Lemma 1]{wocjan2008speedup}}]\label{lem:pi3-amplification}
Let $|\psi\>$ and $|\phi\>$ be two quantum states with $|\<\psi|\phi\>|^2 \ge p$ for some $0 < p \leq 1$. Define the unitaries $R_{\psi},R_{\phi}$ as in \eqn{pi/3-AM-AM-Ri} and the unitaries $U_{m}$ recursively as follows:
\begin{align}\label{eqn:pi/3-Ui}
    U_{0} = I,\qquad U_{m+1} = U_{m}\, R_{\psi} \, U_{m}^\dagger \, R_{\phi} \, U_{m}.
\end{align}
Then we have
 \begin{align}\label{eqn:pi/3-AM-AM}
    |\<\phi|U_{m}|\psi\>|^2 \ge 1 - (1-p)^{3^m},
 \end{align}
and the unitaries in $\{R_{\psi}, R_{\psi}^\dagger, R_{\phi}, R_{\phi}^\dagger\}$ are used at most $3^m$ times in $U_{m}$.
\end{lemma}

Taking $m=\lceil\log_{3}(\ln(1/\epsilon)/p)\rceil$, the inner product between $|\phi\>$ and $U_{m}|\psi\>$ in \eqn{pi/3-AM-AM} is at least $1-\epsilon$, and we use $3^{m}=O(\log(1/\epsilon)/p)$ unitaries from the set $\{R_{\psi}, R_{\psi}^\dagger, R_{\phi}, R_{\phi}^\dagger\}$.

To establish \thm{Wocjan-Abeyesinghe} by \lem{pi3-amplification}, it remains to construct the unitaries $R_{i}:=\omega|\pi_i\>\<\pi_i| + (I - |\pi_i\>\<\pi_i|)$. In~\cite{wocjan2008speedup}, this is achieved by phase estimation of the quantum walk operator $W_{i}$ with precision $\sqrt{\Delta}/2$. Recall that if a classical Markov chain has spectral gap $\delta$, then the corresponding quantum walk operator has phase gap of at least $2\sqrt{\delta}$ (see \sec{quantum-walk-prelim}). Therefore, phase estimation with precision $\sqrt{\Delta}/2$ suffices to distinguish between $|\pi_{i}\>$ and other eigenvectors of $W_{i}$. As a result, we can take
\begin{align}\label{eqn:Wocjan-walk-PhaseEst}
R_i=\textsf{PhaseEst}(W_{i})^{\dagger}\bigl(I\otimes\big(\omega|0\>\<0|+(I-|0\>\<0|)\big)\bigr)\textsf{PhaseEst}(W_{i}).
\end{align}

\subsection{Quantum Chebyshev inequality}\label{sec:quantum-Chebyshev}

Assume we are given a unitary $U$ such that
\begin{align}\label{eq:AmpEst-def}
U|0\>|0\>=\sqrt{p}|0\>|\phi\>+|0^{\perp}\>,
\end{align}
where $|\phi\>$ is a normalized pure state and $(\<0|\otimes I)|0^{\perp}\>=0$. If we measure the output state, we get $0$ in the first register with probability $p$; by the Chernoff bound, it takes $\Theta(1/\epsilon^{2})$ samples to estimate $p$ within $\epsilon$ with high success probability. However, there is a more efficient quantum algorithm, called \emph{amplitude estimation}~\cite{brassard2002amplitude}, that estimates the value of $p$ using only $O(1/\epsilon)$ calls to $U$:

\begin{theorem}[{\cite[Theorem 12]{brassard2002amplitude}}]\label{thm:AmpEst}
Given $U$ satisfying \eq{AmpEst-def}, the amplitude estimation algorithm in \fig{AmpEst} outputs an angle $\tilde{\theta}_{p}\in [-\pi,\pi]$ such that $\tilde{p}:=\sin^{2}(\tilde{\theta}_{p})$ satisfies
\begin{align}\label{eq:AmpEst-1}
|\tilde{p}-p|\leq \frac{2\pi\sqrt{p(1-p)}}{M}+\frac{\pi^{2}}{M^{2}}
\end{align}
with success probability at least $8/\pi^{2}$, using $M$ calls to $U$ and $U^{\dagger}$.
\end{theorem}

\vspace{-6mm}
\begin{figure}[H]
\begin{align*}
\Qcircuit @C=1.2em @R=1.2em {
   |0\> & & \multigate{2}{\textsf{QFT}} & \ctrl{0} {\ar @{-}+<0em,-0.8em>} & \multigate{2}{\textsf{QFT}^{\dagger}} & \qw \\
   \raisebox{6pt}{\vdots} & & & \raisebox{6pt}{\vdots} & & & |\tilde{\theta}_{p}\> \\
   |0\> & & \ghost{\textsf{QFT}} & \ctrl{1}{\ar @{-}+<0em,0.8em>} & \ghost{\textsf{QFT}^{\dagger}} & \qw \\
   & & \qw & \multigate{2}{\mathcal{Q}} & \qw & \qw \\
   U|0\> & \ \ \ \raisebox{6pt}{\vdots} & & & & \raisebox{6pt}{\vdots} \\
   & & \qw & \ghost{\mathcal{Q}} & \qw & \qw
}
\end{align*}
\caption{The quantum circuit for amplitude estimation.}
\label{fig:AmpEst}
\end{figure}
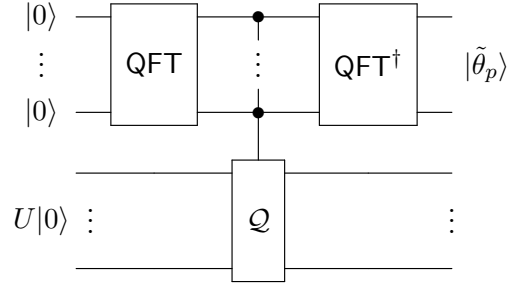

\noindent
Here $\textsf{QFT}$ denotes the quantum Fourier transform over $\Z_M$ and $\mathcal{Q}:=-US_{0}U^{\dagger}S_{1}$ where $S_{0}$ and $S_{1}$ are reflections about $|0\>$ and the target state, respectively, following the pattern of Grover search. The controlled-$\mathcal{Q}$ gate denotes the operation $\sum_{j=0}^{M-1}|j\>\<j|\otimes \mathcal{Q}^{j}$. In fact, it was shown in the proof of \cite[Theorem 12]{brassard2002amplitude} that the state after applying the circuit in \fig{AmpEst} is
\begin{align}\label{eqn:AmpEst-two-angles}
  \frac{e^{i\theta_{p}}}{\sqrt{2}}|\tilde{\theta}_{p}\>|\Psi_{+}\>-\frac{e^{-i\theta_{p}}}{\sqrt{2}}|-\tilde{\theta}_{p}\>|\Psi_{-}\>
\end{align}
where $\theta_{p}\in [0,\pi]$ such that $p=\sin^{2}(\theta_{p})$, and $\tilde{\theta}_{p}\in [0,\pi]$ is a random variable such that $\tilde{p}=\sin^{2}(\tilde{\theta}_{p})$, and $|\Psi_{\pm}\>$ are two eigenvectors of $\mathcal{Q}$. Measuring the first register either gives $\tilde{\theta}_{p}$ or $-\tilde{\theta}_{p}$ with probability $1/2$, but since $\sin^{2}(\tilde{\theta}_{p})=\sin^{2}(-\tilde{\theta}_{p})=\tilde{p}$, this does not influence the success of \thm{AmpEst}.

In \eq{AmpEst-1}, if we take $M=\bigl\lceil 2\pi\bigl(\frac{2\sqrt{p}}{\epsilon}+\frac{1}{\sqrt{\epsilon}}\bigr)\bigr\rceil=O(1/\epsilon)$, we get
\begin{align}
|\tilde{p}-p|\leq\frac{2\pi\sqrt{p(1-p)}}{2\pi}\epsilon+\frac{\pi^{2}}{4\pi^{2}}\epsilon^{2}\leq\frac{\epsilon}{2}+\frac{\epsilon}{4}\leq\epsilon.
\end{align}
Furthermore, the success probability $8/\pi^{2}$ can be boosted to $1-\nu$ by executing the algorithm $\Theta(\log 1/\nu)$ times and taking the median of the estimates.

Amplitude estimation can be generalized from estimating a single probability $p\in[0,1]$ to estimating the expectation of a random variable. Assume that $U$ is a unitary acting on $\C^{S}\otimes\C^{|\Omega|}$ such that
\begin{align}\label{eqn:quantum-sampler-defn}
U|0\>|0\>=\sum_{x\in\Omega}\sqrt{p_{x}}|\psi_x\>|x\>
\end{align}
where $S\in\N$ and $\{|\psi_x\> : x \in \Omega\}$ are unit vectors in $\C^{S}$. Let
\begin{align}
\mu_{U}:=\sum_{x\in\Omega}p_{x}x,\qquad\sigma_{U}^{2}:=\sum_{x\in\Omega}p_{x}(x-\mu_{U})^{2}
\end{align}
denote the expectation and variance of the random variable, respectively. Several quantum algorithms have given speedups for estimating $\mu_{U}$. Specifically, Ref.~\cite{montanaro2015quantum} showed how to estimate $\mu_{U}$ within additive error $\epsilon$ by $\tilde{O}(\sigma_{U}/\epsilon)$ calls to $U$ and $U^{\dagger}$. Given an upper bound $H$ and a lower bound $L>0$ on the random variable, Ref.~\cite{li2019entropy} showed how to estimate $\mu_{U}$ with multiplicative error $\epsilon$ using $\tilde{O}(\sigma_{U}/\epsilon\mu_{U}\cdot H/L)$ calls to $U$ and $U^{\dagger}$. More recently, Ref.~\cite{hamoudi2019Chebyshev} mutually generalized these results and proposed a significantly better quantum algorithm:
\begin{theorem}[{\cite[Theorem 3.5]{hamoudi2019Chebyshev}}]\label{thm:quantum-Chebyshev}
There is a quantum algorithm that, given a quantum sampler $U$ as in \eqn{quantum-sampler-defn}, an integer $\Delta_{U}$, a value $H>0$, and two reals $\epsilon,\delta \in (0,1)$, outputs an estimate $\tilde{\mu}_{U}$. If $\Delta_{U}\geq\sqrt{\sigma_{U}^{2}+\mu_{U}^{2}}/\mu_{U}$ and $H>\mu_{U}$, then $|\tilde{\mu}_{U}-\mu_{U}|\leq\epsilon\mu_{U}$ with probability at least $1-\delta$, and the algorithm uses $\tilde{O}(\Delta_{U}/\epsilon\cdot\log^{3}(H/\mu_{U})\log(1/\delta))$ calls to $U$ and $U^{\dagger}$.
\end{theorem}

The quantum algorithm works as follows. First, assume $\Omega\subseteq[L,H]$ for given real numbers $L,H\geq 0$, there is a basic estimation algorithm (denoted \textsf{BasicEst}) that estimates $H^{-1}\mu_{U}$ up to $\epsilon$-multiplicative error:

\begin{algorithm}[htbp]
\KwInput{A quantum sampler $U$ acting on $\C^{S}\otimes\C^{|\Omega|}$, interval $[L,H]$, precision parameter $\epsilon \in (0,1)$, failure parameter $\delta \in (0,1)$.}
\KwOutput{$\epsilon$-multiplicative approximation of $H^{-1}\mu_{U}$.}
Use controlled rotation to implement a unitary $R_{L,H}$ acting on $\C^{|\Omega|}\otimes\C^{2}$ such that for all $x\in\Omega$, $R_{L,H}|x\>|0\>=
    \begin{cases}
      |x\>(\sqrt{1-\frac{x}{H}}|0\>+\sqrt{\frac{x}{H}}|1\>) & \text{if } L\leq x<H \\
      |x\>|0\> & \text{otherwise}
    \end{cases}$\;
Let $V=(I_{S}\otimes R_{L,H})(U\otimes I_{2})$ and $\Pi=I_{S}\otimes I_{\Omega}\otimes |1\>\<1|$\;
\For{$i=1,\ldots,\Theta(\log(1/\delta))$}{Compute $\tilde{p}_{i}$ by \thm{AmpEst} with $U\leftarrow V$, $S_{1}\leftarrow 2\Pi-I$, and $M\leftarrow\Theta(1/(\epsilon\sqrt{H^{-1}\mu_{U}}))$\;}
Return $\tilde{p}=\median\{\tilde{p}_{1},\ldots,\tilde{p}_{\Theta(\log(1/\delta))}\}$.
\caption{\textsf{BasicEst:} the basic estimation algorithm.}
\label{algo:BasicEst}
\end{algorithm}

However, usually the bounds $L$ and $H$ are not explicitly given. In this case, Ref.~\cite{hamoudi2019Chebyshev} considered the truncated mean $\mu_{<b}$ defined by replacing the outcomes larger than $b$ with 0. The paper then runs \algo{BasicEst} (\textsf{BasicEst}) to estimate $\mu_{<b}/b$. A crucial observation is that $\sqrt{b/\mu_{<b}}$ is smaller than $\Delta_{U}$ for large values of $b$, and it becomes larger than $\Delta_{U}$ when $b\approx\mu_{U}\Delta_{U}^{2}$. As a result, by repeatedly running \textsf{BasicEst} with  $\Delta_{U}$ quantum samples, and applying $O(\log(H/L))$ steps of a binary search on the values of $b$, the first non-zero value is obtained when $b/\Delta_{U}^{2}\approx\mu_{U}$. In~\cite{hamoudi2019Chebyshev}, more precise truncation means are used to improve the precision of the result to $\tilde{O}(1/\epsilon)$ and remove the dependence on $L$.

Note that the quantum algorithm for \thm{quantum-Chebyshev} only relies on \textsf{BasicEst}. This is crucial when we estimate the mean of our simulated annealing algorithm in different iterations \emph{nondestructively} (see \sec{quantum-volume-estimation} for more details).

\subsection{Hit-and-run walk}\label{sec:hit-and-run}
As introduced in \sec{classical-volume}, there are various random walks that mix fast in a convex body $\K$, such as the grid walk~\cite{dyer1991random} and the ball walk~\cite{lovasz1993random,cousins2015bypassing}. In this paper, we mainly use the \emph{hit-and-run walk}~\cite{smith1984efficient,Lovasz99,LV06}. It is defined as follows:
\begin{enumerate}
\item Pick a uniformly distributed random line $\ell$ through the current point $p$.
\item Move to a uniform random point along the chord $\ell\cap\K$.
\end{enumerate}
For any two points $p,q\in\K$, we let $\ell(p,q)$ denote the length of the chord in $\K$ through $p$ and $q$. Then the transition probability of the hit-and-run walk is determined by the following lemma:
\begin{lemma}[{\cite[Lemma 3]{Lovasz99}}]\label{lem:Lovasz99-transition-prob}
If the current point of the hit-and-run walk is $u$, then the density function of the distribution of the next point $x \in \rmk$ is
\begin{align}
p_{u}(x)=\frac{2}{nv_{n}}\cdot\frac{1}{\ell(u,x)|x-u|^{n-1}},
\end{align}
where $v_{n}:=\pi^{\frac{n}{2}}/\Gamma(1+\frac{n}{2})$ is the volume of the $n$-dimensional unit ball. In other words, the probability that the next point is in a (measurable) set $\rma\subseteq\rmk$ is
\begin{align}
P_{u}(\rma)=\int_{\rma}\frac{2}{nv_{n}}\cdot\frac{1}{\ell(u,x)|x-u|^{n-1}}\,\d x.
\end{align}
\end{lemma}

In general, we can also define a hit-and-run walk with a given density. Let $f$ be a density function in $\R^{n}$. For any points $u,v\in\R^{n}$, we let
\begin{align}\label{eqn:mu-f-defn-continuous}
\mu_f(u, v):=\int_{0}^{1}f((1-t)u + tv)\,\d t.
\end{align}
For any line $\ell$, let $\ell^{+}$ and $\ell^{-}$ be the endpoints of the chord $\ell\cap\K$ (with $+$ and $-$ assigned arbitrarily). The density $f$ specifies the following hit-and-run walk:
\begin{enumerate}
\item Pick a uniformly distributed random line $\ell$ through the current point $p$.
\item Move to a random point $x$ along the chord $\ell\cap\K$ with density $\frac{f(x)}{\mu_f(\ell^{-},\ell^{+})}$.
\end{enumerate}

Let $\pi_{\K}$ denote the uniform distribution over $\K$. Smith~\cite{smith1984efficient} proved that the stationary distribution of the hit-and-run walk with uniform density is $\pi_{\K}$. Furthermore, Lov{\'{a}}sz and Vempala~\cite{LV06} proved that the hit-and-run walk mixes rapidly from any initial distribution:
\begin{theorem}[{\cite[Theorem 1.1]{LV06}}]\label{thm:hit-and-run-LV06}
Let $\K$ be a convex body that satisfies \eqn{ball-relationship}: $B_2(0,r) \subseteq \rmk \subseteq B_2(0,R)$. Let $\sigma$ be a starting distribution and let $\sigma^{(m)}$ be the distribution of the current point after $m$ steps of the hit-and-run walk in $\K$. Let $\epsilon>0$, and suppose that the density function $\d\sigma/\d\pi_{\K}$ is upper bounded by $M$ except on a set $\rms$ with $\sigma(\rms)\leq\epsilon/2$. Then for any
\begin{align}
m>10^{10}\frac{n^{2}R^{2}}{r^{2}}\ln\frac{M}{\epsilon},
\end{align}
the total variation distance between $\sigma^{(m)}$ and $\pi_{\K}$ is less than $\epsilon$.
\end{theorem}

\thm{hit-and-run-LV06} can also be generalized to exponential distributions on $\K$:
\begin{theorem}[{\cite[Theorem 1.3]{LV06}}]
  \label{thm:exp-hit-run-mixes}
    Let $\K \subset \R^n$ be a convex body and let $f$ be a density supported on $\K$ that is proportional to $e^{-a^T x}$ for some vector $a \in \R^n$. Assume that the level set of $f$ of probability $1/8$ contains a ball of radius $r$, and $\E_{f}(|x - z_f|^2) \le R^2$, where $z_f$ is the centroid of $f$. Let $\sigma$ be a starting distribution and let $\sigma^m$ be the distribution for the current point after $m$ steps of the hit-and-run walk applied to $f$. Let $\epsilon > 0$, and suppose that the density function $\frac{\d\sigma}{\d\pi_f}$ is upper bounded by $M$ except on a set $S$ with $\sigma(S) \le \frac{\epsilon}{2}$. Then for
    \begin{align*}
      m > 10^{30}\frac{n^2R^2}{r^2}\ln^5{\frac{MnR}{r\epsilon}},
    \end{align*}
    the total variation distance between $\sigma^m$ and $\pi_f$ is less than $\epsilon$.
  \end{theorem}

Roughly speaking, the proofs of \thm{hit-and-run-LV06} and \thm{exp-hit-run-mixes} have two steps. First, for any random walk on a continuous domain $\Omega$ with transition probability $p$ and stationary distribution $\pi$, we define its conductance (which generalizes the discrete case in Eq. \eqn{conductance-defn-discrete}) as
\begin{align}\label{eqn:conductance-defn-hit-and-run}
\Phi:=\inf_{\S\subseteq\Omega}\frac{\int_{\S}\int_{\Omega/\S}\d x\,\d y\,\pi_{x}p_{x\to y}}{\min\{\int_{\S}\d x\,\pi_{x},\int_{\Omega/\S}\d x\,\pi_{x}\}}.
\end{align}
It is well-known that the mixing time of this random walk is proportional to $1/\Phi^{2}$ up to logarithmic factors. This is captured by the following proposition:
\begin{proposition}[{\cite[Corollary 1.5]{lovasz1993random}}]\label{prop:Lovasz-Simonovits}
Let $M:=\sup_{\S\subseteq\Omega}\frac{\sigma(\S)}{\pi(\S)}$ where $\sigma$ is the initial distribution. Then for every $\S\subseteq\Omega$,
\begin{align}
\big|\sigma^{(k)}(\S)-\pi(\S)\big|\leq\sqrt{M}\Big(1-\frac{1}{2}\Phi^{2}\Big)^{k}.
\end{align}
\end{proposition}
\noindent
Furthermore, the conductance in \prop{Lovasz-Simonovits} can be relaxed to that of sets with a fixed small probability $p$:
\begin{proposition}[{\cite[Corollary 1.6]{lovasz1993random}}]\footnote{Note that in the original statement (\cite[Corollary 1.6]{lovasz1993random}), the conditions are given in a slightly different formulation, but it is not hard to obtain the conditions in the original formulation from the conditions of this proposition.}
  \label{prop:mixing-phip}
  Let $M:=\sup_{\S\subseteq\Omega}\frac{\sigma(\S)}{\pi(\S)}$. If the conductance for all connected, measurable set $\rma \subseteq \Omega$ such that $\pi(\rma) = p \leq 1/2$ is at least $\Phi_p$, then for all $\rms \subseteq \Omega$, we have
  \begin{align}
    \bigl|\sigma^{(k)}(\rms) - \pi(\rms)\bigr| \leq 2Mp + 2M\left(1-\frac{1}{2}\Phi_p^2\right)^k.
  \end{align}
\end{proposition}

Second, Ref.~\cite{LV06} proved a lower bound on the conductance of the hit-and-run walk with exponential density:
\begin{proposition}[{\cite[Theorem 6.9]{LV06}}]\label{prop:LV06-Theorem6.9}
Let $f$ be a density in $\R^{n}$ proportional to $e^{-a^T x}$ whose support is a convex body $\rmk$ of diameter $d$. Assume that $\rmb_{2}(0,r)\subseteq \rmk$. Then for any subset $\rms$ with $\pi_f(\rms) = p \leq 1/2$, the conductance of the hit-and-run walk satisfies
\begin{align}
\phi(S)\geq\frac{r}{10^{13}nd\ln(nd/rp)}.
\end{align}
\end{proposition}

\prop{Lovasz-Simonovits} and \prop{LV06-Theorem6.9} imply \thm{hit-and-run-LV06} and \thm{exp-hit-run-mixes}; complete proofs are given in~\cite{LV06}.

For the conductance of the hit-and-run walk with a uniform distribution, Ref.~\cite{LV06} established a stronger lower bound that is independent of $p$:
\begin{proposition}[{\cite[Theorem 4.2]{LV06}}]
  \label{prop:conductance-uniform}
  Assume that $\rmk$ has diameter $d$ and contains a unit ball. Then the conductance of the hit-and-run in $\rmk$ with uniform distribution is at least $\frac{1}{2^{24}nd}$.
\end{proposition}


\section{Theory of continuous-space quantum walks}\label{sec:continuous-q-walk-theory}

In this section, we develop the theory of continuous-space, discrete-time quantum walks.

Specifically, we generalize the discrete-time quantum walk of Szegedy~\cite{szegedy2004quantum} to continuous space. Let $n\in\N$ and suppose $\Omega$ is a continuous\footnote{We say that $\Omega$ is continuous if for any $x,y\in\Omega$ there is a continuous function $f_{x,y}\colon [0,1]\rightarrow\Omega$ such that $f_{x,y}(0)=x$ and $f_{x,y}(1)=y$.} subset of $\R^{n}$. A probability transition density $p$ on $\Omega$ is a continuous function $p\colon\Omega\times\Omega\to[0,+\infty)$ such that
\begin{align}\label{eqn:density-defn}
\int_{\Omega}\d y\,p(x,y)=1\qquad\forall\,x\in\Omega.
\end{align}
We also write $p_{x\to y} := p(x,y)$ for the transition density from $x$ to $y$. Together, $\Omega$ and $p$ specify a continuous-space Markov chain that we denote $(\Omega,p)$ throughout the paper.

For background on the mathematical foundations of quantum mechanics over continuous state spaces, see \cite[Chapter 1]{sakurai2014modern}. In this section, we treat quantum states as square integrable functions $f\colon\Omega\to\R$ in $L^{2}(\Omega)$ if $\int_{\Omega}\d x\,|f(x)|^{2}<\infty$. The inner product $\<\cdot,\cdot\>$ on $L^{2}(\Omega)$ is defined by
\begin{align}
\<f,g\>:=\int_{\Omega}\d x\,f(x)g(x)\qquad\forall\,f,g\in L^{2}(\Omega).
\label{eq:L2ip}
\end{align}
Note that by the Cauchy-Schwarz inequality, we have
\begin{align}
|\<f,g\>|^{2}\leq\left(\int_{\Omega}\d x\,|f(x)|^{2}\right)\left(\int_{\Omega}\d x\,|g(x)|^{2}\right)<\infty;
 \end{align}
 the norm of an $f\in L^{2}(\Omega)$ is subsequently defined as $\|f\|:=\sqrt{\<f,f\>}$. The pure states in $\Omega$ correspond to functions in the set
\begin{align}
\st(\Omega):=\Bigl\{f\colon\Omega\to\R \,\Bigm\vert \int_{\Omega}\d x\,|f(x)|^{2}=1\Bigr\}.
\end{align}

The computational basis elements $|y\>$ for all $y \in \Omega$ correspond to Dirac delta functions $\delta(x-y)$ centered at $y$, where $\delta(x)=0$ for all $x\neq 0$, and $\int_{\R^{n}}\delta(x)\,\d x=1$. Delta functions are not members of the Hilbert space $L^{2}(\Omega)$, however we interpret them in the following sense: for any $y \in \mathrm{int}_{\epsilon}(\Omega)$ we consider a normalized Gaussian with width $\epsilon$, given by $\delta_{y,\epsilon}(x) \propto \frac{1}{(2\pi \epsilon^{{2}})^{n/2}}{e^{-\lVert x - y\rVert^{2}/2\epsilon^{2}}}$ for $x \in \Omega$ and $0$ for $x \not\in \Omega$. It is clear that $\delta_{y,\epsilon} \in \L_{2}(\Omega)$ and its behavior approaches that of the delta function. In the remainder of the section statements such as $A = B$ are to be interpreted as $\lim_{\epsilon \to 0}|A_{\epsilon} - B_{\epsilon}| = 0$ where $A_{\epsilon},B_{\epsilon}$ are obtained from $A,B$ by replacing every occurence of a computational basis vector $|y\rangle$ by $\delta_{y,\epsilon}$.\footnote{As in most treatments of continous quantum mechanics, we shall not be fully rigorous with respect to operations such as interchanging orders of limits. We have two reasons to believe that pathological cases do not occur: (1) The states that occur during the algorithm are mostly well-behaved and correspond to probability distributions that are obtained during classical geometric random walks. (2) We later show \sec{discretization-hit-and-run} that our algorithm can also be executed discretely, and we work in the continuous setting for ease of analysis.} Integrals over $L_{2}(\Omega)$ are computed pointwise. It can be verified that
\begin{align}\label{eqn:identity}
  \int_{\Omega} \d x\, |x \>\< x| = I
\end{align}
and
\begin{align}
  \<x|x'\> = \delta(x - x'), \quad \forall x,x' \in \Omega.
\end{align}

\subsection{Continuous-space quantum walk}\label{sec:continuous-space-qwalk-defn}
Given a transition density function $p$, the quantum walk is characterized by the following states:
\begin{align}\label{eqn:quantum-walk-phi-defn}
|\phi_{x}\>:=|x\>\otimes\int_{\Omega}\d y\,\sqrt{p_{x\to y}}|y\>\qquad\forall x\in\R^{n}.
\end{align}
Since $p_{x \to y}$ is a normalized probability density function, $|\phi_{x}\> \in L_{2}(\Omega)$.
Now, denote
\begin{align}\label{eqn:quantum-walk-defn}
U:=\int_{\Omega}\d x\,|\phi_{x}\>(\<x|\otimes\<0|),\ \Pi:=\int_{\Omega}\d x\,|\phi_{x}\>\<\phi_{x}|,\ S:=\int_{\Omega}\int_{\Omega}\d x\,\d y\,|x,y\>\<y,x|.
\end{align}
Notice that $\Pi$ is the projection onto $\spn\{|\phi_{x}\>\}_{x\in\R^{n}}$ because
\begin{align}
\Pi^{2}=\int_{\Omega}\int_{\Omega}\d x\,\d x'\,|\phi_{x}\>\<\phi_{x}|\phi_{x'}\>\<\phi_{x'}|=\int_{\Omega}\int_{\Omega}\d x\,\d x'\,\delta(x-x')|\phi_{x}\>\<\phi_{x'}|=\Pi,
\end{align}
and $S$ is the swap operator for the two registers. A single step of the quantum walk is defined as the unitary operator
\begin{align}\label{eqn:walk-operator}
W:=S(2\Pi-I).
\end{align}

The first main result of this subsection is the following theorem:
\begin{theorem}\label{thm:quantum-walk-main}
Let
\begin{align}\label{eqn:discriminant}
D:=\int_{\Omega}\int_{\Omega}\d x\,\d y\,\sqrt{p_{x\to y}p_{y\to x}}|x\>\<y|
\end{align}
denote the \emph{discriminant operator} of $p$. Let $\Lambda$ be the set of eigenvalues of $D$, so that $D=\int_{\Lambda}\d\lambda\,\lambda|\lambda\>\<\lambda|$. Then the eigenvalues of the quantum walk operator $W$ in \eqn{walk-operator} are $\pm 1$ and $\lambda\pm i\sqrt{1-\lambda^{2}}$ for all $\lambda\in\Lambda$.
\end{theorem}

To prove \thm{quantum-walk-main}, we first prove the following lemma:
\begin{lemma}\label{lem:lambda-bound}
For any $\lambda\in\Lambda$, we have $|\lambda|\leq 1$.
\end{lemma}

\begin{proof}
Since $\lambda$ is an eigenvalue of $D$, we have $D|\lambda\>=\lambda|\lambda\>$. As a result, we have
\begin{align}
|\lambda|\delta(0)&=|\lambda|\<\lambda|\lambda\> \label{eqn:lem-spectrum-D-1} \\
&=|\<\lambda|D|\lambda\>| \\
&=\Big|\int_{\Omega}\int_{\Omega}\d x\,\d y\,\sqrt{p_{y\to x}p_{x\to y}}\<\lambda|x\>\<y|\lambda\>\Big| \\
\text{(by Cauchy-Schwarz)}\quad&\leq\sqrt{\Big(\int_{\Omega}\int_{\Omega}\d x\,\d y\,p_{y\to x}|\<y|\lambda\>|^{2}\Big)\Big(\int_{\Omega}\int_{\Omega}\d x\,\d y\,p_{x\to y}|\<\lambda|x\>|^{2}\Big)} \nonumber \\
&=\sqrt{\Big(\int_{\Omega}\d y\,|\<y|\lambda\>|^{2}\Big)\Big(\int_{\Omega}\d x\,|\<\lambda|x\>|^{2}\Big)}\quad\text{(by $\int_{\Omega}\d y\,p_{x\to y}=1$)} \\
&=\int_{\Omega}\d x\,\<\lambda|x\>\<x|\lambda\> \\
&=\<\lambda|\Big(\int_{\Omega}\d x\,|x\>\<x|\Big)|\lambda\>\quad\text{(by \eqn{identity})} \\
&=\delta(0). \label{eqn:lem-spectrum-D-2}
\end{align}
Hence the result follows.
\end{proof}

\begin{proof}[Proof of {\thm{quantum-walk-main}}]
Define an isometry
\begin{align}
T:=\int_{\Omega}\d x\,|\phi_{x}\>\<x|=\int_{\Omega}\int_{\Omega}\d x\,\d y\,\sqrt{p_{x\to y}}|x,y\>\<x|.
\end{align}
Then
\begin{align}
TT^{\dagger}=\int_{\Omega}\int_{\Omega}\d x\,\d y\,|\phi_{x}\>\<x|y\>\<\phi_{y}|=\int_{\Omega}\d x\,|\phi_{x}\>\<\phi_{x}|=\Pi,
\end{align}
and
\begin{align}
T^{\dagger}T&=\int_{\Omega}\int_{\Omega}\d x\,\d y\,|x\>\<\phi_{x}|\phi_{y}\>\<y| \\
&=\int_{\Omega}\int_{\Omega}\int_{\Omega}\int_{\Omega}\d x\,\d y\,\d a\,\d b\,\<x|y\>\<a|b\>\sqrt{p_{x\to a}p_{y\to b}}|x\>\<y| \\
&=\int_{\Omega}\int_{\Omega}\d x\,\d a\,p_{x\to a}|x\>\<x| \\
&=\int_{\Omega}\d x\,|x\>\<x| \\
&=I.
\end{align}
Furthermore,
\begin{align}
T^{\dagger}ST&=\int_{\Omega}\int_{\Omega}\d x\,\d y\,|x\>\<\phi_{x}|S|\phi_{y}\>\<y| \\
&=\int_{\Omega}\int_{\Omega}\int_{\Omega}\int_{\Omega}\d x\,\d y\,d a\,\d b\,\<x,a|S|y,b\>\sqrt{p_{x\to a}p_{y\to b}}|x\>\<y| \\
&=\int_{\Omega}\int_{\Omega}\d x\,\d a\,\sqrt{p_{x\to a}p_{a\to x}}|x\>\<a| \\
&=D.
\end{align}
As a result, for any $\lambda\in\Lambda$ we have
\begin{align}
WT|\lambda\>=S(2\Pi-I)T|\lambda\>=(2STT^{\dagger}T-ST)|\lambda\>=(2ST-ST)|\lambda\>=ST|\lambda\>.
\end{align}
Similarly, we have
\begin{align}
WST|\lambda\>&=S(2\Pi-I)ST|\lambda\> \\
&=(2STT^{\dagger}ST-S^{2}T)|\lambda\>=(2STD-T)|\lambda\>=(2\lambda S-I)T|\lambda\>. \nonumber
\end{align}
By \lem{lambda-bound}, $|\lambda|\leq 1$. As a result, we have
\begin{align}
W\big(I-(\lambda+i\sqrt{1-\lambda^{2}})S\big)T|\lambda\>&=WT|\lambda\>-(\lambda+i\sqrt{1-\lambda^{2}})WST|\lambda\> \\
&=ST|\lambda\>-(\lambda+i\sqrt{1-\lambda^{2}})(2\lambda S-I)T|\lambda\> \\
&=\big(S-(\lambda+i\sqrt{1-\lambda^{2}})(2\lambda S-I)\big)T|\lambda\> \\
&=(\lambda+i\sqrt{1-\lambda^{2}})\big(I-(\lambda+i\sqrt{1-\lambda^{2}})S\big)T|\lambda\>;
\end{align}
in other words, $\lambda+i\sqrt{1-\lambda^{2}}$ is an eigenvalue of $W$ with eigenvector $\big(I-(\lambda+i\sqrt{1-\lambda^{2}})S\big)T|\lambda\>$. Similarly, we have
\begin{align}
W\big(I-(\lambda-i\sqrt{1-\lambda^{2}})S\big)T|\lambda\>=(\lambda-i\sqrt{1-\lambda^{2}})\big(I-(\lambda-i\sqrt{1-\lambda^{2}})S\big)T|\lambda\>,
\end{align}
i.e., $\lambda-i\sqrt{1-\lambda^{2}}$ is an eigenvalue of $W$ with eigenvector $\big(I-(\lambda-i\sqrt{1-\lambda^{2}})S\big)T|\lambda\>$.

Finally, note that for any vector $|u\>$ in the orthogonal complement of the space $\spn_{\lambda\in\Lambda}\{T|\lambda\>,ST|\lambda\>\}$, $W$ simply acts as $-S$ since
\begin{align}
\Pi=TT^{\dagger}=\int_{\Lambda}\d\lambda\,T|\lambda\>\<\lambda|T^{\dagger},
\end{align}
which projects onto $\spn_{\lambda\in\Lambda}\{T|\lambda\>\}$. In this orthogonal complement subspace, the eigenvalues are $\pm 1$ because $S^{2}=I$.
\end{proof}

\subsection{Stationary distribution}\label{sec:continuous-space-qwalk-stationary}

Classically, the density $\pi=(\pi_{x})_{x\in\Omega}$ corresponding to the stationary distribution of a Markov chain $(\Omega,p)$ satisfies
\begin{align}\label{eqn:stationary-defn}
\int_{\Omega}\d x\,\pi_{x}=1;\qquad\int_{\Omega}\d y\,p_{y\to x}\pi_{y}=\pi_{x}\qquad\forall\,x\in\Omega.
\end{align}
In other words, we can naturally define a transition operator as
\begin{align}\label{eqn:transition-P}
P:=\int_{\Omega}\int_{\Omega}\d x\,\d y\,p_{y\to x}|x\>\<y|,
\end{align}
and the stationary density $\pi$ satisfies $P\pi=\pi$. The Markov chain $(\Omega,p)$ is \emph{reversible} if there exists a classical density $\sigma=(\sigma_{x})_{x\in\Omega}$ such that
\begin{align}\label{eqn:detailed-balance}
p_{y\to x}\sigma_{y}=p_{x\to y}\sigma_{x}\qquad\forall\,x,y\in\Omega.
\end{align}
(This is called the \emph{detailed balance condition}.) Notice that for all $x\in\Omega$,
\begin{align}
\int_{\Omega}\d y\,p_{y\to x}\sigma_{y}=\int_{\Omega}\d y\,p_{x\to y}\sigma_{x}=\sigma_{x}\int_{\Omega}\d y\,p_{x\to y}=\sigma_{x};
\end{align}
therefore, we must have $P\sigma=\sigma$, i.e., $\sigma$ is a stationary density of $P$. In this paper, we focus on Markov chains $(\Omega,p)$ that are reversible and have a unique stationary distribution (i.e., $\sigma=\pi$). Such assumptions are natural for Markov chains in practice, including the Metropolis-Hastings algorithm, simple random walks on graphs, etc.

We point out that if $\pi$ is the classical stationary density of a reversible Markov chain $(\Omega,p)$, then
\begin{align}\label{eqn:quantum-stationary-defn}
|\pi_{W}\>:=\int_{\Omega}\d x\,\sqrt{\pi_{x}}|\phi_{x}\>
\end{align}
is the unique eigenvalue-$1$ eigenstate of the quantum walk operator $W$ restricted to the subspace $\spn_{\lambda\in\Lambda}\{T|\lambda\>,ST|\lambda\>\}$. First, a simple calculation shows that
\begin{align}
W|\pi_{W}\>&=S(2\Pi-I)|\pi_{W}\> \\
&=S|\pi_{W}\> \label{eqn:quantum-stationary-proof-1} \\
&=\Big(\int_{\Omega}\int_{\Omega}\d x\,\d y\,|x,y\>\<y,x|\Big)\Big(\int_{\Omega}\int_{\Omega}\d x\,\d y\,\sqrt{\pi_{y}p_{y\to x}}|y,x\>\Big) \label{eqn:quantum-stationary-proof-2} \\
&=\int_{\Omega}\int_{\Omega}\d x\,\d y\,\sqrt{\pi_{y}p_{y\to x}}|x,y\> \\
&=\int_{\Omega}\int_{\Omega}\d x\,\d y\,\sqrt{\pi_{x}p_{x\to y}}|x\>|y\> \label{eqn:quantum-stationary-proof-3} \\
&=\int_{\Omega}\d x\,\sqrt{\pi_{x}}|x\>\Big(\int_{\Omega}\d y\,\sqrt{p_{x\to y}}|y\>\Big) \\
&=\int_{\Omega}\d x\,\sqrt{\pi_{x}}|\phi_{x}\> \label{eqn:quantum-stationary-proof-4} \\
&=|\pi_{W}\>,
\end{align}
where \eqn{quantum-stationary-proof-1} follows from $|\pi_{W}\>\in\spn_{x\in\Omega}\{|\phi_{x}\>\}$, \eqn{quantum-stationary-proof-2} follows from the definition of $S$ in \eqn{quantum-walk-defn}, \eqn{quantum-stationary-proof-3} follows from \eqn{detailed-balance}, and \eqn{quantum-stationary-proof-3} follows from the definition of $|\phi_{x}\>$ in \eqn{quantum-walk-phi-defn}. Thus $|\pi_W\>$ is an eigenvector of $W$ with eigenvalue $1$. On the other hand, since $(\Omega,p)$ is reversible, $P$ is similar to $D$: if we denote $D_{\pi}:=\int_{\Omega}\d x\,\sqrt{\pi_{x}}|x\>\<x|$, then
\begin{align}
D_{\pi}DD_{\pi}^{-1}&=\Big(\int_{\Omega}\d x\sqrt{\pi_{x}}|x\>\<x|\Big)\Big(\int_{\Omega}\int_{\Omega}\d x\,\d y\sqrt{p_{x\to y}p_{y\to x}}|x\>\<y|\Big)\Big(\int_{\Omega}\d y\sqrt{\pi_{y}^{-1}}|y\>\<y|\Big) \nonumber\\
&=\int_{\Omega}\int_{\Omega}\d x\,\d y\,\sqrt{\pi_{x}\pi_{y}^{-1}p_{x\to y}p_{y\to x}}|x\>\<y| \\
&=\int_{\Omega}\int_{\Omega}\d x\,\d y\,p_{y\to x}|x\>\<y|\quad\text{(by \eqn{detailed-balance})} \\
&=P.\label{eq:DPsimilar}
\end{align}
As a result, $D$ and $P$ have the same set of eigenvalues. Furthermore, \lem{lambda-bound} implies that all eigenvalues of $P$ have norm at most 1, and the proof of \thm{quantum-walk-main} shows that $|\pi_{W}\>$ is the unique eigenvector with this eigenvalue within $\spn_{\lambda\in\Lambda}\{T|\lambda\>,ST|\lambda\>\}$.

The state
\begin{align}\label{eqn:quantum-stationary}
|\pi\>:=\int_{\Omega}\d x\,\sqrt{\pi_{x}}|x\>
\end{align}
represents a quantum sample from the density $\pi$; in particular, measuring $|\pi\>$ in the computational basis gives a classical sample from $\pi$. Furthermore, the unitary operator in \eqn{quantum-walk-defn} satisfies
\begin{align}
U^{\dagger}|\pi_{W}\>=\Big(\int_{\Omega}\d x\,|x\>|0\>\<\phi_{x}|\Big)\Big(\int_{\Omega}\d x\,\sqrt{\pi_{x}}|\phi_{x}\>\Big)=\int_{\Omega}\d x\,\sqrt{\pi_{x}}|x\>|0\>=|\pi\>|0\>,
\end{align}
so we have $U|\pi\>|0\> = |\pi_{W}\>$.


\section{Quantum speedup for volume estimation}\label{sec:quantum-volume}
In this section, we present and analyze our quantum algorithm for volume estimation. First, we review the classical state-of-the-art volume estimation algorithm in \sec{volume-estimation-review}. We then describe our quantum algorithm for estimating the volume of well-rounded convex bodies (i.e., $R/r=O(\sqrt{n})$) with query complexity $\tilde{O}(n^{2.5}/\epsilon)$ in \sec{quantum-volume-estimation}, with detailed proofs given in \sec{quantum-volume-estimation-proof}. Finally, we remove the well-rounded condition by giving a quantum algorithm with interlaced rounding and volume estimation with additional cost $\tilde{O}(n^{2.5})$ in each iteration in \sec{round-quantum}.

\subsection{Review of classical algorithms for volume estimation}\label{sec:volume-estimation-review}
The state-of-the-art classical algorithm for volume estimation uses $\tilde{O}(n^{4}+n^{3}/\epsilon^{2})$ classical queries, where $\tilde{O}(n^{4})$ queries are used to construct the affine transformation that makes convex body well-rounded~\cite{lovasz2006simulated} and $\tilde{O}(n^{3}/\epsilon^{2})$ queries are used to estimate the volume of the well-rounded convex body (after the affine transformation)~\cite{cousins2015bypassing}.

We now review the algorithm of \cite{lovasz2006simulated} for estimating volumes of well-rounded convex bodies. This algorithm estimates the volume of a convex body obtained by the following \emph{pencil construction}. Define the cone
\begin{align}
\Cone:=\Big\{x\in\R^{n+1}: x_{0}\geq 0,\sum_{i=1}^{n}x_{i}^{2}\leq x_{0}^{2}\Big\}.
\end{align}
Let $\K'$ be the intersection of the cone $\Cone$ and a cylinder $[0,2D]\times \K$, i.e.,
\begin{align}\label{eqn:pencil-defn}
\K':=([0,2D]\times \K)\cap \Cone
\end{align}
(recall $D=R/r$). Without loss of generality we renormalize to $r=1$, so that $\B_{2}(0,1)\subseteq\K\subseteq \B_{2}(0,D)$. Since $D\vol(\K)\leq \vol(\K')\leq 2D\vol(\K)$, with the knowledge of $\vol(\K')$ we can estimate $\vol(\K)$ with multiplicative error $\epsilon$ by generating $O(1/\epsilon^{2})$ sample points from the uniform distribution on $[0,2D]\times \K$ and then counting how many of them fall into $\K'$. Such an approximation succeeds with high probability by a Chernoff-type argument (see \sec{pencil-construction-analysis} for a formal proof).

Lov{\'a}sz and Vempala \cite{lovasz2006simulated} considers simulated annealing under the pencil construction. For any $a>0$, define
\begin{align}
Z(a):=\int_{\K'}e^{-ax_{0}}\,\d x.
\end{align}
It can be shown that for any $a\leq\epsilon/D$,
\begin{align}\label{eqn:LV06-SA-small}
(1-\epsilon)\vol(\K')\leq Z(a)\leq\vol(\K').
\end{align}
On the other hand, it is shown in \cite[Section 2.3]{lovasz2006simulated} that for any $a\geq 2n$ and $\epsilon>(3/4)^{n}$,
\begin{align}\label{eqn:LV06-SA-large}
(1-\epsilon)\int_{\Cone}e^{-ax_{0}}\,\d x\leq Z(a)\leq\int_{\Cone}e^{-ax_{0}}\,\d x.
\end{align}
This suggests using a simulated annealing procedure for estimating $\vol(\K')$. Specifically, if we select a sequence $a_{0}>a_{1}>\cdots>a_{m}$ for which $a_{0}=2n$ and $a_{m}\leq\epsilon/D$, then we can estimate $\vol(\K')$ by
\begin{align}
Z(a_{m})=Z(a_{0})\prod_{i=0}^{m-1}\frac{Z(a_{i+1})}{Z(a_{i})}.
\end{align}
(Note that this procedure uses an increasing sequence of temperatures $1/a_i$, unlike standard simulated annealing in which temperature is decreased.)

Let $\pi_{i}$ be the probability distribution over $\K'$ with density proportional to $e^{-a_{i}x_{0}}$, i.e., $\d\pi_{i}(x)=\frac{e^{-a_{i}x_{0}}}{Z(a_{i})}\d x$. Let $X_{i}$ be a random sample from $\pi_{i}$, and let $(X_{i})_{0}$ be its first coordinate; define $V_{i}:=e^{(a_{i}-a_{i+1})(X_{i})_{0}}$. We have
\begin{align}\label{eqn:Vi-expectation}
\E_{\pi_{i}}[V_{i}]=\int_{\K'}e^{(a_{i}-a_{i+1})x_{0}}\,\d\pi_{i}(x)=\int_{\K'}e^{(a_{i}-a_{i+1})x_{0}}\frac{e^{-a_{i}x_{0}}}{Z(a_{i})}\,\d x=\frac{Z(a_{i+1})}{Z(a_{i})}.
\end{align}
Furthermore, if the simulated annealing schedule satisfies $a_{i+1}\geq(1-\frac{1}{\sqrt{n}})a_{i}$, then $V_{i}$ satisfies (see~\cite[Lemma 4.1]{lovasz2006simulated})
\begin{align}\label{eqn:LV06-Lemma4.1}
\frac{\E_{\pi_{i}}[V_{i}^{2}]}{\E_{\pi_{i}}[V_{i}]^{2}}\leq\Big(\frac{a_{i+1}^{2}}{a_{i}(2a_{i+1}-a_{i})}\Big)^{n+1}<8\qquad\forall\,i\in\range{m},
\end{align}
i.e., the variance of $V_{i}$ is bounded by a constant multiple of the square of its expectation. Thus, this simulated annealing procedure constitutes \emph{Chebyshev cooling} (see also \sec{nondestructive}), ensuring its correctness (see \prop{Chebyshev-classical}).

\algo{LV06-volume} presents this approach in detail. Note that sampling from $\pi_{0}$ in \lin{LV06-initial-sample} is straightforward: select a random positive real number $x_{0}$ from the distribution with density $e^{-2nx}$ and choose a uniformly random point $(v_{1},\ldots,v_{n})$ from the unit ball. If $X=(x_{0},x_{0}v_{1},\ldots,x_{0}v_{n})\notin\K'$, try again; else return $X$. Equation \eqn{LV06-SA-large} ensures that we succeed with probability at least $1-\epsilon$ for each sample.

\begin{algorithm}[htbp]
\KwInput{Membership oracle $O_{\K}$ of $\K$; $R$ such that $\B_{2}(0,1)\subseteq\K\subseteq \B_{2}(0,R)$; $R=O(\sqrt{n})$, i.e., $\K$ is well-rounded.}
\KwOutput{$\epsilon$-multiplicative approximation of $\vol(\K)$.}
Set $m=2\lceil\sqrt{n}\ln(n/\epsilon)\rceil$, $k=\frac{512}{\epsilon^{2}}\sqrt{n}\ln(n/\epsilon)$, $\delta=\epsilon^{2}n^{-10}$, and $a_{i}=2n(1-\frac{1}{\sqrt{n}})^{i}$ for $i\in\range{m}$\;
Take $k$ samples $X_{0}^{(1)},\ldots,X_{0}^{(k)}$ from $\pi_{0}$\; \label{lin:LV06-initial-sample}
\For{$i\in\range{m}$}{
	Take $k$ samples from $\pi_{i}$ with error parameter $\delta$ and starting points $X_{i-1}^{(1)},\ldots,X_{i-1}^{(k)}$, giving points $X_{i}^{(1)},\ldots,X_{i}^{(k)}$\;\label{lin:LV06-sample-SA}
	Compute $V_{i}=\frac{1}{k}\sum_{j=1}^{k}e^{(a_{i}-a_{i+1})(X_{i}^{(j)})_{0}}$\;\label{lin:LV06-sample-average}
}
Return $n!v_{n}(2n)^{-(n+1)}V_{1}\cdots V_{m}$ as the estimate of the volume of $\K'$, where $v_{n}:=\pi^{\frac{n}{2}}/\Gamma(1+\frac{n}{2})$ is the volume of the $n$-dimensional unit ball\;
\caption{Volume estimation of well-rounded $\K$ with $\tilde{O}(n^{4}/\epsilon^{2})$ classical queries~\cite{lovasz2006simulated}.}
\label{algo:LV06-volume}
\end{algorithm}

\subsection{Quantum algorithm for volume estimation}\label{sec:quantum-volume-estimation}

As introduced in \sec{techniques}, our quantum algorithm has four main improvements that contribute to the quantum speedup of \algo{LV06-volume}:
\begin{enumerate}[leftmargin=*]
\item We replace the classical hit-and-run walk in \sec{hit-and-run} by a quantum hit-and-run walk, defined using the framework of \sec{continuous-q-walk-theory}. Classically, the hit-and-run walk mixes in $\tilde{O}(n^{3})$ steps in a well-rounded convex body given a warm start (see \thm{hit-and-run-LV06}). Quantumly, we can use the quantum hit-and-run walk operator to prepare its stationary state given a warm start state using only $\tilde{O}(n^{1.5})$ queries to the membership oracle for the well-rounded convex body.

\item We replace the simulated annealing framework in \algo{LV06-volume} by the quantum MCMC framework described in \sec{quantum-MCMC}. Classically, we sample from $\pi_i$ in the $i^{\text{th}}$ iteration by running the classical hit-and-run walk starting from the samples taken in the $(i-1)^{\text{st}}$ iteration. Quantumly, we prepare the quantum sample $|\pi_i\>$ in the $i^{\text{th}}$ iteration by applying $\pi/3$-amplitude amplification to a quantum sample produced in the $(i-1)^{\text{st}}$ iteration, where the unitaries in the $\pi/3$-amplitude amplification are implemented by phase estimation of the quantum hit-and-run walk operators as in \eqn{Wocjan-walk-PhaseEst}.

\item We use the quantum Chebyshev inequality (see \sec{quantum-Chebyshev}) to give a quadratic quantum speedup in $\epsilon^{-1}$ when taking the average $e^{(a_{i}-a_{i+1})(\bar{X}_{i})_{0}}$ in \lin{LV06-sample-average} of \algo{LV06-volume}. However, we must be cautious because the resulting points $X_{i}^{(1)},\ldots,X_{i}^{(k)}$ in \lin{LV06-sample-SA} follow the distribution $\pi_{i}$, which varies in different iterations of simulated annealing. Instead, our quantum algorithm must be \emph{nondestructive}: it must still have a copy of $|\pi_{i}\>$ after estimating the average $e^{(a_{i}-a_{i+1})(\bar{X}_{i})_{0}}$, so that we can map this state to $|\pi_{i+1}\>$ by $\pi/3$-amplitude amplification for the next iteration. This is achieved in \sec{nondestructive}.

\item In \sec{round-quantum}, we show how the densities can be transformed to be well-rounded by an affine transformation at each stage of the algorithm. This is to ensure that the hit-and-run walk mixes fast assuming the densities $\pi_i$ to be sampled from are well-rounded (see \thm{exp-hit-run-mixes}). The high-level idea is to sample points from density $\pi_{i}$ and compute an affine transformation $S_{i+1}$ that rounds $\pi_{i}$ and the next density $\pi_{i+1}$ (see \lem{affine-transform-isotropic}). To sample these points, we use $\pi/3$-amplitude amplification to map the states corresponding to the uniform distributions for one stage to those for the next. The affine transformation can be computed coherently using nondestructive mean estimation, with $\tilde{O}(n^{2.5})$ quantum queries in each iteration.
\end{enumerate}

\algo{quantum-volume} is our quantum volume estimation algorithm that satisfies our main theorem:

\maintheorem*

\begin{algorithm}[ht]
\KwInput{Membership oracle $O_{\K}$ for $\K$; $R = O(\sqrt n)$ such that $\B_{2}(0,1)\subseteq\K\subseteq \B_{2}(0,R)$.}
\KwOutput{$\epsilon$-multiplicative approximation of $\vol(\K)$.}
Set $m=\Theta(\sqrt{n}\log(n/\epsilon))$ to be the number of iterations of simulated annealing and $a_{i}=2n(1-\frac{1}{\sqrt{n}})^{i}$ for $i\in\range{m}$. Let $\pi_{i}$ be the probability distribution over $\K'$ with density proportional to $e^{-a_{i}x_{0}}$\; \nonl
Set error parameters $\delta,\epsilon' = \Theta(\epsilon/m^2), \epsilon_1 = \epsilon/2m$; let $k=\tilde{\Theta}(\sqrt{n}/\epsilon)$ be the number of copies of stationary states for applying the quantum Chebyshev inequality; let $l = \tilde{\Theta}(n)$ be the number of copies of stationary states needed to obtain the affine transformation $S_i$\; \nonl
Prepare $k+l$ (approximate) copies of $|\pi_{0}\>$, denoted $|\tilde\pi_{0}^{(1)}\>,\ldots,|\tilde\pi_{0}^{(k+l)}\>$ (\lem{mean-estimation-error})\;\label{lin:quantum-initial-sample}
\For{$i\in\range{m}$}{
  Use the quantum Chebyshev inequality on the $k$ copies of the state $|\tilde{\pi}_{i-1}\>$ with parameters $\epsilon_1,\delta$ to estimate the expectation $\E_{\pi_{i}}[V_i]$ (in Eq.~\eqn{Vi-expectation}) as $\tilde{V_i}$ (\lem{err-non-destructive} and \fig{SA-block-nondes}). The post-measurement states are denoted $|\hat\pi_{i-1}^{(1)}\>,\ldots,|\hat\pi_{i-1}^{(k)}\>$\;\label{lin:quantum-sample-SA}
  Use the $l$ copies of the state $|\pi_{i-1}\>$ to nondestructively obtain the affine transformation $S_{i}$
  that rounds $\pi_{i-1}$ and $\pi_{i}$ (\sec{round-quantum}). The post-measurement states are denoted $|\hat\pi_{i-1}^{(k+1)}\>,\ldots,|\hat\pi_{i-1}^{(k+l)}\>$\;\label{lin:get-affine}
  Apply $\pi/3$-amplitude amplification with error $\epsilon'$ (\sec{quantum-MCMC} and \lem{pi3-error}) and affine transformation $S_i$ to map $|S_i\hat\pi_{i-1}^{(1)}\>,\ldots,|S_i\hat\pi_{i-1}^{(k+l)}\>$ to $|S_i\tilde\pi_{i}^{(1)}\>,\ldots,|S_i\tilde\pi_{i}^{(k+l)}\>$, using the quantum hit-and-run walk\;\label{lin:quantum-pi/3-AM-AM}
  Invert $S_i$ to get $k + l$ (approximate) copies of the stationary distribution $|\pi_{i}\>$ for use in the next iteration\;
}
Compute an estimate $\widetilde{\text{Vol}(\K')} = n!v_{n}(2n)^{-(n+1)}\tilde{V}_{1}\cdots \tilde{V}_{m}$ of the volume of $\K'$, where $v_{n}$ is the volume of the $n$-dimensional unit ball\;
Use $\widetilde{\text{Vol}(\K')}$ to estimate the volume of $\K$ as $\widetilde{\text{Vol}(\K)}$ (\sec{pencil-construction-analysis}).
\caption{Volume estimation of convex $\K$ with $\tilde{O}(n^{3} + n^{2.5}/\epsilon)$ quantum queries.}
\label{algo:quantum-volume}
\end{algorithm}

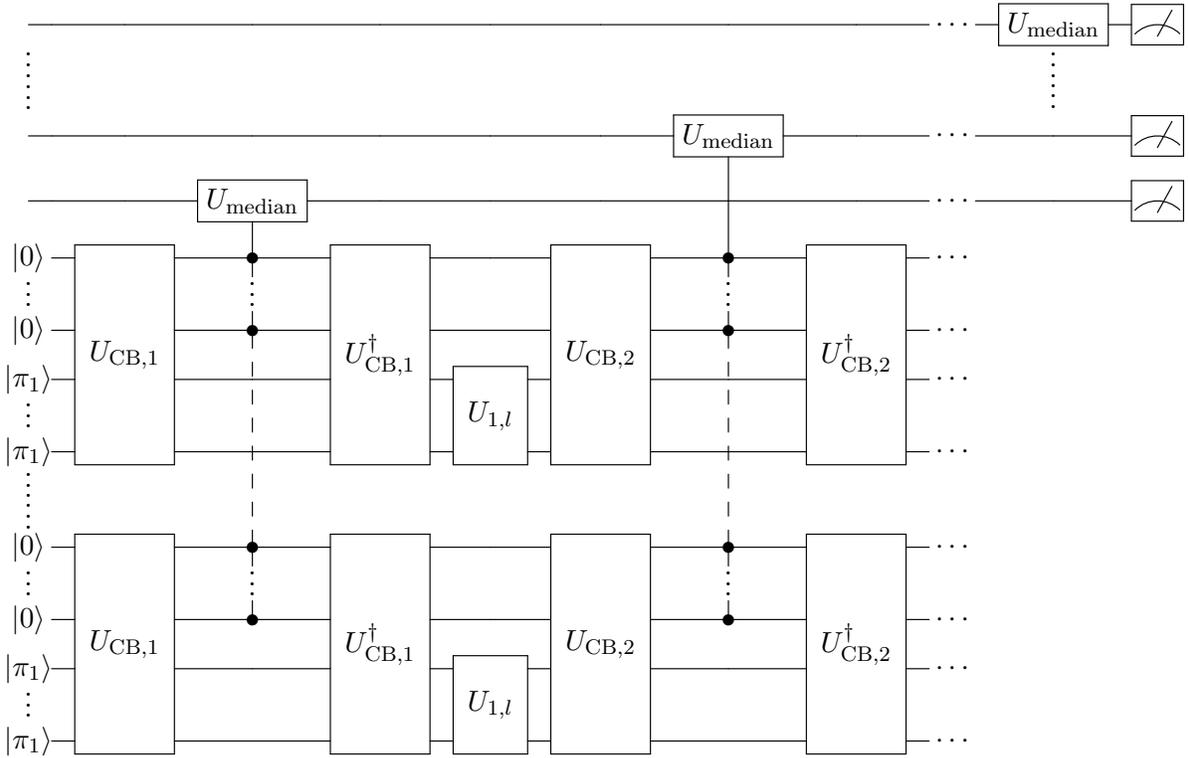
\begin{figure}[ht]
\resizebox{1\linewidth}{!}{
\Qcircuit @C=0.9em @R=0.9em {
   & \qw & \qw & \qw & \qw & \qw & \qw & \qw & \qw & \qw & \cdots & & \gate{U_{\median}} & \meter \\
   \raisebox{9pt}{\vdots} & & & & & & & & & & & & \raisebox{9pt}{\vdots} \\
   \raisebox{2pt}{\vdots} & & & & & & & & & & & & \raisebox{3pt}{\vdots} \\
   & \qw & \qw & \qw & \qw & \qw & \qw & \gate{U_{\median}} & \qw & \qw & \cdots & & \qw & \meter \\
   & \qw & \qw & \gate{U_{\median}} & \qw & \qw & \qw & \qw & \qw & \qw & \cdots & & \qw & \meter \\
   |0\> & & \multigate{5}{U_{\Cheby,1}} & \ctrl{-1} {\ar @{-}+<0em,-0.65em>} & \multigate{5}{U_{\Cheby,1}^{\dagger}} & \qw & \multigate{5}{U_{\Cheby,2}} & \ctrl{-2} {\ar @{-}+<0em,-0.65em>} & \multigate{5}{U_{\Cheby,2}^{\dagger}} & \qw & \cdots & \\
   \raisebox{6pt}{\vdots} & & & \raisebox{6pt}{\vdots} & & & & \raisebox{6pt}{\vdots} & & & \\
   |0\> & & \ghost{U_{\Cheby,1}} & \dctrl{6}{\ar @{-}+<0em,0.65em>} & \ghost{U_{\Cheby,1}^{\dagger}} & \qw & \ghost{U_{\Cheby,2}} & \dctrl{6}{\ar @{-}+<0em,0.65em>} & \ghost{U_{\Cheby,2}^{\dagger}} & \qw & \cdots \\
   |\pi_{1}\> & & \ghost{U_{\Cheby,1}} & \qw & \ghost{U_{\Cheby,1}^{\dagger}} & \multigate{2}{U_{1,l}} & \ghost{U_{\Cheby,2}} & \qw & \ghost{U_{\Cheby,2}^{\dagger}} & \qw & \cdots \\
   \raisebox{6pt}{\vdots} & & & & & & & & & & \\
   |\pi_{1}\> & & \ghost{U_{\Cheby,1}} & \qw & \ghost{U_{\Cheby,1}^{\dagger}} & \ghost{U_{1,l}} & \ghost{U_{\Cheby,2}} & \qw & \ghost{U_{\Cheby,2}^{\dagger}} & \qw & \cdots \\
      \raisebox{8pt}{\vdots} & & & & & & & & & & \\
      \raisebox{2pt}{\vdots} & & & & & & & & & & \\
|0\> & & \multigate{5}{U_{\Cheby,1}} & \ctrl{0} {\ar @{-}+<0em,-0.65em>} & \multigate{5}{U_{\Cheby,1}^{\dagger}} & \qw & \multigate{5}{U_{\Cheby,2}} & \ctrl{0} {\ar @{-}+<0em,-0.65em>} & \multigate{5}{U_{\Cheby,2}^{\dagger}} & \qw & \cdots \\
   \raisebox{6pt}{\vdots} & & & \raisebox{6pt}{\vdots} & & & & \raisebox{6pt}{\vdots} & & & \\
   |0\> & & \ghost{U_{\Cheby,1}} & \ctrl{0}{\ar @{-}+<0em,0.65em>} & \ghost{U_{\Cheby,1}^{\dagger}} & \qw & \ghost{U_{\Cheby,2}} & \ctrl{0}{\ar @{-}+<0em,0.65em>} & \ghost{U_{\Cheby,2}^{\dagger}} & \qw & \cdots \\
   |\pi_{1}\> & & \ghost{U_{\Cheby,1}} & \qw & \ghost{U_{\Cheby,1}^{\dagger}} & \multigate{2}{U_{1,l}} & \ghost{U_{\Cheby,2}} & \qw & \ghost{U_{\Cheby,2}^{\dagger}} & \qw & \cdots \\
   \raisebox{6pt}{\vdots} & & & & & & & & & & \\
   |\pi_{1}\> & & \ghost{U_{\Cheby,1}} & \qw & \ghost{U_{\Cheby,1}^{\dagger}} & \ghost{U_{1,l}} & \ghost{U_{\Cheby,2}} & \qw & \ghost{U_{\Cheby,2}^{\dagger}} & \qw & \cdots
}
}
\caption{The quantum circuit for \algo{quantum-volume} (assuming well-roundedness). Here $U_{\Cheby,i}$ is the circuit of the quantum Chebyshev inequality (\thm{quantum-Chebyshev}) in the $i^{\text{th}}$ iteration and $U_{i,l}$ is $\pi/3$-amplitude amplification from $|\pi_{i}\>$ to $|\pi_{i+1}\>$.}
\label{fig:quantum-volume}
\end{figure}

More generally, our framework paves the way of combining several different ingredients in quantum computing, and it could be used to provide quantum speedup for classical simulated annealing algorithms based on Chebyshev cooling, i.e., those with the property that the random variable in each iteration has bounded ratio between its variance and the square of its expectation. This might be of independent interest.

The proof of \thm{main} is organized as follows. We first assume that in each iteration, $S_{i+1}$ puts $\pi_{i+1}$ in isotropic position, i.e., the densities are promised to be well-rounded. The rest of this subsection presents an overview of the proof of \thm{main} (including a quantum circuit in \fig{quantum-volume}), and proofs details are given in \sec{quantum-volume-estimation-proof}. In \sec{round-quantum}, we show how the well-roundedness be maintained at an additional cost of $\tilde{O}(n^{2.5})$ quantum queries in each iteration.

Following the discussion in \sec{techniques}, our proof can be described at three levels:

\paragraph{\textsf{High level} (the simulated annealing framework)} In \sec{pencil-construction-analysis}, we show how to estimate $\vol(\K)$ given an estimate of the volume of the pencil construction, $\vol(\K')$:

\begin{restatable}{lemma}{pencil}\label{lem:pencil-to-original}
If we have access to $\widetilde{\vol(\K')}$ such that
\begin{align}\label{eq:pencil-to-original-approx}
\frac{1}{1+\epsilon/2}\vol(\K')\leq\widetilde{\vol(\K')}\leq(1+\epsilon/2)\vol(\K')
\end{align}
with probability at least 0.7, then we can return a value $\widetilde{\vol(\K)}$ such that
\begin{align}
\frac{1}{1+\epsilon}\vol(\K)\leq\widetilde{\vol(\K)}\leq(1+\epsilon)\vol(\K)
\end{align}
holds with probability at least $2/3$, using $\tilde{O}(n^{2.5}/\epsilon)$ quantum queries to the membership oracle $O_{\K}$.
\end{restatable}

In \sec{proof-inner-product}, we prove that the inner product between stationary states of consecutive simulated annealing steps is at least a constant:

\begin{restatable}{lemma}{innerproduct}\label{lem:inner-product}
Let $\ket{\pi_i}$ be the stationary distribution state of the quantum walk $W_i$ for $i \in [m]$ defined in \eqn{quantum-stationary}. For $n \geq 2$, we have $\<\pi_i|\pi_{i+1}\> > 1/3$ for $i \in [m-1]$.
\end{restatable}

This allows $\pi/3$-amplitude amplification to transform the stationary state of one Markov chain to that of the next. The total number of iterations of $\pi/3$-amplitude amplification is thus $\tilde{O}(\sqrt{n})$, which equals to the number of iterations of the classical volume estimation algorithm by~\cite{lovasz2006simulated}.

\paragraph{\textsf{Middle level} (each telescoping ratio)} In \sec{nondestructive}, we describe how we apply the quantum Chebyshev inequality (\thm{quantum-Chebyshev}) to the Chebyshev cooling schedule.

\begin{restatable}{lemma}{Chebyshev}\label{lem:Chebyshev-lemma}

  Given $\tilde{O}(\log(1/\delta)/\epsilon)$ copies of the quantum states $|\pi_{i-1}\>$, there exists a quantum algorithm that outputs an estimate of  $\E_{\pi_{i}}[V_i]$ (in Eq.~\eqn{Vi-expectation}) with relative error less than $\epsilon$ with probability at least $1 - O(\delta)$ using $\tilde{O}(\mathcal{C}\log(1/\delta)/\epsilon)$ oracle calls, where $\mathcal{C}$ oracle calls are required to implement a sampler for $|\pi_i\>$. Moreover, this quantum algorithm is nondestructive, i.e., the initial copies of quantum states $|\pi_{i-1}\>$ are restored after the computation with probability at least $1 - O(\delta)$.
\end{restatable}

Because the relative error in each iteration for estimating the volume via Chebyshev cooling is $\Theta(\epsilon/m)=\tilde{\Theta}(\epsilon/\sqrt{n})$, \lem{Chebyshev-lemma} implies that $O(\log(1/\delta)/(\epsilon/\sqrt{n}))=\tilde{O}(\sqrt{n}/\epsilon)$ copies of the stationary state suffice for the simulated annealing framework.\footnote{\label{fnote:copies}Notice that in the quantum Chebyshev inequality by Hamoudi and Magniez, they did not distinguish the number of copies of the initial state from the number of quantum samples~\cite[Theorem 5.3]{hamoudi2019Chebyshev}. In fact, their proof uses only $O(\log(1/\delta))$ copies of the initial state $|\pi_{i-1}\>$ in \lem{Chebyshev-lemma}, which reduces the total number of copies of the stationary states in the simulated annealing framework to $O(\log(1/\delta))$. Nevertheless, this does not change the total query and time complexities of our quantum algorithms because the total number of calls to the quantum sampler remains the same.}

\paragraph{\textsf{Low level} (the quantum hit-and-run walk)} In \sec{proof-error-analysis}, we give a careful analysis of the errors coming from the quantum Chebyshev inequality as well as the $\pi/3$-amplitude amplification:

\begin{restatable}{lemma}{error}\label{lem:mean-estimation-error}
For $\epsilon_{1} < 1$, given $\tilde{O}(\log(1/\delta)/\epsilon_1)$ copies of a state $|\tilde\pi_{i-1}\>$ such that $\lVert |\tilde\pi_{i-1}\> - |\pi_{i-1}\> \rVert \le \epsilon_1$,
there exists a quantum procedure (using $\pi/3$-amplitude amplification and the quantum Chebyshev inequality) that outputs a $\tilde{V}_{i}$ such that $|\tilde{V}_{i} - \E_{\pi_{i}}[V_{i}]| \le \epsilon_1\E_{\pi_{i}}[V_{i}]$ (where $\E_{\pi_{i}}[V_{i}]$ is defined in Eq.~\eqn{Vi-expectation}) with success probability $1-\delta^4$ using $\tilde{O}(n^{3/2}\log(1/\delta)/\epsilon_1 + n^{3/2}\log(1/\epsilon'))$ calls to the membership oracle and returns $\tilde{O}(\log(1/\delta)/\epsilon_1)$ copies of final states $|\tilde{\pi}_{i}\>$ such that $\lVert|\tilde\pi_{i}\> - |\pi_{i}\>\rVert = O(\epsilon_1 + \delta + \epsilon')$.
\end{restatable}

Having the four lemmas above from all the three levels, we establish \thm{main} as follows.

\begin{proof}[Proof of {\thm{main}}]
We prove the correctness and analyze the cost separately.

\hd{Correctness}
By \lem{pencil-to-original}, it suffices to compute the volume of the pencil construction $\K'$ to relative error $\epsilon/2$ with probability at least $0.7$ in order to compute the volume of the well-rounded convex body $\K$. This is computed as a telescoping sum of $m = O(\sqrt{n} \log{n/\epsilon})$ products of the form $Z(a_{i+1})/Z(a_{i})$. The random variable $V_i$ is an unbiased estimator for $Z(a_{i+1})/Z(a_{i})$, i.e., $\E_{\pi_i}[V_i] = Z(a_{i+1})/Z(a_{i})$. Consider applying \lem{mean-estimation-error} $m$ times with sufficiently small $\delta,\epsilon'\leq \epsilon/12m^2$ and $\epsilon_1 = \epsilon/3m$. This will promise that $\epsilon_1 + \delta + \epsilon'\leq \epsilon/2m$. At each iteration $i$ we  have a state $|\tilde\pi_{i-1}\>$ such that $\lVert|\tilde\pi_{i-1}\> - |\pi_{i-1}\>\rVert \le O(\epsilon/4m)$. Thus each telescoping sum can be computed with a relative error of $\epsilon/2m$, resulting in a relative error of less than $\epsilon/2$ for the final volume. The probability of success for each iteration is at least $1 - \delta^4 = 1 - \Theta(\epsilon^4/4m^8)$. Thus the probability of success for the whole algorithm is at least $1 - \Theta(\epsilon^4/4m^7)=1-\tilde{O}(\epsilon^{11}/n^{3.5})$, which is greater than $0.7$ for a large enough $n$.

\hd{Cost}
Ignoring the cost of obtaining the affine transformation to round the logconcave densities to be sampled (assuming that all the relevant densities are well rounded), we have from
\lem{mean-estimation-error}, the number of calls to the membership oracle in each iteration of \algo{quantum-volume} is $\tilde{O}(n^{3/2}\log(1/\delta)/\epsilon_1 + n^{3/2}\log(1/\epsilon')) = \tilde{O}(n^2/\epsilon)$. The total number of oracle calls is thus $\tilde{O}(n^{2.5}/\epsilon)$. The argument for correctness above applies for well-rounded logconcave densities. This is ensured by maintaining $\tilde{\Theta}(n)$ states that are used to round the densities in each iteration (\algo{quantum-interlacing}). By \prop{rounding-interlacing}, this procedure uses $\tilde{O}(n^{2.5})$ calls to the membership oracle in each iteration, resulting in a final query complexity of $\tilde{O}(n^{3} + n^{2.5}/\epsilon)$. Since the affine transformation is an $n$-dimensional matrix-vector product, the number of additional arithmetic operations is hence $O(n^{2})\cdot\tilde{O}(n^{3} + n^{2.5}/\epsilon)=\tilde{O}(n^{5} + n^{4.5}/\epsilon)$.
\end{proof}

\subsection{Proofs of lemmas in \sec{quantum-volume-estimation}}\label{sec:quantum-volume-estimation-proof}

We now prove the lemmas in \sec{quantum-volume-estimation} that establish \thm{main}.

\subsubsection{From the pencil construction to the original convex body}\label{sec:pencil-construction-analysis}

Here we prove

\pencil*

\begin{proof}
We follow the same notation in \sec{volume-estimation-review}, i.e., without loss of generality we assume that $r=1$ and denote $D=R/r=R$. Since $R$ and $r$ are both given, $D$ is a known value. The pencil construction is
\begin{align}
\K':=([0,2D]\times \K)\cap\Big\{x\in\R^{n+1}: x_{0}\geq 0,\sum_{i=1}^{n}x_{i}^{2}\leq x_{0}^{2}\Big\}.
\end{align}
By the definition of $D$, for any $(x_{1},\ldots,x_{n})\in\K$ we have $\sum_{i=1}^{n}x_{i}^{2}\leq D^{2}$, so $[D,2D]\times\K\subseteq\K'$. This implies that $D\vol(\K)\leq \vol(\K')\leq 2D\vol(\K)$. In other words, letting $\xi_{\K}:=\frac{\vol(\K')}{2D\vol(\K)}$, we have $0.5 \leq\xi_{\K}\leq 1$.

Classically, we consider a Monte Carlo approach to approximating $\vol(\K)$: we take $k$ (approximately) uniform samples $x_{1},\ldots,x_{k}$ from $[0,2D]\times \K$, and if $k'$ of them are in $\K'$, we return $\frac{k'}{k}\cdot\widetilde{\vol(\K')}$. For each $i\in\range{k}$, $\delta[x_{i}\in\K']$ is a boolean random variable with expectation $\xi_{\K}=\Theta(1)$. Any boolean random variable has variance $O(1)$. Therefore, by Chebyshev's inequality, taking $k=O(1/\epsilon^{2})$ suffices to ensure that
\begin{align}
\Pr\Big[\Big|\frac{k'}{k}-\xi_{\K}\Big|\leq\frac{\epsilon\xi_{\K}}{2}\Big]\geq 0.99.
\end{align}

Quantumly, we adopt the same Monte Carlo approach but we implement two steps using quantum techniques:
\begin{itemize}[leftmargin=*]
  \item We take an approximately uniform sample from $K' = [0,2D]\times \K$ via the quantum hit-and-run walk. To obtain a quantum stationary state, we use a similar idea as in~\cite{dyer1991random} to construct a sequence of $m = \lceil n \log_2(2D)\rceil$ convex bodies. Let $\hat{\rmk}_0 := \rmb_{2}^{n+1}(0,1)$ and $\hat{\rmk}_i := 2^{i/n}\rmb_{2}^{n+1}(0,1) \cap \rmk'$ for $i \in [m]$. As the length of the pencil is $2D$, $\hat{\rmk}_m = \rmk'$. The state $\ket{\pi_0}$ corresponding to $\hat{\rmk}_0$ is easy to prepare. It is straightforward to verify that $\<\pi_i|\pi_{i+1}\> \geq c$ for some constant $c$, as $\vol(\hat{\rmk}_{i+1}) \leq 2 \vol(\hat{\rmk}_{i})$. To utilize the quantum speedup for MCMC framework (\thm{Wocjan-Abeyesinghe}), it remains to lower bound the phase gap of the quantum walk operator for $\hat{\rmk}_i$. It can be shown that the mixing property of the hit-and-run walk in \thm{exp-hit-run-mixes} implies that the phase gap of the quantum walk operator is $\tilde{\Omega}(n^{-1.5})$; see the proof of \lem{pi3-error}. Thus, by \thm{Wocjan-Abeyesinghe}, $\ket{\pi_m}$ can be prepared using $\tilde{O}(n) \cdot \tilde{O}(n^{1.5}) = \tilde{O}(n^{2.5})$ quantum queries to $O_{\K}$.

\item We estimate $\xi_{\K}$ with multiplicative error $\epsilon/2$ using the quantum Chebyshev inequality (\thm{quantum-Chebyshev}) instead of its classical counterpart. This means that $O(1/\epsilon)$ executions of quantum sampling in the first step suffice.
\end{itemize}

Overall, $\tilde{O}(n^{2.5}/\epsilon)$ quantum queries to $O_{\K}$ suffice to ensure that we obtain an estimate of $\xi_{\K}$ within multiplicative error $\epsilon/2$ with success probability at least 0.99. Since \eq{pencil-to-original-approx} ensures that $\widetilde{\vol(\K')}$ estimates $\vol(\K')$ up to multiplicative error $\epsilon/2$ with probability at least 0.7, $\frac{\widetilde{\vol(\K')}}{2D\xi_{\K}}$ estimates $\vol(\K)$ up to multiplicative error $\epsilon/2+\epsilon/2=\epsilon$ with success probability $0.99\cdot 0.7>2/3$.
\end{proof}

\subsubsection{Inner product between stationary states of consecutive steps}\label{sec:proof-inner-product}

We now show that the inner product between stationary states of consecutive steps is at least a constant. More precisely, we have the following:

\innerproduct*

\begin{proof}
  Recall that the stationary distribution $\pi_i$ of step $i$ has density proportional to $e^{-a_i x_0}$ as discussed in \sec{volume-estimation-review}. The corresponding stationary distribution state is $\ket{\pi_i} = \int_{\rmk'}\d x \sqrt{\frac{e^{-a_i x_0}}{Z(a_i)}} \ket{x}$. Lov{\'a}sz and Vempala \cite[Lemma 3.2]{lovasz2006simulated} proved that $a^{n+1}Z(a)$ is log-concave (noting that the dimension of $\rmk'$ is $n+1$). This implies that
  \begin{align}
    \label{eq:sqrt-z}
    \sqrt{a_i^{n+1}Z(a_i)}\sqrt{a_{i+1}^{n+1}Z(a_{i+1})} \leq \left(\frac{a_i + a_{i+1}}{2}\right)^{n+1}Z\left(\frac{a_i + a_{i+1}}{2}\right).
  \end{align}
  Now, we have
  \begin{align}
   \<\pi_i|\pi_{i+1}\> &= \int_{\rmk'}\d x \frac{\exp(-\frac{a_i + a_{i+1}}{2}x_0)}{\sqrt{Z(a_i)}\sqrt{Z(a_{i+1})}} \\
                              &\geq \left(\frac{2\sqrt{a_i}\sqrt{a_{i+1}}}{a_i + a_{i+1}}\right)^{n+1} \frac{\int_{\rmk'}\d x \exp(-\frac{a_i+a_{i+1}}{2}x_0)}{Z\left(\frac{a_i + a_{i+1}}{2}\right)} \\
                              &= \left(\frac{2\sqrt{a_i}\sqrt{a_{i+1}}}{a_i + a_{i+1}}\right)^{n+1} \\
                              &= \left(\frac{2\sqrt{a_i}\sqrt{a_i(1-\frac{1}{\sqrt{n}})}}{a_i + a_i(1-\frac{1}{\sqrt{n}})}\right)^{n+1}= \left(\frac{2\sqrt{1-\frac{1}{\sqrt{n}}}}{2-\frac{1}{\sqrt{n}}}\right)^{n+1},
  \end{align}
  where the inequality follows from~\eq{sqrt-z}. When $n=2$ or $n=3$, the above inequality holds. When $n\geq 4$, to lower bound the above quantity, we use the fact that $\sqrt{1-1/\sqrt{n}} \geq 1-\frac{1}{2\sqrt{n}}-\frac{1}{2n}$. Hence, for $n \geq 2$ we have
  \begin{align*}
     \<\pi_i|\pi_{i+1}\> &\geq \left(\frac{2-\frac{1}{\sqrt{n}} - \frac{1}{n}}{2-\frac{1}{\sqrt{n}}}\right)^{n+1}=\left(1-\frac{\frac{1}{n}}{2-\frac{1}{\sqrt{n}}}\right)^{n+1}\geq \left(1-\frac{1}{(2-\frac{1}{\sqrt{2}})n}\right)^{n+1}>\frac{1}{3}
  \end{align*}
  as claimed.
\end{proof}

\subsubsection{Chebyshev cooling and nondestructive mean estimation}\label{sec:nondestructive}

Now we briefly review the classical framework for Chebyshev cooling and discuss how to adapt it to quantum algorithms. Suppose we want to compute the expectation of a product
\begin{align}
  V = \prod_{i=1}^{m} V_i
\end{align}
of independent random variables. The following theorem of Dyer and Frieze~\cite{dyer1991computing} upper bounds the number of samples from the $V_i$ that suffices to estimate $\E[V]$ with bounded relative error.

\begin{proposition}[{\cite[Section 4.1]{dyer1991computing}}]\label{prop:Chebyshev-classical}
  Let $V_1, \ldots, V_m$ be independent random variables such that $\frac{\E[V_i^2]}{\E[V_i]^2} \leq B$ for all $i \in [m]$. Let $X_j^{(1)}, \ldots, X_j^{(k)}$ be $k$ samples of $V_j$ for $j \in [m]$, and define $\overline{X}_j=\frac{1}{k}\sum_{\ell=1}^{k}X_j^{(\ell)}$. Let $V=\prod_{j=1}^{m}V_j$ and $X=\prod_{j=1}^m\overline{X}_j$. Then, taking $k = 16Bm/\epsilon^2$ ensures that
  \begin{align}
    \Pr\big[(1-\epsilon)\E[V] \leq X \leq (1+\epsilon)\E[V]\big] \geq \frac{3}{4}.
  \end{align}
\end{proposition}
With standard techniques, the probability can be boosted to $1-\delta$ with a $\log(1/\delta)$ overhead.

In applications such as volume estimation~\cite{LV06} and estimating partition functions~\cite{SVV009}, the samples are produced by a random walk. Let the mixing time for each random walk be at most $T$. Then the total complexity for estimating $\E[V]$ with success probability $1-\delta$ is $\tilde{O}(TBm\log(1/\delta)/\epsilon^2)$. Replacing the random walk with a quantum walk can potentially improve the mixing time; see \sec{quantum-MCMC-literature} for relevant literature. In particular, Montanaro~\cite{montanaro2015quantum} proposed a quantum algorithm for the simulated annealing framework with complexity $\tilde{O}(TBm\log(1/\delta)/\epsilon)$, which has a quadratic improvement in precision. Note that the dependence on $T$ was not improved, as multiple copies of quantum states were prepared for the mean estimation (which uses measurements). In this paper, we use the quantum Chebyshev inequality (see \thm{quantum-Chebyshev}) to estimate the expectation of $V_i$ in a nondestructive manner which, together with \thm{Wocjan-Abeyesinghe}, achieves complexity $\tilde{O}(\sqrt{T}Bm\log(1/\delta)/\epsilon)$.

Recall that the random variables $V_i$ (determined by the cooling schedule) satisfy Eq.~\eqn{LV06-Lemma4.1}.
The following lemma uses this property of the simulated annealing procedure to show that the quantum Chebyshev inequality can be used to estimate the mean of $V_i$ on the distribution $\pi_i$, which gives an estimate of the ratio $\frac{Z(a_{i+1})}{Z(a_i)}$ in the volume estimation algorithm. We first show that our random variables can be made to satisfy the conditions of \thm{quantum-Chebyshev}, and then we outline how the corresponding circuit can be implemented. A detailed error analysis is deferred to \sec{proof-error-analysis}.  To make the mean estimation nondestructive, we use the following theorem of Harrow and Wei.

\begin{theorem}[{\cite[Theorem 6]{harrow2019adaptive}}]
  \label{thm:hw-nondestructive}
  Given state $|\psi\>$ and reflections $R_{\psi} = 2|\psi\>\<\psi| - I$ and $R = 2P - I$, and any $\eta > 0$, there exists a quantum algorithm that outputs $\tilde{a}$, an approximation to $a = \<\psi|P|\psi\>$, so that
  \begin{align}
    \label{eq:ndmean}
    |\tilde{a} - a| \le 2\pi \frac{a(1-a)}{M} + \frac{\pi^{2}}{M^{2}}
  \end{align}
  with probability at least $1 - \eta$ and $O(\log(1/\eta)M)$ uses of $R_{\psi}$ and $R$. Morover the algorithm restores the state $\psi$ with probability at least $1 - \eta$.
\end{theorem}

\Chebyshev*

\begin{proof}
We apply the quantum Chebyshev inequality (\thm{quantum-Chebyshev}). For the random variables $V_i$, we let $\mu_i$ denote their mean and $\sigma_i^2$ their variance. From \eqn{LV06-Lemma4.1}, $\sqrt{\sigma_i^2 + \mu_i^2}/\mu_i \le \sqrt{8} < 3$.
For a small constant $c$, we use $\log(1/\delta)/c^2$ copies of $|\pi_{i-1}\>$ to create copies of $|\pi_{i}\>$ using $\pi/3$-amplitude amplification. We now use a quantum circuit that given $|x\>|0\>$ computes $|x\>|e^{a_i x_0 - a_{i-1}x_0}\>$, and then apply a circuit $U_{\median}$ that computes the median of all the ancilla registers:
\begin{align}\label{eqn:U-median}
U_{\median}|0\>|a_{1}\>\cdots|a_{s}\>=|\textrm{median}\{a_{1},\ldots,a_{s}\}\>|a_{1}\>\cdots|a_{s}\>.
\end{align}
By the classical Chebyshev inequality, we measure $\hat\mu_i$ such that $|\hat\mu_i - \mu_i| \le c\mu_i$ with probability at least $1 - \delta$. Thus the probability that $\hat\mu_i/(1-c) < \mu$ is less than $\delta$. Taking $H = \hat\mu_i/(1-c)$, our variables satisfy the conditions of the quantum Chebyshev inequality. In order to output an estimate of the mean with relative error at most $\epsilon$, the quantum Chebyshev inequality now requires $\tilde{O}(\log(1/\delta)/\epsilon)$ calls to a sampler for the state $|\pi_i\>$, which we construct using $\pi/3$-amplitude amplification on copies of $|\pi_{i-1}\>$. By the union bound, the probability of failure of the whole procedure is $O(\delta)$.

To be more specific, we replace $U|0\>$ in \textsf{BasicEst} (\algo{BasicEst}) by $U_{i-1,l}|\pi_{i-1}\>$ (where $U_{i-1,l}$ is the circuit transforming the $l^{th}$ copy of $|\pi_{i-1}\>$ to $|\pi_{i}\>$), and replace $\mathcal{Q}$ by $-U_{i-1,l}(\Pi_{i-1}-\Pi_{i-1}^{\perp})U_{i-1,l}^{\dagger}(\Pi_{i}-\Pi_{i}^{\perp})$ (where $\Pi_{i} = |\pi_i\>\<\pi_i|$ and $\Pi_{i}^{\perp} = I - \Pi_i$ for all $i\in\range{m}$). The quantum circuit for nondestructive \textsf{BasicEst} is shown in \fig{SA-block-nondes}. Here, we run $\Theta(\log(1/\delta))$ executions of amplitude estimation (\fig{AmpEst}) in parallel. Note that by \eqn{AmpEst-two-angles}, each amplitude estimation returns a state $\frac{e^{i\theta_{p}}}{\sqrt{2}}|\tilde{\theta}_{p}\>-\frac{e^{-i\theta_{p}}}{\sqrt{2}}|-\tilde{\theta}_{p}\>$. We use an ancilla register and apply the unitary
\begin{align}\label{eqn:sine-square}
U_{\sinsq}|\theta\>|0\>:=|\theta\>|\sin^{2}\theta\>;
\end{align}
because $\sin^{2}(\tilde{\theta}_{p})=\sin^{2}(-\tilde{\theta}_{p})=\tilde{p}$, the ancilla register becomes $|\tilde{p}\>$, where $\tilde{p}$ estimates $p$ well as claimed in \thm{AmpEst}. We then take the median of such $\Theta(\log(1/\delta))$ executions using \eqn{U-median}, and finally run the inverse of $U_{\sinsq}$ gates and amplitude estimations. The correctness follows from the proof of \thm{quantum-Chebyshev} in~\cite{hamoudi2019Chebyshev}.

To assure non-destructiveness, we replace every application of Quantum Amplitude Estimation with the Nondestructive Mean Estimation as in \thm{hw-nondestructive}. The resulting guarantees on the error are the same as with the original amplitude estimation algorithm. To ensure an overall success probability of $1 - O(\delta)$, it suffices to perform each instance of Nondestructive Amplitude Estimation with success probability $1 - \tilde{O}(\delta)$. Note that since we estimate an unweighted mean, $2P - I$ with $P = H|0\>\<0|H$ can be implemented as $HR_{0}H$ where $R_{0}$ is a reflection around the $|0\>$ state. Finally, we show in \cor{rotation-gates} that $R_{\pi_{i}}$ (the reflection around $|\pi_{i}\>$) can be implemented using the same number of oracle calls and gates as that required to sample $\pi_{i}$ (up to polylogarithmic factors).
\end{proof}

A detailed error analysis is given in the next subsection (see \lem{err-non-destructive}).
\begin{figure}[ht]
\resizebox{1\linewidth}{!}{
\Qcircuit @C=1.2em @R=0.9em {
   & & \qw & \qw & \qw & \qw & \gate{U_{\median}} & \qw & \qw & \qw & \qw & \qw \\
   |0\> & & \qw & \qw & \qw & \gate{U_{\sinsq}} & \dctrl{9}{\ar @{-} [-1,0]} & \gate{U_{\sinsq}^{\dagger}} & \qw & \qw & \qw & \qw & |0\> \\
   |0\> & & \multigate{2}{\textsf{QFT}} & \ctrl{0} {\ar @{-}+<0em,-0.7em>} & \multigate{2}{\textsf{QFT}^{\dagger}} & \ctrl{-1} {\ar @{-}+<0em,-0.7em>} & \qw & \ctrl{-1} {\ar @{-}+<0em,-0.7em>} & \multigate{2}{\textsf{QFT}^{\dagger}} & \ctrl{0} {\ar @{-}+<0em,-0.7em>} & \multigate{2}{\textsf{QFT}} & \qw & |0\> \\
   \raisebox{6pt}{\vdots} & & & \raisebox{6pt}{\vdots} & & \raisebox{6pt}{\vdots} & & \raisebox{6pt}{\vdots} & & \raisebox{6pt}{\vdots} \\
   |0\> & & \ghost{\textsf{QFT}} & \ctrl{1}{\ar @{-}+<0em,0.7em>} & \ghost{\textsf{QFT}^{\dagger}} & \ctrl{0}{\ar @{-}+<0em,0.7em>} & \qw & \ctrl{0}{\ar @{-}+<0em,0.7em>} & \ghost{\textsf{QFT}^{\dagger}} & \ctrl{1}{\ar @{-}+<0em,0.7em>} & \ghost{\textsf{QFT}} & \qw & |0\> \\
   |\pi_{i}\> & & \qw & \multigate{2}{\mathcal{Q}} & \qw & \qw & \qw & \qw & \qw & \multigate{2}{\mathcal{Q}^{\dagger}} & \qw & \qw & |\pi_{i}\> \\
   \raisebox{6pt}{\vdots} & & & & & & & & & & & & \raisebox{6pt}{\vdots} \\
   |\pi_{i}\> & & \qw & \ghost{\mathcal{Q}} & \qw & \qw & \qw & \qw & \qw & \ghost{\mathcal{Q}^{\dagger}} & \qw & \qw & |\pi_{i}\> \\
      \raisebox{6pt}{\vdots} & & & & & & & & & & & & \raisebox{6pt}{\vdots} \\
      \raisebox{2pt}{\vdots} & & & & & & & & & & & & \raisebox{2pt}{\vdots} \\
   |0\> & & \qw & \qw & \qw & \gate{U_{\sinsq}} & \ctrl{0} & \gate{U_{\sinsq}^{\dagger}} & \qw & \qw & \qw & \qw & |0\> \\
   |0\> & & \multigate{2}{\textsf{QFT}} & \ctrl{0} {\ar @{-}+<0em,-0.7em>} & \multigate{2}{\textsf{QFT}^{\dagger}} & \ctrl{-1} {\ar @{-}+<0em,-0.7em>} & \qw & \ctrl{-1} {\ar @{-}+<0em,-0.7em>} & \multigate{2}{\textsf{QFT}^{\dagger}} & \ctrl{0} {\ar @{-}+<0em,-0.7em>} & \multigate{2}{\textsf{QFT}} & \qw & |0\> \\
   \raisebox{6pt}{\vdots} & & & \raisebox{6pt}{\vdots} & & \raisebox{6pt}{\vdots} & & \raisebox{6pt}{\vdots} & & \raisebox{6pt}{\vdots} \\
   |0\> & & \ghost{\textsf{QFT}} & \ctrl{1}{\ar @{-}+<0em,0.7em>} & \ghost{\textsf{QFT}^{\dagger}} & \ctrl{0}{\ar @{-}+<0em,0.7em>} & \qw & \ctrl{0}{\ar @{-}+<0em,0.7em>} & \ghost{\textsf{QFT}^{\dagger}} & \ctrl{1}{\ar @{-}+<0em,0.7em>} & \ghost{\textsf{QFT}} & \qw & |0\> \\
   |\pi_{i}\> & & \qw & \multigate{2}{\mathcal{Q}} & \qw & \qw & \qw & \qw & \qw & \multigate{2}{\mathcal{Q}^{\dagger}} & \qw & \qw & |\pi_{i}\> \\
   \raisebox{6pt}{\vdots} & & & & & & & & & & & & \raisebox{6pt}{\vdots} \\
   |\pi_{i}\> & & \qw & \ghost{\mathcal{Q}} & \qw & \qw & \qw & \qw & \qw & \ghost{\mathcal{Q}^{\dagger}} & \qw & \qw & |\pi_{i}\>
}
}
\caption{The quantum circuit for nondestructive \textsf{BasicEst}.}
\label{fig:SA-block-nondes}
\end{figure}

\subsubsection{Error analysis}\label{sec:proof-error-analysis}
In this section, we analyze the error incurred by both the quantum Chebyshev inequality (\lin{quantum-sample-SA}) and $\pi/3$-amplitude amplification (\lin{quantum-pi/3-AM-AM}) in \algo{quantum-volume}.

\error*

We first show that $\pi/3$-amplitude amplification can be used to rotate $|\pi_{i}\>$ into $|\pi_{i-1}\>$ with error $\epsilon'$ using $\tilde{O}(\log(1/\epsilon))$ oracle calls.
This procedure is used as a subroutine in a mean estimation circuit that estimates the mean of the random variable $V_{i}$ using multiple approximate copies of $|\pi_{i-1}\>$. We ensure that the measurement probabilities are highly peaked so that the state is not disturbed very much. Finally $\pi/3$-amplitude estimation is used again to rotate the approximate copies of the state $|\pi_{i-1}\>$ to approximate copies of the state $|\pi_{i}\>$.

\paragraph{Large effective spectral gap}
Consider an ergodic, reversible Markov chain $(\Omega,p)$ with transition matrix $P$ and a unique stationary distribution with density $\pi$. Let $a(x)$ be a probability measure over $\Omega$ such that the Markov chain mixes to its stationary distribution with a corresponding probability density $\pi(x)$ within a total variation distance of $\epsilon$ within $t$ steps. Further let $a(x)$ be a warm start for $\pi(x)$.  From the definition of the transition matrix $P(x,y) = \<x|P|y\> = p_{y \to x}$.

 The discriminant matrix $D$ defined in \eqn{discriminant} is related to the transition matrix as $P = D_\pi D D_\pi^{-1}$, as shown in \eq{DPsimilar}. For a hit-and-run walk, the transition matrix $P$ represents a convolution with an $L_{2}$ normalized function (corresponding to the square root of the density $p_{x \to y}$). Bounded subsets of $L_{2}(\Omega)$ are therefore mapped by $P$ to other bounded subsets, and hence $P$ is compact. Since $D$ is connected to $P$ by a similarity relation, $D$ is a compact Hermitian operator over $L_2(\Omega)$ and thus has a countable set of real eigenvalues $\lambda_i$ and corresponding orthonormal eigenvectors (eigenfunctions) $v_i \in L_2(\Omega)$. Orthonormality implies that $\int_{\Omega}v_i(x)v_j(x)\,\d x =\,\delta_{i,j}$. Notice that
\begin{align}
  \label{eq:eig-of-P}
  PD_\pi v_i = D_\pi D(v_i) = \lambda_j D_\pi v_i;
\end{align}
thus $f_i = D_\pi v_i$ is an eigenvector of $P'$ with eigenvalue $\lambda_i$. The eigenvectors $f_i$ may not be orthogonal under the standard inner product on $L_2(\Omega)$. However, we can define an inner product
\begin{align}
  \label{eq:inner-prod-pi}
  \< f, g \>_\pi := \< D_\pi^{-1} f, D_\pi^{-1} g \> = \int_\Omega \frac{f(x) g(x)} {\pi(x)}\,\d x
\end{align}
over the space $L_2(\Omega)$. It is easy to see that $\< f_i, f_j \>_\pi = \< v_i, v_j \> =\,\delta_{i,j}$. A corresponding norm can be defined as $\lVert f \rVert_{\pi} = \langle f, f \rangle_{\pi}$.

It can be verified that $\sqrt{\pi}(x)$ is an eigenfunction of $D$ with eigenvalue $1$. Thus the stationary state $\pi(x)$ is an eigenfunction of the transition operator $P$ with eigenvalue $1$. Since $P$ is stochastic, this is the leading eigenvalue. The eigenvalues of $P$ are thus $1,\lambda_1,\lambda_2,\dots$ with corresponding eigenfunctions $\pi(x),f_2(x),f_3(x),\dots$. From the orthonormality of the $f$ under $\<\cdot,\cdot\>_\pi$, for any function $g$ in $L_2(\Omega)$ we have
\begin{align}
  \label{eq:eigenbasis-f}
  g = \sum_{i=1}^{\infty} \< g, f_i \>_\pi f_i &= \<g, \pi \>_\pi + \sum_{i=2}^{\infty}\< g, f_i \>_\pi f_i \\
  &= \left(\int_\Omega \frac{g(x)\pi(x)}{\pi(x)}\,\d x\right)\pi + \sum_{i=2}^{\infty}\< g, f_i \>_\pi f_i \\
  &= \left(\int_\Omega g(x)\,\d x\right)\pi + \sum_{i=2}^{\infty}\< g, f_i \>_\pi f_i.
\end{align}
Since $a$ is a probability density,  $a = \pi + \sum_{i=2}^{\infty}\< a, f_i \>_\pi f_i$. After $t$ steps of the Markov chain $M$ on $a$ we obtain the state $P^t a = \pi + \sum_{i=2}^{\infty}\lambda_i^t\< a, f_i \>_\pi f_i$. Since the walk mixes to total variation distance $\epsilon$ we have $\lVert P^ta - \pi \rVert_1 \le \epsilon$, and further since $a$ is a warm start $\lVert P^ta - \pi \rVert_{\pi}$. Consequently, $\lVert \sum_{i=2}^{\infty}\lambda_i^t \< a, f_i \>_\pi f_i \rVert_{\pi} \le \epsilon$ and from the orthonormality of $f$, $\<a,f_i\>_\pi\lambda_i^t \le \epsilon$.
If $1 > \lambda_i \ge 1 - \frac{1}{\Omega(t)}$ then $\lambda_i^t=\Omega(1)$ and $\<a,f_i\>_\pi=\O(\epsilon)$.

The above analysis indicates that if a probability density $a$ (that is a warm start) mixes in $t$ steps under a Markov chain $(\Omega,p)$, then it has small overlap with each of the ``bad'' eigenfunctions (with spectral gap less than $\frac{1}{\Omega(t)}$). Thus $P$ effectively has a large spectral gap when it acts on $a$.

Corresponding to $a$, consider the quantum states
\begin{align}
|a\>:= \int_\Omega\sqrt{a(x)} |x\>\,\d x, \qquad |\phi_a\>:= \int_\Omega\int_\Omega\sqrt{a_xp_{x \to y}}|x\>|y\>\,\d x\,\d y.
\end{align}
For an eigenvector $v_i$ of $D$ (with eigenvalue $\lambda_i$),
define the state $|v_i\> := \int_\Omega v_i(x)\,\d x = \int_\Omega \frac{f_i(x)}{\sqrt{\pi(x)}}\,\d x$. Then the walk operator $W$ has the corresponding eigenvector $|u_i\> = \big(I-(\lambda_i-i\sqrt{1-\lambda_i^{2}})S\big)T|v_i\>$ following the proof of \thm{quantum-walk-main}. Let $C_i := \lambda_i-i\sqrt{1-\lambda_{i}^{2}}$; then $\< \phi_a | u_i \> = \<\phi_a|T|v_i\> - C_i\<\phi_a|ST|u_i\>$. Furthermore,
\begin{align}
  \label{eq:ip-term1}
   \<\phi_a|T|v_i\> = \<a|v_i\> = \int_\Omega \frac{\sqrt{a(x)}f_i(x)}{\sqrt{\pi(x)}}\,\d x,
\end{align}
and
\begin{align}
  \label{eq:ip-term2}
  \<\phi_a|ST|v_i\> &= \left(\int_\Omega\sqrt{a_xp_{x \to y}}\<y|\<x|\right)\left(\int_\Omega\sqrt{v_{x'}p_{x' \to y'}}|x'\>|y'\>\right) \\
                      &= \int_\Omega\sqrt{a_x}\Big(\int_\Omega\sqrt{p_{x \to y}p_{y \to x}v_i(y)}\,\d y\Big)\,\d x \\
                    &= \int_\Omega\sqrt{a_x}(Dv_i)(x)\,\d x \\
                    &= \lambda_i \int_\Omega\sqrt{a_x}v_i(x)\,\d x \\
                    &= \lambda_i \<a|v_i\>.
\end{align}
We have $\< \phi_a | u_i \> = (1 - \lambda_iC_i)\<a|v_i\>$ and therefore
\begin{align}
  \label{eq:magnitude-in-terms-of-integral}
  |\< \phi_a | u_i \>| = (1 - \lambda_iC_i)\<a|v_i\> = \sqrt{(1 - \lambda_i^2)^2 + (1 - \lambda_i)^2}\<a|v_i\> \le 2|\<a|v_i\>|.
\end{align}
In addition,
\begin{align}
  \label{eq:ratio-bound}
  \<a|v_i\> = \int_\Omega \frac{\sqrt{a(x)}f_i(x)}{\sqrt{\pi(x)}}\,\d x = \int_\Omega \sqrt{\frac{\pi(x)}{a(x)}} \frac{a(x)f_i(x)}{\pi(x)}\,\d x.
\end{align}

The above discussion establishes the following proposition indicating that if a distribution with density $a(x)$ mixes fast and the stationary distribution with density $\pi(x)$ has a bounded $L_{2}$-norm with respect to $a(x)$, then the quantum walk operator $W$ acting on the subspace spanned by $|\pi\>$ and $|a\>$ has a large effective spectral gap.

\begin{proposition}
  \label{prop:approx-warmness-quantum}
  Let $M = (\Omega,p)$ be an ergodic reversible Markov chain with a transition operator $P$ and unique stationary state with a corresponding density $\pi \in L_2(\Omega)$. Let $\{(\lambda_i,f_i)\}$ be the set of eigenvalues and eigenfunctions of $P$, and $|u_i\>$ be the eigenvectors of the corresponding quantum walk operator $W$. Let $a \in L_2(\Omega)$ be a probability density that is a warm start for $\pi$ and mixes up to total variation distance $\epsilon$ in $t$ steps of $M$.  Furthermore, assume that
  $\int_\Omega \frac{\pi(x)}{a(x)}\,\pi(x) \d x \le c$ for some constant $c$.
  Define
  \begin{align}
    |a\> &= \int_\Omega\sqrt{a(x)}|x\>\,\d x; \\
    |\phi_a\> &= \int_\Omega\sqrt{a(x)}\int_\Omega\sqrt{p_{x \to y}}|x\>|y\>\,\d x\,\d y.
  \end{align}
  Then $\<\phi_a|u_i\>=O(\epsilon^{1/2})$ for all $i$ such that $1 > \lambda_i \ge 1 - \frac{1}{\Omega(t)}$.
\end{proposition}

\begin{proof}
  Define $S = \lbrace x | \frac{\pi(x)}{a(x)} \ge \sqrt{\frac{c}{\epsilon}} \rbrace$. Because $\int_\Omega\frac{\pi(x)^2}{a(x)^2}\,a(x)\d x=\int_\Omega \frac{\pi(x)}{a(x)}\,\pi(x)\d x\leq c$, Markov's inequality implies that $\int_S a(x)\d x \le \epsilon$.

  We now define the quantum state $|a'\>$ such that $\<x|a'\> = \<x|a\>$ if $x \notin S$ and $\<a|x'\> = 0$ otherwise, and $|\phi_{a'}\> = T|a'\>$. Then
  \begin{align}
    \label{eq:diff-a-aprime}
    \lVert |\phi_a\> - |\phi_{a'}\> \rVert = \Big\lVert \int_S \sqrt{a(x)} T|x\>\,\d x \Big\rVert = \sqrt{\int_\Omega a(x)\,\d x} = \sqrt{\epsilon}.
  \end{align}
  From \eq{magnitude-in-terms-of-integral} and \eq{ratio-bound}, if $1 > \lambda_i \ge 1 - \frac{1}{O(t)}$, then
  \begin{align}
    \label{eq:bound-a-prime}
    |\<\phi_{a'}|u_i\>| &\le \Big\lvert{2\int_\Omega \sqrt{\frac{\pi(x)}{a(x)}} \frac{a(x)f_i(x)}{\pi(x)}\,\d x}\Big\rvert \le \frac{2c^{1/4}\<a,f_i\>_{\pi}}{\epsilon^{1/4}} \le 2c^{1/4}\epsilon^{3/4}.
  \end{align}
  Finally,
    \begin{align}
  \<\phi_a|u_i\> = \<\phi_{a'}|u_i\> + \<\phi_a-\phi_{a'}|u_i\> \le 2c^{1/4}\epsilon^{3/4} + \sqrt{\epsilon} =\O(\sqrt{\epsilon})
  \end{align}
  if $1 > \lambda_i \ge 1 - \frac{1}{\Omega(t)}$. Hence the result follows.
\end{proof}

\paragraph{Warmness of $\pi_{i+1}$ with respect to $\pi_i$}
We show that density $\pi_i$ mixes to $\pi_{i+1}$ under the walk $W_{i+1}$ and vice versa. To apply \thm{exp-hit-run-mixes}, we show that the two distributions are warm with respect to each other.

The $L_2$-norm of a distribution with density $\pi_1 \in L_2(\Omega)$ with respect to another with density $\pi_2 \in L_2(\Omega)$ is defined as
\begin{align}\label{eq:defn-l2-norm}
\lVert \pi_1 / \pi_2 \rVert = \E_{X \sim \pi_1}\left[ \frac{\pi_1(X)}{\pi_2(X)} \right] = \int_\Omega \frac{\pi_1(x)}{\pi_2(x)} \, \pi_1(x)\,\d x.
\end{align}
A density $\pi_1 \in L_2(\Omega)$ is said to be a warm start for $\pi_2 \in L_2(\Omega)$ if the $L_2$-norm $\lVert \pi_1 / \pi_2 \rVert$ is bounded by a constant.

\begin{lemma}[{\cite[Lemma 4.4]{lovasz2006simulated}}]
  \label{lem:l2warm-pi-plus1}
  The $L_2$-norm of the probability distribution with density $\pi_{i}=\frac{e^{-a_i x_0}}{Z(a_i)}$ with respect to that with density $\pi_{i+1} = \frac{e^{-a_{i+1} x_0}}{Z(a_{i+1})}$ is at most 8.
\end{lemma}

\begin{lemma}
  \label{lem:l2warm-pi-plus1-reverse}
  The $L_2$-norm of the probability distribution with density $\pi_{i+1} = \frac{e^{-a_{i+1} x_0}}{Z(a_{i+1})}$ with respect to that with density $\pi_{i}=\frac{e^{-a_i x_0}}{Z(a_i)}$ is at most $e$.
\end{lemma}
\begin{proof}
  Since $a^nZ(a)$ is a log-concave function \cite[Lemma 3.2]{lovasz2006simulated}, we have
\begin{align}
    \E_{X \sim \pi_{i+1}}\left[\frac{\pi_{i+1}(X)}{\pi_i(X)}\right] &= \frac{\int_{\K'}e^{(a_{i} - a_{i+1})x_0}e^{-a_{i+1}x_0}\d x \int_{\K'}e^{-a_{i}x_0}\d x}{\int_{\K'}e^{-a_{i+1}x_0}\d x \int_{\K'}e^{-a_{i+1}x_0}\d x} \nonumber \\
    &= \frac{Z(2a_{i+1}-a_{i})Z(a_{i})}{Z(a_{i+1})^2}\qquad\quad\text{(definition of $Z$)}\\
    &\le \left(\frac{a_{i+1}^2}{a_i(2a_{i+1}-a_{i})}\right)^{n}\qquad\text{(logconcavity of $a^nZ(a)$)} \\
    &\le \left(\frac{\big(1 - \frac{1}{\sqrt{n}}\big)^2}{1 - \frac{2}{\sqrt{n}}}\right)^{n}\qquad\quad\ \text{(definition of $a_i$)}\\
    &\le \left(1 + \frac{2}{n}\right)^{n} < e^{2}, \label{eqn:last-line-lemma4.13}
  \end{align}
where \eqn{last-line-lemma4.13} holds because $1+\frac{1}{n}-\frac{2}{\sqrt{n}}\leq (1+\frac{2}{n})(1-\frac{2}{\sqrt{n}})$ as long as $n\geq 16$.
\end{proof}

\paragraph{Error analysis of $\pi/3$-amplitude amplification}
Consider a simulated annealing procedure that follows a sequence of Markov chains $M_1,M_2,\dots$ with stationary states $\mu_1,\mu_2,\dots$. Consider an alternate walk operator (used in \cite{wocjan2008speedup}) of the form
\begin{equation}
  \label{eqn:mod-walk-operator}
  W_i' = U_i^\dagger S U_i R_{\mathcal{A}} U_i^\dagger S U_i R_{\mathcal{A}}
\end{equation}
where $R_{\mathcal{A}}$ denotes the reflection about the subspace $\mathcal{A} := \spn \{|x\>|0\> : x \in \K\}$ and $S$ is the swap operator. We have $U_i|x\>|0\> = \int_{y \in \K}\sqrt{p^{(i)}_{x \to y}}|x\>|y\>\,\d y$ where $p^{(i)}$ is the transition probability corresponding to the $i^{\text{th}}$ chain.

The $W_i'$ operator is related to the walk operator $W_i = S(2\Pi_i - I)$ via conjugation by $U_i$, i.e., $W_i = U_iW_i'U_i^\dagger$. Thus $W_i'$ has the same eigenvalues as $W_i$, and if $|u_j\>$ is an eigenvector of $W_i$ with eigenvalue $\lambda_{j}$, then $|v\> = U_i^\dagger|u_{j}\>$ is an eigenvector of $W_i'$ with the same eigenvalue $\lambda_j$. For any classical distribution $f$, we define $|f\> = \int_\Omega \sqrt{f(x)}|x\>\,\d x$ and $|\phi_f^{(i)}\> = \int_\Omega \sqrt{f(x)}|x\>\int_\Omega\sqrt{p_{x \to y}^{(i)}}|y\>\,\d y\,\d x$. Since $|\phi_{\pi_i}^{(i)}\>$ is a stationary state of $W_i$ with eigenvalue $1$, it follows that $|\pi_i\>|0\>$ is an eigenvalue of $W_i$ with eigenvalue $1$.

In each stage of the volume estimation algorithm, we sample from a state with density $\pi_i(x) = \frac{e^{- a_i x_0}}{Z(a_i)}$ . Each such distribution is the stationary state of a hit-and-run walk with the corresponding target density. Thus the corresponding state $|\pi_i\>$ is the stationary state of the corresponding walk operators $W_i$ and $W_i'$. Both $W_i$ and $W_{i'}$ can be implemented using a constant number of $U_i$ gates.

From \lem{inner-product}, we know that the inner product $\<\pi_{i}|\pi_{i+1}\>$ between the states at any stage of the algorithm is at least $\frac{1}{3}$. This implies that the inner product between $|\pi_i\>|0\>$ and $|\pi_{i+1}\>|0\>$ is also at least $\frac{1}{3}$. In the following we abuse notation by sometimes writing only $|\pi_i\>$ to denote $|\pi_i\>|0\>$, as it is easy to tell from context whether the ancilla register should be present.

\lem{pi3-amplification} in \sec{quantum-MCMC} indicates that $\pi/3$-amplitude amplification can be used to rotate the state $|\pi_i\>$ to $|\pi_{i+1}\>$ if we can implement the rotation unitaries
\begin{align*}
  R_i = \omega|\pi_i\>\<\pi_i| + \left(I - |\pi_i\>\<\pi_i|\right)\quad\text{and}\quad R_{i+1} = \omega|\pi_{1+1}\>\<\pi_{i+1}| + \left(I - |\pi_{i+1}\>\<\pi_{i+1}|\right).
\end{align*}
To implement these rotations we use the fact that $\pi_{i}$ and $\pi_{i+1}$ are the eigenvectors of the operators $W'_i$ and $W'_{i+1}$ with eigenvalue 1, respectively. We show the following lemmas which are adapted variants of Lemma 2 and Corollary 2 in \cite{wocjan2008speedup}:

\begin{lemma}\label{lem:phase-estimation}
  Let $W$ be a unitary operator with a unique leading eigenvector $|\psi_0\>$ with eigenvalue $1$. Denote the remaining eigenvectors by $|\psi_j\>$ with corresponding eigenvalues $e^{2\pi i \xi_j}$. For any $\Delta \in (0,1]$ and $\epsilon_2 < 1/2$, define $a := \log(1/\Delta)$ and $c := \log(1/\sqrt{\epsilon_{2}})$. There exists a quantum circuit $V$ that uses $ac$ ancilla qubits and invokes the controlled-$W$ gate $2^{a} c$ times such that
  \begin{align}
    V|\psi_0\>|0\>^{\otimes ac} = |\psi_0\>|0\>^{\otimes ac}
  \end{align}
  and
  \begin{align}
    V|\psi_j\>|0\>^{\otimes ac} = \sqrt{1 - \epsilon_2(j)}|\psi_j\>|\chi_j\> + \sqrt{\epsilon_2(j)}|\psi_j\>|0\>^{\otimes ac}
  \end{align}
  where $|\chi_j\>$ is orthogonal to $|0\>^{\otimes ac}$ for all $|\psi_j\>$ such that $\xi_j \ge\,\Delta$, and $\epsilon_2(j)\leq \epsilon_2$ for all $j$.
\end{lemma}
\begin{proof}
  Consider a quantum phase estimation circuit $U$ with $a$ ancilla qubits that invokes the controlled-$W$ gate $2^a$ times (see \fig{quantum-phase}). The phase estimation circuit first creates an equal superposition over $a$ ancilla qubits using Hadamard gates. For $k = 0,\dots,a-1$ we apply a controlled-$W^k$ operator to the input register, controlled by the $(a-k)^{\text{th}}$ register. Finally the inverse quantum Fourier transform is applied on the ancilla registers. Then
  \begin{align}
    U|\psi_j\>|0\>^{\otimes a} &= |\psi_j\> \otimes \textsf{QFT}^\dagger\left(\frac{1}{\sqrt{2^a}}\sum_{m=0}^{2^a -1}e^{2\pi i m \xi_j}|m\>\right) \\
    &= |\psi_j\> \otimes \frac{1}{2^a}\sum_{m,m'=0}^{2^a-1}e^{2\pi i m \left(\xi_j - m'/2^a\right)}|m'\>.
  \end{align}
  The amplitude corresponding to $|0\>$ on the ancilla registers is
  \begin{align}
    a_{j,0} := \frac{1}{2^a}\sum_{m=0}^{2^a-1}e^{2\pi i m \xi_j} = \frac{1 - e^{2\pi i 2^a\xi_j}}{2^a(1 - e^{2\pi i \xi_j})}
  \end{align}
for $j \ne 0$, and $a_{0,0}=1$. If $j \ne 0$ then
  \begin{align}
    |a_{j,0}| = \Big\lvert\frac{1 - e^{2\pi i 2^a\xi_j}}{2^a(1 - e^{2\pi i \xi_j})} \Big\rvert \le  \Big\lvert\frac{1}{2^{a-1}(1 - e^{2\pi i \xi_j})} \Big\rvert \le \frac{1}{2^{a+1}|\xi_j|}.
  \end{align}
  Thus $|a_{j,0}| \le \frac{1}{2}$ if $\xi_j \ge\,\Delta$. Using $c$ copies of the circuit (resulting in $ac$ ancilla registers and $2^{a}c$ controlled-$W$ gates), the amplitude for $0$ in all the ancilla registers if $\xi_j \ge\,\Delta$ is at most $\frac{1}{2^c} = \sqrt{\epsilon}$.
\end{proof}

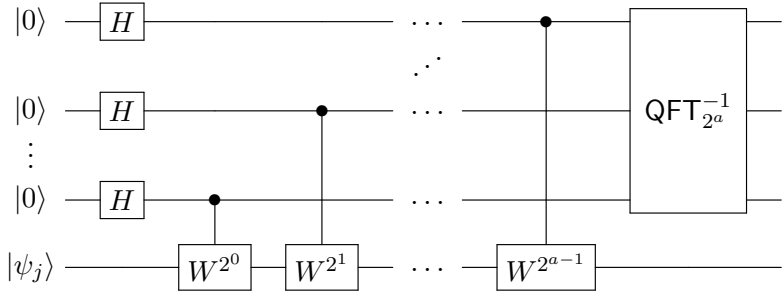
\begin{figure}[ht]
\begin{align*}
\Qcircuit @C=1.2em @R=0.9em {
  |0\> && \gate{H} & \qw & \qw & \qw& \dots & & \ctrl{5} & \multigate{4}{\textsf{QFT}_{2^a}^{-1}} & \qw\\
  &&&&&& \Ddots\\
  |0\> && \gate{H} & \qw & \ctrl{3} & \qw& \dots & & \qw & \ghost{\textsf{QFT}_{2^a}^{-1}} & \qw\\
  \raisebox{6pt}{\vdots} \\
  |0\> && \gate{H} & \ctrl{1} & \qw & \qw & \dots & & \qw & \ghost{\textsf{QFT}_{2^a}^{-1}} & \qw\\
  |\psi_j\> && \qw & \gate{W^{2^0}} & \gate{W^{2^1}} & \qw & \dots & & \gate{W^{2^{a-1}}} & \qw & \qw
}
\end{align*}
\caption{The quantum phase estimation circuit. Here $W$ is a unitary operator with eigenvector $|\psi_j\>$; in $\pi/3$-amplitude estimation it is the quantum walk operator $W'_i$ in \eqn{mod-walk-operator}.}
\label{fig:quantum-phase}
\end{figure}

\begin{corollary}
  \label{cor:rotation-gates}
  Let $W$ be a unitary operator with a unique leading eigenvector $|\psi_0\>$ with eigenvalue $1$. Denote the remaining eigenvectors by $|\psi_j\>$ with corresponding eigenvalues $e^{2\pi i \xi_j}$.
  For any $\Delta \in (0,1]$ and $\epsilon_2 < 1/2$, define $a := \log(1/\Delta)$ and $c := \log(1/\sqrt{\epsilon_2})$. For any constant $\alpha \in \C$, there exists a quantum circuit $\tilde{R}$ that uses $ac$ ancilla qubits and invokes the controlled-$W$ gate $2^{a+1}c$ times such that
  \begin{align}
    \tilde{R}|\psi_0\>|0\>^{\otimes ac} &= (R|\psi_0\>)|0\>^{\otimes ac}
  \end{align}
  (where $R = \alpha |\psi_0\>\<\psi_0| - (I - |\psi_0\>\<\psi_0|)$)
  and
  \begin{align}
    \lVert\tilde{R}|\psi_j\>|0\>^{\otimes ac} - (R|\psi_j\>)|0\>^{\otimes ac}\rVert &\le \sqrt{\epsilon_2}
  \end{align}
  for $j \ne 0$ such that $\xi_j \ge\,\Delta$.
\end{corollary}
\begin{proof}
  Let $\tilde{R} := V^\dagger(I \otimes Q)V$ where $V$ is the quantum circuit in \lem{phase-estimation} and $Q := \alpha|0\>\<0|^{\otimes ac} + (I - |0\>\<0|^{\otimes ac})$. Then we have
  \begin{align}
    \tilde{R}|\psi_0\>|0\>^{\otimes ac} = V^\dagger(I \otimes Q)|\psi_0\>|0\>^{\otimes ac} 
        = \alpha|\psi_0\>|0\>^{\otimes ac} = R|\psi_0\>|0\>^{\otimes ac}.
  \end{align}
  For $j \ne 0$ such that $\xi_j \ge\,\Delta$,
  \begin{align}
    \label{eq:16}
    \tilde{R}|\psi_j\>|0\>^{\otimes ac} &= V^\dagger(I \otimes Q)(\sqrt{1 - \epsilon_2}|\psi_j\>|\chi_j\> + \sqrt{\epsilon_2}|\psi_j\>|0\>^{\otimes ac})\\
                                        &=V^\dagger(\sqrt{1-\epsilon_2}|\psi_j\>|\chi_j\> + \sqrt{\epsilon_2}\alpha|\psi_j\>|0\>^{\otimes ac}) \\
                                        &=V^\dagger(|\psi_j\> \otimes (\sqrt{1-\epsilon_2}|\chi_j\> + \sqrt{\epsilon_2}|0\>^{\otimes ac}) + \sqrt{\epsilon_2}(\alpha - 1)|\psi_j\>|0\>^{\otimes ac})\\
                                        &=|\psi_j\>|0_j\> + V^\dagger\sqrt{\epsilon_2}(\alpha - 1)|\psi_j\>|0\>^{\otimes ac}.
  \end{align}
  Thus $\lVert\tilde{R}|\psi_j\>|0\>^{\otimes ac} - (R|\psi_j\>)|0\>^{\otimes ac}\rVert \le \lVert V^\dagger\sqrt{\epsilon_2}(\alpha - 1)|\psi_j\>|0\>^{\otimes ac} \rVert \le \sqrt{\epsilon_2}$.
\end{proof}

Finally, we prove the following lemma for analyzing the error incurred by $\pi/3$-amplitude amplification in our quantum volume estimation algorithm:
\begin{lemma}
  \label{lem:pi3-error}
  Starting from $|\pi_i\>$, we can obtain a state $|\tilde\pi_{i+1}\>$ such that  $\lVert |\pi_{i+1}\> - |\tilde\pi_{i+1}\> \rVert \le \epsilon$ using $\tilde{O}(n^{3/2}\log(1/\epsilon))$ calls to the controlled walk operators $W'_i,W'_{i+1}$. This results in $\tilde{O}(n^{3/2}\log(1/\epsilon))$ calls to the membership oracle $O_\K$.
\end{lemma}

\begin{proof}
  From \thm{exp-hit-run-mixes}, \lem{l2warm-pi-plus1}, and \lem{l2warm-pi-plus1-reverse}, we find that
  \begin{itemize}
  \item $\pi_i(x)$ mixes up to total variation distance $\epsilon_1$ in $O\big(n^3\log^5 \frac{n}{\epsilon_1}\big)$ steps of the Markov chain $M_{i+1}$, and
  \item $\pi_{i+1}(x)$ mixes up to total variation distance $\epsilon_1$ in $O\big(n^3\log^5 \frac{n}{\epsilon_1}\big)$ steps of the Markov chain $M_{i}$.
  \end{itemize}
  From \prop{approx-warmness-quantum}, we find the following:
  \begin{itemize}
  \item $|\pi_i\> = |\pi'_i\> + |e_1\>$ where $|\pi'_i\>$ lies in the space of eigenvectors $|v^{(i+1)}_j\>$ of $W'_{i+1}$ such that $\lambda^{(i+1)}_j = 1$ or $\lambda^{(i+1)}_j \le 1 - \frac{1}{O(n^3\log^5(n/\epsilon_1))}$, and $\lVert |e_1\> \rVert \le \epsilon_1$; and
  \item $|\pi_{i+1}\> = |\pi'_{i+1}\> + |e_2\>$ where $|\pi'_{i+1}\>$ lies in the space of eigenvectors $|v^{(i)}_j\>$ of $W'_{i}$ such that $\lambda^{(i)}_j = 1$ or $\lambda^{(i)}_j \le 1 - \frac{1}{O(n^3\log^5(n/\epsilon_1))}$, and $\lVert |e_2\> \rVert \le \epsilon_1$.
  \end{itemize}

  Note that $|\pi_i\>$ and $|\pi_{i+1}\>$ are simply the leading eigenvectors of $W_i$ and $W_{i+1}$, respectively. Thus both $|\pi_i\>$ and $|\pi_{i+1}\>$ lie $\epsilon_1$ close to the ``good'' subspaces corresponding to $W'_i$ (respectively $W'_{i+1}$) which are spanned by eigenvectors $|v_{j}^{(i)}\>$ (respectively $|v_{j}^{(i+1)}\>$) with eigenvalues $e^{2\pi i {\xi^{(i)}_j}}$  (respectively $e^{2\pi i \xi_{j}^{(i+1)}}$) such that $\xi_j^{(i)} = 0$ or $\xi_j^{(i)} \ge \frac{1}{O(n^{3/2}\log^{5/2}(n/\epsilon_1))}$. Each state that occurs during $\pi/3$-amplitude amplification to rotate $|\pi_i\>$ to $|\pi_{i+1}\>$ or vice versa is a linear combination of $|\pi_i\>$ and $|\pi_{i+1}\>$ and is thus also close to the good subspaces of $W'_{i}$ and $W'_{i+1}$.

  Applying \cor{rotation-gates} with $\Delta = \frac{1}{n^{3/2}\ln^{5/2}(n/\epsilon_1)}$ and $\epsilon_2 = \epsilon_1^2$, we can implement a quantum operators $\tilde{R}_{i},\tilde{R}_{i+1}$ such that $\lVert R_{i} - \tilde{R}_{i} \rVert \le 2\epsilon_1$  and $\lVert R_{i+1} - \tilde{R}_{i+1} \rVert \le 2\epsilon_1$, using $O(n^{3/2}\log^{5/2}(n/\epsilon_1)\log(1/\epsilon_1))$ calls to the controlled-$W'_{i}$  and controlled-$W'_{i+1}$ operators, respectively.

  The above shows how to approximately implement $R_i$ and $R_{i+1}$. If these operators could be implemented perfectly, \lem{pi3-amplification} and \lem{inner-product} show that we can prepare a state $|\tilde\pi_{i+i}\>$ such that $\<\pi_{i+1}|\tilde\pi_{i+1}\> \le 1 - (2/3)^{3^m}$ by applying $m$ recursive levels of $\pi/3$-amplitude amplification to $|\pi_i\>$, using $3^m$ calls to $R_i,R_i^\dagger,R_{i+1},R_{i+1}^\dagger$. Since $\lVert \pi_{i+1} - \tilde\pi_{i+1} \rVert = \sqrt{2(1 - \<\pi_{i+1}|\tilde\pi_{i+1}\>)}$, after $O(\log(1/\epsilon_2))$ calls to the rotation gates we obtain a final state with error $\epsilon_2$. However, each rotation gate can cause an error of $\epsilon_1$ by itself. By making $O(n^{3/2}\log^{5/2}(n/\epsilon_1)\log(1/\epsilon_1)\log(1/\epsilon_2))$ calls to controlled-$W'_i$ and controlled-$W'_{i+1}$ operators, we obtain a final error of $O(\epsilon_1\log(1/\epsilon_2) + \epsilon_2)$. Choosing $\epsilon_2 = \epsilon/2$ and $\epsilon_1 = \epsilon/(2\ln(2/\epsilon))$ gives the result.
\end{proof}

\paragraph{Error analysis for the quantum Chebyshev inequality}
We also analyze the error from the quantum Chebyshev inequality (\thm{quantum-Chebyshev}), giving a robust version of \lem{Chebyshev-lemma}.

\begin{lemma}\label{lem:err-non-destructive}
Suppose we have $\tilde{O}(\log(1/\delta)/\epsilon)$ copies of a state $|\tilde\pi_{i-1}\>$ such that $\lVert |\tilde\pi_{i-1}\> - |\pi_{i-1}\> \rVert \le \epsilon$. Then the quantum Chebyshev inequality can be used to output $\tilde{V}_{i}$ such that $|\tilde{V}_{i} - \E_{\pi_{i}}[V_{i}]| \le O(\epsilon)\E_{\pi_{i}}[V_{i}]$ with success probability $1-\delta^4$ using $\tilde{O}(n^{3/2}\log(1/\delta)/\epsilon)$ calls to the membership oracle. The output state $|\hat\pi_{i-1}\>$ satisfies $\lVert|\hat\pi_{i-1}\> - |\pi_{i-1}\>\rVert = O(\epsilon + \delta)$.
\end{lemma}

\begin{proof}
  The error-free version of this lemma was proven in \lem{Chebyshev-lemma}. Here we focus on the error analysis. The quantum Chebyshev inequality uses an implementation of $US_0U^\dagger S_i$ where $U$ is a unitary operator satisfying $U|\pi_{i-1}\> = |\pi_i\>$. From \lem{pi3-error}, using $\log (1/\epsilon_2)$ iterations of $\pi/3$-amplitude amplification ($U_{\log 1/\epsilon_2}$ in \eqn{pi/3-Ui}) instead of $U$ induces an error of $\epsilon_2$ and uses $O(n^{3/2}\log(1/\epsilon_2))$ oracle calls. Using approximate phase estimation as in \cor{rotation-gates} and \lem{pi3-error}, $\Pi_{i-1}$ and $\Pi_{i}$ can be implemented up to error $\epsilon_3$ using $O(n^{3/2}\log(1/\epsilon_3))$ oracle calls. Thus each block corresponding to \thm{AmpEst} induces an error of $O(\epsilon_2 + \epsilon_3)$, and the final state before the median is measured has an error of $O(\epsilon + \epsilon_2 + \epsilon_3)$. Therefore, using $O(\log(1/\delta_1)/\epsilon)$ copies of $|\tilde\pi_{i-1}\>$ returns a sample $\tilde{V}_{i}$ such that $|\tilde{V}_{i} - \E_{\pi_{i}}[V_{i}]| \le O(\epsilon_2+\epsilon_3+\epsilon) \E_{\pi_{i}}[V_{i}]$ with success probability $1 - \delta_1$. Performing a measurement with success probability $1 - \delta_1$ implies that the posterior state has an overlap $\sqrt{1-\delta_1}$ with the initial state. This induces an error of magnitude at most $\sqrt{2(1 - \sqrt{1-\delta_1})}=O(\delta_1^{1/4})$.

The measurement on the $\log(1/\delta)/c$ copies of $|\tilde\pi_{i-1}\>$ used to estimate $\hat\mu$ has relative error at most $c$ with probability $1-\delta$. This causes an error $O(\delta_1^{1/4})$ in addition to the error $\epsilon_2$ from $\pi/3$-amplitude amplification.

Finally, note that the basic amplitude estimation circuit (analyzed in \thm{AmpEst}) is a subroutine of the quantum Chebyshev inequality (\thm{quantum-Chebyshev}), and uncomputing the block corresponding to \thm{AmpEst} induces an error of $O(\epsilon_2 + \epsilon_3)$, giving an overall error of $O(\epsilon_2 + \epsilon_3 + \epsilon +\,\delta^{1/4})$. The result follows by taking $\epsilon_2 = \epsilon_3 = \epsilon$ and $\delta_1 =\,\delta^4$.
\end{proof}

We finally prove \lem{mean-estimation-error} here.
\begin{proof}
  \lem{err-non-destructive} is used to estimate the mean with $\epsilon = \epsilon_1$ and leaves a posterior state $|\hat\pi_{i-1}\>$ such that $\lVert|\hat\pi_{i-1}\> - |\pi_{i-1}\>\rVert=O(\epsilon_1 +\,\delta)$. We can then use $\pi/3$-amplitude amplification to rotate this state into $|\tilde\pi_{i}\>$, adding error $O(\epsilon')$ at the cost of $O(n^{3/2}\log(1/\epsilon'))$. This completes the proof.
\end{proof}

\subsection{Quantum algorithms for rounding logconcave densities}\label{sec:round-quantum}
We first define roundedness of logconcave density functions as follows:
\begin{definition}\label{defn:well-roundedness-logconcave}
A logconcave density function $f$ is said to be $c$-rounded if
\begin{enumerate}[ref={condition~\arabic*}]
\item The level set of $f$ of probability $1/8$ contains a ball of radius $r$; \label{condition1}
\item $\E_{f}\left(\lvert x - z_{f}\rvert\right) \le R^2$, where $z_{f}$ is the centroid of $f$, i.e., $z_{f}:=\E_{f}(x)$; \label{condition2}
\end{enumerate}
and $R/r \le c\sqrt{n}$.
\end{definition}
In the previous section we assumed that the distributions $\pi_i$ sampled during the hit-and-run walk are $O(1)$-rounded (i.e., well-rounded). From \thm{exp-hit-run-mixes}, this implies that the hit-and-run walk for the distribution $\pi_{i}$ mixes from a warm start in time $\tilde{O}(n^3)$. In this subsection we show how the distributions $\pi_{i}$ can be transformed to satisfy this condition.

Following the classical discussion in \cite{lovasz2006fast}, we actually show a stronger condition: the distributions are transformed to be in ``near-isotropic'' position. A density function $f$ is said to be in \emph{isotropic} position if
\begin{align}\label{eq:isotropic-defn}
\E_{f}[x]=0\quad\text{and}\quad \E_{f}[xx^{T}]=I.
\end{align}
The latter equation is equivalent to $\int_{\R^n}(u^{T}x)^{2}f(x)\,\d x=|u|^{2}$ for every vector $u\in\R^{n}$. We say that $\K$ is \emph{near-isotropic} up to a factor of $c$ if
\begin{align}
\frac{1}{c}\leq\int_{\R^n}(u^{T}(x-z_f))^{2}f(x)\,\d x\leq c
\end{align}
for every unit vector $u \in \R^n$.

The following lemma shows that logconcave density functions in isotropic position are also $O(1)$-rounded:
\begin{lemma}[{\cite[Lemma 5.13]{lovasz2007geometry}}]
  \label{lem:isotropic-is-rounded}
  Every isotropic logconcave density is $(1/e)$-rounded.
\end{lemma}

The following lemma shows that any logconcave density function can be put into isotropic position by applying an affine transformation, generalizing the same result for uniform distributions by Rudelson~\cite{rudelson99}:
\begin{lemma}[{\cite[Lemma 2.2]{lovasz2006fast}}]
  \label{lem:affine-transform-isotropic}
  Let $f$ be a logconcave function in $\R^n$ such that there is no linear subspace $\mathcal{S} \subseteq \R^{n}$ such that $\int_{\mathcal{S}}f(x) \,dx > 1/2$, and let $X^1,\dots,X^k$ be independent random points from the corresponding distribution. There is a constant $C_0$ such that if $k > C_0t^3 \ln n$, then the transformation $g(x) = T^{-1/2}x$ where
  \begin{align}
    \label{eq:affine-trans}
    \bar{X} = \frac{1}{k}\sum_{i=1}^{k}X^{i}, \qquad T = \frac{1}{k}\sum_{i=1}^{k}(X^{i} - \bar{X})(X^{i} - \bar{X})^T
  \end{align}
  puts $f$ in $2$-isotropic position with probability at least $1 - 1/2^t$.
\end{lemma}
From \lem{affine-transform-isotropic}, $k = \lceil C_0n\ln^5 n \rceil = \tilde{\Theta}(n)$ samples from a logconcave density $f$ suffice to put it into near-isotropic position. However, efficiently obtaining samples from a density $\pi_{i}$ requires it to be well-rounded to start with. To overcome this difficulty, we interlace the rounding with the stages of the volume estimation algorithm where in each stage, we obtain an affine transformation that puts the density to be sampled in the next stage into isotropic position. The density $\pi_0$ is very close to an exponential distribution  (since it is concentrated inside the convex body) and can hence be sampled without resorting to a random walk.

To show that samples from $\pi_{i}$ can be used to transform $\pi_{i+1}$ into isotropic position, we use the following lemma:
\begin{lemma}[{\cite[Lemma 4.3]{kalai2006simulated}}]
  \label{lem:variance-relation}
  Let $f$ and $g$ be logconcave densities over $K$ with centroids $z_f$ and $z_g$ respectively. Then for any $u \in \R^n$,
  \begin{align}
    \label{eq:2}
    \E_{f}[(u \cdot (x-z_f))^2] \le 16\E_f\left[\frac{f}{g}\right]\E_{g}[(u \cdot (x - z_g))^2].
  \end{align}
\end{lemma}

We now have the following proposition:
\begin{proposition}
  \label{prop:isotropic-adjacent}
  If affine transformation $S_i$ puts $\pi_{i}$ in near-isotropic position then it also puts $\pi_{i+1}$ in near-isotropic position.
\end{proposition}
\begin{proof}
  Let $S_i$ put $\pi_{i}$ in $2$-isotropic position.
  Applying \lem{variance-relation} with $f = \pi_{i+1}, g = \pi_{i}$, we have that for any unit vector $u \in \R^n$,
  \begin{align}
    \E_{\pi_{i+1}}[(u \cdot (x - z_{\pi_{i+1}}))^2] &\le 16\E_{\pi_{i+1}}\left[\frac{\pi_{i+1}}{\pi_{i}}\right]\E_{\pi_{i}}[(u \cdot (x - z_{\pi_{i}}))^2] \le 32e^{2}
  \end{align}
  since $\E_{\pi_{i+1}}\left[\frac{\pi_{i+1}}{\pi_{i}}\right] \le e^{2}$ from \lem{l2warm-pi-plus1-reverse}. Again applying \lem{variance-relation}
  \begin{align}
    \frac{1}{2} \le \E_{\pi_{i}}[(u \cdot (x - z_{\pi_{i}}))^2] &\le \E_{\pi_{i}}\left[\frac{\pi_{i}}{\pi_{i+1}}\right]\E_{\pi_{i+1}}[(u \cdot (x - z_{\pi_{i+1}}))^2]
  \end{align}
  $\E_{\pi_{i}}\left[\frac{\pi_{i}}{\pi_{i+1}}\right] \le 8$ from \lem{l2warm-pi-plus1}. Therefore,
  \begin{align}
    \frac{1}{2} \le 128e^{2}\E_{\pi_{i+1}}[(u \cdot (x - z_{\pi_{i+1}}))^2]
  \end{align}
  Thus $E_{\pi_{i+1}}$ is also put in near-isotropic position.
\end{proof}

We finally have the main result of this section:
\begin{proposition}\label{prop:rounding-interlacing}
At each stage $i$ of \algo{quantum-interlacing}, the affine transformation puts the distribution $\pi_{i+1}$ in near-isotropic position using an additional $\tilde{O}(n^{2.5})$ quantum queries to $O_K$.
\end{proposition}
\begin{proof}
  Since $\pi_0$ is nearly an exponential distribution, it can be sampled without using a random walk and thus the proposition is true for $i=0$. Assume that the proposition is true for $1,2,\dots,k$. Then an affine transformation can be found to put $\pi_{k}$ in near-isotropic position. Thus a classical hit-and-run walk starting from $\pi_{k-1}$ converges to $\pi_{k}$ in $\tilde{O}(n^3)$ steps. By the analysis in \sec{proof-error-analysis}, a quantum sample $|\pi_{k-1}\>$ can be rotated to $|\pi_{k}\>$ using $\tilde{O}(n^{1.5})$ quantum queries. $\tilde{O}(n)$ such samples suffice to compute the covariance matrix $T$ in \eq{affine-trans}, which puts $\pi_k$ in $2$-isotropic position. By \prop{isotropic-adjacent}, this also puts $\pi_{k+1}$ in near-isotropic position. This concludes the proof.
\end{proof}

\begin{algorithm}[ht]
\KwInput{Membership oracle $O_{\K}$ for $\K$.}
\KwOutput{$\epsilon$-multiplicative approximation of $\vol(\K)$.}
Set $m=\Theta(\sqrt{n}\log(n/\epsilon))$ to be the number of iterations of simulated annealing and $a_{i}=2n(1-\frac{1}{\sqrt{n}})^{i}$ for $i\in\range{m}$. Let $\pi_{i}$ be the probability distribution over $\K'$ with density proportional to $e^{-a_{i}x_{0}}$\; \nonl
Set error parameters $\delta,\epsilon' = \Theta(\epsilon/m^2), \epsilon_1 = \epsilon/2m$; let $k=\tilde{\Theta}(\sqrt{n}/\epsilon)$ be the number of copies of stationary states for applying the quantum Chebyshev inequality; let $l = \tilde{\Theta}(n)$ be the number of copies of stationary states needed to obtain the affine transformation $S_i$; \nonl
Prepare $k+l$ (approximate) copies of $|\pi_{0}\>$, denoted $|\tilde\pi_{0}^{(1)}\>,\ldots,|\tilde\pi_{0}^{(k+l)}\>$\;
\For{$i\in\range{m}$}{
  Use the quantum Chebyshev inequality on the $k$ copies of the state $|\tilde{\pi}_{i-1}\>$ with parameters $\epsilon_1,\delta$ to estimate the expectation $\E_{\pi_{i}}[V_i]$ (in Eq.~\eqn{Vi-expectation}) as $\tilde{V_i}$ (\lem{err-non-destructive} and \fig{SA-block-nondes}). The post-measurement states are denoted $|\hat\pi_{i-1}^{(1)}\>,\ldots,|\hat\pi_{i-1}^{(k)}\>$\;
  Use the $l$ copies of the state $|\pi_{i-1}\>$ to nondestructively\footnotemark{} obtain the affine transformation $S_{i} = T = \frac{1}{l}\sum_{q=1}^{l}(X^{q} - \bar{X})(X^{q} - \bar{X})^T$ where the $X_q$ are samples from the density $\pi_{i-1}$ and $\bar{X} = \frac{1}{l}\sum_{q=1}^l X^q$. The post-measurement states are denoted $|\hat\pi_{i-1}^{(k+1)}\>,\ldots,|\hat\pi_{i-1}^{(k+l)}\>$\;
  Apply $\pi/3$-amplitude amplification with error $\epsilon'$ (\sec{quantum-MCMC} and \lem{pi3-error}) and affine transformation $S_i$ to map $|S_i\hat\pi_{i-1}^{(1)}\>,\ldots,|S_i\hat\pi_{i-1}^{(k+l)}\>$ to $|S_i\tilde\pi_{i}^{(1)}\>,\ldots,|S_i\tilde\pi_{i}^{(k+l)}\>$, using the quantum hit-and-run walk\;
  Invert $S_i$ to get $k + l$ (approximate) copies of the stationary distribution $|\pi_{i}\>$ for use in the next iteration\;
}
Compute an estimate $\widetilde{\text{Vol}(\K')} = n!v_{n}(2n)^{-(n+1)}\tilde{V}_{1}\cdots \tilde{V}_{m}$ of the volume of $\K'$, where $v_{n}$ is the volume of the $n$-dimensional unit ball\;
Use $\widetilde{\text{Vol}(\K')}$ to estimate the volume of $\K$ as $\widetilde{\text{Vol}(\K)}$ (\sec{pencil-construction-analysis}).
\caption{Volume estimation of convex $\K$ with interlaced rounding.}
\label{algo:quantum-interlacing}
\end{algorithm}

\paragraph{Rounding the convex body as a preprocessing step}
Consider applying only the rounding part of \algo{quantum-interlacing}. By \prop{rounding-interlacing}, the final affine transformation puts the density $\pi_m \propto e^{-a_m x_0}$ in near-isotropic position. Since $a_m \le \epsilon^2/n$, we have
\begin{align}
  \label{eq:4}
  (1 - \epsilon^2)\E_{\K'}[|X - \bar{X}|]^2 \le \int_{\K'}\frac{e^{-a_m x_0}|x - \bar{x}|^2}{Z(a_m)}\,\d x \le 2n;
\end{align}
thus $\E_{\K'}[|X - \bar{X}|]^2 \le 2n/(1 - \epsilon^2)$. From \cite[Lemma 3.3]{lovasz2006fast}, all but an $\epsilon$-fraction of the body is contained inside a ball of radius $O(\sqrt{n})$. Combined with our assumption that $B_2(0,1) \subseteq \K'$, this shows that $S_{m+1}$ puts the convex body $\K'$ in well-rounded position.

\footnotetext{Similar to \lem{Chebyshev-lemma}, we do not directly measure the states; instead we use a quantum circuit to (classically) compute the affine transformation $S_{i}$ and apply it to the convex body coherently for the next iteration. Note that the quantum register holding the affine transformation will be in some superposition, but by using $O(\log n)$ copies and taking the mean (as in \lem{Chebyshev-lemma}), the amplitude of the correct affine transformation will be arbitrarily close to 1.}


\section{Implementation of the quantum hit-and-run walk}\label{sec:implement-hit-and-run}
Due to the precision of representing real numbers, the implementation of volume estimation algorithms in practice requires to walk in a discrete domain that is a subset of $\R^n$. It is known that walks only taking local steps within a short distance (such as the grid walk and the ball walk) can be discretized with good approximation by dividing $\R^{n}$ into small hypercubes and walking on their centers (see e.g.~\cite{frieze1999log}), but such error analysis does not automatically apply to hit-and-run walks for which we did not find existing classical discretizations. We emphasize the discretization in contrast to most classical treatments for two reasons: (1) Quantum algorithms are typically presented in a circuit model, in contrast to the RAM model used by classical algorithms. Continuous variables in the circuit model correspond to registers of infinite size, preventing a clear analysis of the resources of the algorithm in terms of gate count. Specifically to obtain the performance of the algorithm in reality, we must show that $\poly(\log(1/\epsilon))$ bit registers suffice. (2) Standard methods of preparing walk operators corresponding to classical Markov Chains (see for eg. \cite{szegedy2004quantum}) rely on the sparsity of the transition matrix. In the case of geometric random walks sparsity is not well-defined in the continuous case and may not hold even for discretizations (for example, the hit-and-run walk has a non-zero transition density to any point in the convex body.) The efficient preparation of quantum states corresponding to classical distributions is not always a trivial operation, and there has been research~\cite{holmes2020efficient} about preparing common distributions for quantum Monte-Carlo methods. Most existing general procedures come without provable guarantees on the resources required for sufficiently accurate samples; we provide here a simple analysis for the cost of implementing the hit-and-run walk via the Grover-Rudolph method~\cite{grover2002creating}.

In this section, we introduce a discretized quantum hit-and-run walk and give an explicit analysis of its implementation. The basic idea of the discretization is to represent the coordinates with rational numbers. We approximate $\rmk$ by a set of discretized points in $\rmk$ and define a Markov chain on these points (see \sec{discretization-hit-and-run}). We use a two-level discretization: the hit-and-run process is performed with a coarser discretization and then a point in a finer discretization of the coarse grid is chosen uniformly at random as the actual point to jump to. This ensures that the starting and ending points (in the coarser discretization) of one jump are far from the boundary so that a small perturbation does not change the length of the chord induced by the two points significantly. Then in \sec{discretization-conductance}, the discrete conductance can be bounded by bounding the distance between the discrete and continuous transition probabilities as well as the distance between the discrete and continuous subset measures. In \sec{impl-quant-walk}, we prove that the quantum gate complexity of implementing the discretized quantum hit-and-run walk is $\tilde{O}(n)$, the same overhead as for implementing classical hit-and-run walks.

\subsection{Discretization of the hit-and-run walk}\label{sec:discretization-hit-and-run}
For a convex body $\rmk \subseteq \bbr^n$, we let $\rmke$ denote the set of vectors in $\rmk$ whose coordinates can be represented by some fixed-point representation using $\log(1/\epsilon)$ bits\footnote{Note that this $\epsilon$ is different from the multiplicative error in the problem definition. However, this $\epsilon$ is not the dominating error and the overhead is only logarithmic.} We can use the . We call $\rmke$ an \emph{$\epsilon$-discretization} of $\rmk$. The finite set $\rmke$ provides an $\epsilon$-net for $\rmk$. We also define $(\bbr^n)_{\epsilon}$ as a $\epsilon$-discretization of $\bbr^n$.

We consider a Markov chain whose states are the points in $\rmke$.
For any $v \in \bbr^n$, we define the \emph{$\epsilon$-box} $b_{\epsilon}(v) := \{x \in \bbr^n: x(i) \in [v(i)-\epsilon/2, v(i)+\epsilon/2],\, \forall i \in \range{n}\}$. Let $\overline{\rmk}_{\epsilon}$ be the continuous set formed by the $\epsilon$-boxes of the points in $\rmke$, i.e., $\overline{\rmk}_{\epsilon} = \bigcup_{x \in \rmke}b_{\epsilon}(x)$. For two distinct points $u, v \in \bbr^n$, we denote by $\ell_{uv}$ the line through them. For a line $\ell \subseteq \bbr^n$, let $\ell(\overline{\rmk}_{\epsilon})$ be the segment of $\ell$ contained in $\overline{\rmk}_{\epsilon}$, i.e., $\ell(\overline{\rmk}_{\epsilon}) = \{x \in \ell: x \in \overline{\rmk}_{\epsilon}\}$. In addition, for $u \in \ell$, we define $\ell(\overline{\rmk}_{\epsilon}, u, \epsilon')$ as the $\epsilon'$-discretization of $\ell(\overline{\rmk}_{\epsilon})$ starting from $u$, i.e., $\ell(\overline{\rmk}_{\epsilon}, u, \epsilon') = \{x \in \ell(\overline{\rmk}_{\epsilon}): |x-u| = k\epsilon' \text{ for some } k \in \{0, 1, \ldots\}\}$. Analogous to the distribution $\pi_f$ for the continuous-space case, we define its corresponding discrete distribution $\hat{\pi}_f$ with $\hat{\pi}_f(\rms) = \sum_{x \in \rms}f(x)/\sum_{x \in (\bbr^n)_{\epsilon}}f(x)$.

To implement the hit-and-run walk (see \sec{hit-and-run}), we sample a uniformly random direction from a point $u$. We achieve this by sampling $n$ coordinates according to the standard normal distribution from the corresponding coordinate of $u$ and normalizing the new point to have unit length; the uniformity of such sampling is well known (see for example \cite{muller1959note,marsaglia1972choosing}).\footnote{A one-line proof is that this distribution is invariant under orthogonal transformations, but the uniform distribution on the $n$-dimensional unit sphere $\rmb_{n}$ is the unique distribution that satisfies this property. Although the Gaussian distributions we sample from are discretized, the invariance under orthogonal transformations holds approximately, so we have approximate uniformity.} Let this normalized point be $v$, so that the sampled direction is $\ell_{u v}$. Note that the coordinate we sample from is discrete. The directions we can sample form a discrete set denoted $\rml(u, \epsilon')$, where $\epsilon'$ is the precision for sampling directions.

Now we compute the probability that a specific direction is sampled. After normalization, the point will ``snap'' to a point in $(\bbr^n)_{\epsilon}$. Consider $\bigcup_{v:\, b_{\epsilon'}(v)\cap\rmb_n \neq \emptyset}b_{\epsilon}(v)$, where $\rmb_n$ is the $n$-dimensional unit sphere. We use the $(n-1)$-dimensional volume (surface area) of this body to approximate that of $\rmb_n$, with up to a $\sqrt{2}$ enlargement factor due to the fact that $\epsilon$-boxes have sharp corners. Thus, the number of points that $v$ can snap to is in the range $[n\vol(\rmb_n)/\epsilon^{n-1}, \sqrt{2}n\vol(\rmb_n)/\epsilon^{n-1}]$, which is also the range of $|\rml(u, \epsilon')|$. To make the lines in $\rml(u, \epsilon')$ cover every $\epsilon$-box on the boundary of $\sqrt{n}\rmb_n$ (so that it is possible to sample all the points in $\rmke$), we need $\epsilon' \leq \epsilon/\sqrt{n}$.

Let $L := |\rml(u, \epsilon')|$. We label the lines in $\rml(u, \epsilon')$ as $\{\ell_1, \ldots, \ell_L\}$ (ordered arbitrarily). For each $i \in [L]$, let $v_i$ be the point after normalization. Intuitively, $v_i$ approximates a point on the ``surface'' of the unit ball around $u$ (see~\fig{hyperpyramid}). There are \emph{hyperfaces} of $b_{\epsilon'}(v_i)$ that are out-facing and not adjacent to any $\epsilon'$-box in $\bigcup_{v:\, b_{\epsilon'}(v)\cap\rmb_n \neq \emptyset}b_{\epsilon}(v)$ (as an illustration, see dashed edges in~\fig{hyperpyramid}). For all points $v''$ in these hyperfaces, the line segments from $u$ through $v''$ of length $\sqrt{n}$ form a set, which we refer to as a \emph{hyperpyramid}, denoted by $\rmp_i$. The apex of each hyperpyramid is $u$, and the base of each hyperpyramid is a subset of the hyperspherical surface. Intuitively, the bases of $\rmp_1, \ldots, \rmp_L$ form a partition of the ``surface'' of the ball of radius $\sqrt{n}$ around $u$, and therefore $\{\rmp_1, \ldots, \rmp_L\}$ forms a partition of the ball of radius $\sqrt{n}$ around $u$.

\begin{figure}[h]
  \centering
  \includegraphics[width=0.7\textwidth]{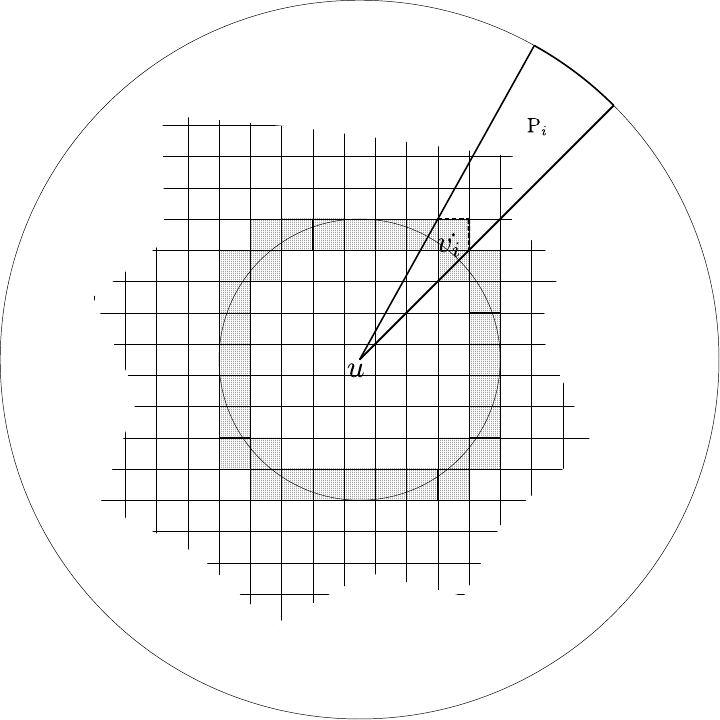}
  \caption{Constructing a hyperpyramid. The inner circle represents the unit ball and the outer circle represents the ball of radius $\sqrt{n}$. The grids represents the $\epsilon'$-discretization of $\bbr^n$; each grid is an $\epsilon'$-box. The shaded boxes are points where a direction ``snaps'' to after normalization, and the dashed edges of $b_{\epsilon'}(v_i)$ is its ``outer face.'' The hyperpyramid $\rmp_i$ is represented by a circular sector.}
  \label{fig:hyperpyramid}
\end{figure}

\subsection{Conductance lower bound on the discretized hit-and-run walk}\label{sec:discretization-conductance}

The discretized hit-and-run walk on $\rmke$ described above can be summarized as \algo{hit-and-run-1step}.

\begin{algorithm}[htbp]
  \SetKwInOut{Input}{Input}
  \Input{Current point $u \in \rmke$.}
  Uniformly sample a line $\ell \in \rml(u, \epsilon)$ by independently sampling $n$ coordinates around $u$ according to the standard normal distribution and then normalizing to unit length\;
  Sample a point $v'$ in $\ell(\overline{\rmk}_{\epsilon}, u, \epsilon')$ according to $f$\;
  Let $v'' \in \rmkse$ that is closest to $v'$\;
  Output a uniform sample $v$ in $b_{\sqrt{\epsilon}n^{1/4}}(v'') \cap (\bbr^n)_{\epsilon}$\;
  \caption{One step of the discretized hit-and-run walk.}
  \label{algo:hit-and-run-1step}
\end{algorithm}

Note that we have used a two-level discretization of $\rmk$, as illustrated in \fig{2-level-dis}. The first level is a coarser discretization $\rmk_{\sqrt{\epsilon}n^{1/4}}$ and the second level is a finer discretization $\rmke$. We first choose a temporary point $v''$ in $\rmk_{\sqrt{\epsilon}n^{1/4}}$. Then we choose a point $v$ uniformly at random in $b_{\sqrt{\epsilon}n^{1/4}}(v'')\cap (\bbr^n)_{\epsilon}$ to jump to. The purpose of this two-level discretization is to avoid having a small change of the original point $u$ cause a huge difference in $\ell_{uv}(\overline{\rmk}_{\epsilon})$ (when $u$ is very close to the boundary).
\begin{figure}[h]
  \centering
  \includegraphics[width=0.5\textwidth]{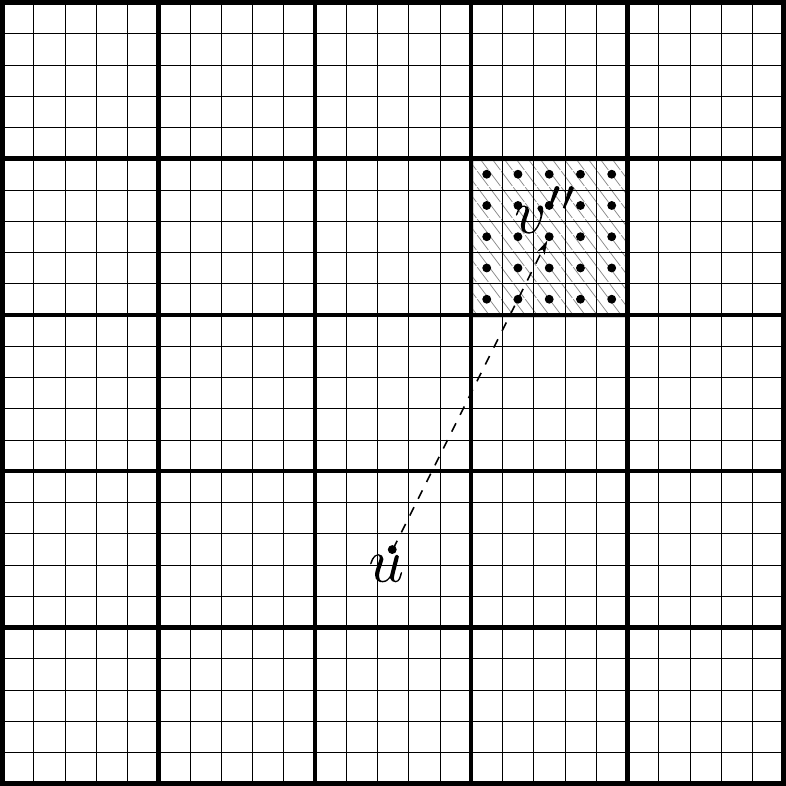}
  \caption{A demonstration of the 2-level discretization of $\rmk$. The thicker grid represents the coarser discretization $\rmk_{\sqrt{\epsilon}n^{1/4}}$ and the thinner grid represents the finer discretization $\rmke$. When $v''$ is chosen from $\rmk_{\sqrt{\epsilon}n^{1/4}}$, an actual point $v$ to jump is chosen uniformly at random in $b_{\sqrt{\epsilon}n^{1/4}}(v'')\cap (\bbr^n)_{\epsilon}$ marked by the points in the shaded region.}
  \label{fig:2-level-dis}
\end{figure}

We first compute the transition probability of the discretized hit-and-run walk.

\begin{lemma}
  The transition probabilities defined by \algo{hit-and-run-1step} satisfy
  \begin{align}
    \label{eq:pxy}
    P_{uv} \geq \sum_{\substack{v' \in \ell(\overline{\rmk}_{\epsilon} u, \epsilon'), \ell \in \rml(x, \epsilon):\\\ell(\overline{\rmk}_{\epsilon}, u, \epsilon') \cap b_{\sqrt{\epsilon}n^{1/4}}(v) \neq \emptyset}}\frac{\epsilon^{n-1}(\sqrt{\epsilon})^nf(v')}{\sqrt{2}n^{1+n/4}\vol(\rmb_n)(\sqrt{n})^{n-1}\hat{\mu}_f(\ell(\overline{\rmk}_{\epsilon}, u, \epsilon'))},
\end{align}
where for any $\rms\subseteq\R^{n}$, we define
\begin{align}\label{eqn:mu-f-defn-discrete}
\hat{\mu}_f(\rms):=\sum_{x\in \rms}f(x).
\end{align}
\end{lemma}

\begin{proof}
  First note that the probability of a line $\ell \in \rml(u, \epsilon)$ being sampled is at least $\frac{\epsilon^{n-1}}{\sqrt{2}n\vol(\rmb_n)(\sqrt{n})^{n-1}}$. Along $\ell$, the probability of sampling $v'$ is $f(v')/\hat{\mu}_f(\ell(\overline{\rmk}_{\epsilon}, u, \epsilon'))$, and the probability of choosing $v$ in $b_{\sqrt{\epsilon}n^{1/4}}(v'')\cap\rmke$ is $(\sqrt{\epsilon})^n/n^{n/4}$.
\end{proof}

According to the definition in \eqn{conductance-defn-discrete}, the conductance of any subset $\rms \subseteq \rmke$ is
\begin{align}
  \phi(\rms) = \frac{\sum_{u \in \rms}\sum_{v \in \rmke\setminus\rms}P_{uv}\hat{\pi}_f(u)}{\min\{\hat{\pi}_f(\rms), \hat{\pi}_f(\rmke\setminus\rms)\}},
\end{align}
where $\hat{\pi}_f$ is defined as $\hat{\pi}_f(\rma) = \sum_{x \in \rma}f(x)$. The conductance of the Markov chain is then
\begin{align}
  \phi = \min_{\rms \subseteq \rmke} \phi(\rms).
\end{align}
Now we prove the main theorem of this section, which shows that the conductance of the discretized hit-and-run walk does not differ significantly from that of the continuous hit-and-run walk.

\begin{theorem}
  \label{thm:conductance}
  Let $\rmke$ be the discretization of convex body $\K$ that contains a unit ball and is contained in a ball with radius $R \leq \sqrt{n}$. Let the density function be $f(x) = e^{-a^Tx}$ having support $\K$ where $a = (1, 0, \ldots, 0)$. Let $\epsilon' \leq \sqrt{\epsilon}n^{-3/4}$. For $\rms \subseteq \rmke$ such that $\hat{\pi}_f(\rms) \leq 1/2$, we have
  \begin{align}
    \phi(\rms) \geq \frac{1}{10^{16}n\sqrt{n}\ln(\frac{2n\sqrt{n}}{\hat{\pi}_f(\rms)})} - \epsilon.
  \end{align}
\end{theorem}

\begin{proof}
  This proof closely follows that of \cite[Theorem 6.9]{LV06}.
  We first consider the transition probability for the continuous hit-and-run walk in $\rmk$. For $u, v \in \rmk$, recall that
  \begin{align}
    P'_u(b_{\epsilon}(v)) = \frac{2}{n\vol(\rmb_n)}\int_{b_{\epsilon}(v)}\frac{f(x)\,\d x}{\mu_f(u, x)|x-u|^{n-1}},
  \end{align}
  where $\mu_f(u,x)$ is a shorthand for $\mu_f(\ell_{ux}(\overline{\rmk}_\epsilon))$. We compare $P'_u(b_{\sqrt{\epsilon}n^{1/4}}(v))$ with $P_{uv}$ for $u \in \rmke$ and $v \in \rmkse$. To this end, we use $\hat{\mu}_f$ to approximate $\mu_f$: for each $\ell$, we have
  \begin{align}
    \label{eq:uf-uhatf}
    \epsilon'\hat{\mu}_f(\ell(\overline{\rmk}_{\epsilon}, u, \epsilon')) \leq e^{\epsilon'}\mu_f(\ell(\overline{\rmk}_{\epsilon})).
  \end{align}
  Consider each hyperpyramid $\rmp_i$ defined in  \sec{discretization-hit-and-run} whose associated line through its apex is $\ell_i$ and $\ell_i(\overline{\rmk}_{\epsilon}, u, \epsilon') \cap b_{\sqrt{\epsilon}n^{1/4}}(v) \neq \emptyset$. Note that the distance between each $u \in \rmke$ and the boundary of $\overline{\rmk}_{\epsilon}$ is at least $\epsilon/2$. Inside each hyperpyramid, the length of the chords through $u$ can differ by a factor at most 2. For each $\ell \subset \rmp_i$, $\hat{\mu}_f(\ell_i(\overline{\rmk}_{\epsilon}, u, \epsilon')) \leq 2e^{\epsilon'}\hat{\mu}_f(\ell(\overline{\rmk}_{\epsilon}, u, \epsilon'))$. Together with \eq{uf-uhatf}, it follows that
  \begin{align}
    \label{eq:mu-muhat}
    \epsilon'\hat{\mu}_f(\ell_i(\overline{\rmk}_{\epsilon}, u, \epsilon')) \leq 2e^{2\epsilon'}\mu_f(\ell(\overline{\rmk}_{\epsilon}))
  \end{align}
  for all $\ell \subset \rmp_i$. Define $c_i := |\ell_i(\overline{\rmk}_{\epsilon}, u, \epsilon') \cap b_{\sqrt{\epsilon}n^{1/4}}(v)|$ (the number of points in this set) and $d_i := |\ell_i(\overline{\rmk}_{\epsilon}) \cap b_{\sqrt{\epsilon}n^{1/4}}(v)|$ (the length of this line). Note that $c_i \leq d_i / \epsilon'$. We further partition $\rmp_i$ into $c_i$ sets $\rmq_{i,1}, \ldots, \rmq_{i, c_i}$ along the direction of $\ell_i$ so that the distance between the hyperplanes that separate adjacent sets is at most $\epsilon'$. For each $j \in [c_i]$, we have
  \begin{align}
    \frac{\epsilon^{n-1}f(v')}{n\vol(\rmb_n)(\sqrt{n})^{n-1}\hat{\mu}_f(\ell_i(\overline{\rmk}_{\epsilon}, u, \epsilon'))}
&= \frac{\epsilon^{n-1}f(v')\epsilon'|v'-u|^{n-1}}{\epsilon'n\vol(\rmb_n)(\sqrt{n})^{n-1}\hat{\mu}_f(\ell_i(\overline{\rmk}_{\epsilon}, u, \epsilon'))|v'-u|^{n-1}} \nonumber \\
&\geq \frac{f(v')\vol(\rmq_{i,j}\cap b_{\sqrt{\epsilon}n^{1/4}}(v))}{2\epsilon'n\vol(\rmb_n)\hat{\mu}_f(\ell_i(\overline{\rmk}_{\epsilon}, u, \epsilon'))|v'-u|^{n-1}},
  \end{align}
  where we have used the fact that the distance between adjacent $\rmq_{i,j}$ and $\rmq_{i, j+1}$ can be bounded from below by $|\rmq_{i,j}\cap\ell_i|/(1+\epsilon'/2) \ge |\rmq_{i, j}\cap\ell_i|/2$.

  Now we consider the integration in $\rmq_{i, j}\cap b_{\sqrt{\epsilon}n^{1/4}}(v)$. We use $f(v)$ to approximate $f(v')$ which causes a relative error at most $e^{\sqrt{\epsilon}n^{1/4}}$, and use $|v'-u|^{n-1}$ to approximate $|x-u|^{n-1}$ for all $x \in \rmq_{i, j}\cap b_{\sqrt{\epsilon}n^{1/4}}(v)$ which causes a relative error at most $e$ provided $\epsilon' \leq \sqrt{\epsilon}n^{-3/4}$ (noting that the distance between $x$ and $u$ is at most $\sqrt{\epsilon}n^{1/4}$). We have
\begin{align}
  &\int_{\rmq_{i,j}\cap b_{\sqrt{\epsilon}n^{1/4}}(v)}\frac{f(x)\,\d x}{n\vol(\rmb_n)\mu_f(u, x)|x-u|^{n-1}} \nonumber \\
  &\qquad\qquad \leq \frac{2e^{\sqrt{\epsilon}n^{1/4}+2\epsilon'+1}f(v')}{\epsilon'n\vol(\rmb_n)\hat{\mu}_f(\ell_i(\overline{\rmk}_{\epsilon}, u, \epsilon'))|v'-u|^{n-1}}\int_{\rmq_{i,j}\cap b_{\sqrt{\epsilon}n^{1/4}}(v)}\,\d x \\
  &\qquad\qquad = \frac{2e^{\sqrt{\epsilon}n^{1/4}+2\epsilon'+1}f(v')\vol(\rmq_{i,j}\cap b_{\sqrt{\epsilon}n^{1/4}}(v))}{\epsilon'n\vol(\rmb_n)\hat{\mu}_f(\ell_i(\overline{\rmk}_{\epsilon}, u, \epsilon'))|v'-u|^{n-1}},
\end{align}
where the inequality follows from \eq{mu-muhat}. Let $i_1, \ldots, i_t$ be the indices such that $\rmp_{i_j} \cap b_{\sqrt{\epsilon}n^{1/4}}(v) \neq \emptyset$ for $j \in [t]$. We use $\bigcup_{j \in [t]}\rmp_{i_j}\cap b_{\sqrt{\epsilon}n^{1/4}}(v)$ as a partition to approximate $b_{\sqrt{\epsilon}n^{1/4}}(v)$, which causes a relative error at most $(1+\epsilon)^{n}$ for $\vol(b_{\sqrt{\epsilon}n^{1/4}}(v))$. We have
\begin{align*}
\int_{b_{\sqrt{\epsilon}n^{1/4}}(v)}\frac{f(x)\,\d x}{\mu_f(u, x)|x-u|^{n-1}} \leq (1+\epsilon)^n\sum_{j \in [t]} \int_{b_{\sqrt{\epsilon}n^{1/4}}(v)\cap \rmp_{i_j}}\frac{f(x)\,\d x}{\mu_f(u, x)|x-u|^{n-1}}.
\end{align*}
Hence,
  \begin{align}
    &\frac{(\sqrt{\epsilon})^n}{n^{n/4}}P'_u(b_{\sqrt{\epsilon}n^{1/4}}(v)) \nonumber \\
    &= \frac{2(\sqrt{\epsilon})^n}{n^{1+n/4}\vol(\rmb_n)}\int_{b_{\sqrt{\epsilon}n^{1/4}}(v)}\frac{f(x)\,\d x}{\mu_f(u, x)|x-u|^{n-1}} \\
                                 &\leq \frac{2(\sqrt{\epsilon})^n(1+\epsilon)^n}{n\vol(\rmb_n)}\sum_{j \in [t]} \int_{b_{\sqrt{\epsilon}n^{1/4}}(v) \cap \rmp_{i_j}}\frac{f(x)\,\d x}{\mu_f(u, x)|x-u|^{n-1}} \\
                                 &= \frac{2(\sqrt{\epsilon})^n(1+\epsilon)^n}{n\vol(\rmb_n)}\sum_{j \in [t]} \sum_{k \in [c_{i_j}]}\int_{b_{\sqrt{\epsilon}n^{1/4}}(v) \cap \rmq_{i_j, k}}\frac{f(x)\,\d x}{\mu_f(u, x)|x-u|^{n-1}} \\
                                 &\leq \sum_{j \in [t]} \sum_{k \in [c_{i_j}]}\frac{4(1+\epsilon)^ne^{\sqrt{\epsilon}n^{1/4}+2\epsilon'+1}(\sqrt{\epsilon})^nf(v')\vol(\rmq_{i,j}\cap b_{\sqrt{\epsilon}n^{1/4}}(v))}{\epsilon'n\vol(\rmb_n)\hat{\mu}_f(\ell_{i_j}(\overline{\rmk}_{\epsilon}, u, \epsilon'))|u-v'|^{n-1}} \\
                                 &\leq 4(1+\epsilon)^ne^{\sqrt{\epsilon}n^{1/4}+2\epsilon'+1}\sum_{j \in [t]} \sum_{k \in [c_{i_j}]}\frac{2\epsilon^{n-1}(\sqrt{\epsilon})^nf(v')}{n\vol(\rmb_n)(\sqrt{n})^{n-1}\hat{\mu}(\ell_{i_j}(\overline{\rmk}_{\epsilon}, u, \epsilon'))} \\
                                 &= 4(1+\epsilon)^ne^{\sqrt{\epsilon}n^{1/4}+2\epsilon'+1}P_{uv} \leq e^{5 + 2\epsilon'}P_{uv},
  \end{align}
  where the last inequality holds when $\epsilon \leq 1/n$.

  For $u \in \rmke$ and $v \in \rmkse$, we approximate $\int_{x \in b_{\epsilon}(u)}P'_u(b_{\sqrt{\epsilon}n^{1/4}}(v))\,\d\pi_f(x)$ by $\epsilon^nP_{uv}$. Note that for all $u' \in b_{\epsilon}(u)$, we have $|u'-v|^n \leq e|u-v|^n$. Also, the lengths of $\ell_{uv}$ and $\ell_{u'v}$ can differ by at most a factor of 2. As a result, $\hat{\pi}_f(\ell_{u'v}(\overline{\rmk}_{\epsilon}, u, \epsilon') \leq 2\hat{\pi}_f(\ell_{uv}(\overline{\rmk}_{\epsilon}, u', \epsilon')$. It follows that $P_{uv} \geq P_{u'v}/(2e)$. Therefore,
  \begin{align}
    \int_{x \in b_{\epsilon}(u)}P'_u(b_{\sqrt{\epsilon}n^{1/4}}(v))\,\d\pi_f(x) &\leq \int_{x \in b_{\epsilon}(u)}\frac{2e^{5+2\epsilon'}n^{n/4}}{(\sqrt{\epsilon})^n}P_{xv}\,\d\pi_f(x) \\
                                                                    &\leq \frac{2e^{5+2\epsilon'+\epsilon}n^{n/4}}{(\sqrt{\epsilon})^n}P_{uv}\hat{\pi}_f(u)\epsilon^n.
  \end{align}

  Next, for the relationship between $\hat{\pi}_f$ and $\pi_f$, we consider the sets $\overline{\rmk}_{\epsilon} \cap \rmk$, $\overline{\rmk}_{\epsilon} \setminus \rmk$, and $\rmk \setminus \overline{\rmk}_{\epsilon}$ separately. Without loss of generality, assume $\hat{\pi}_f(\rms) \leq \hat{\pi}_f(\rmke\setminus\rms)$. We partition $\rms$ as $\rms_1 \cup \rms_2$, where $\rms_1 = \{x \in \rms: b_{\epsilon}(x) \subseteq \rmk\}$ and $\rms_2 = \rms \setminus \rms_1$. We also define $\overline{\rms}_1 := \bigcup_{x \in \rms_1}b_{\epsilon}(x)$ and $\overline{\rms}_2 := \bigcup_{x \in \rms_2}b_{\epsilon}(x)$. For $\overline{\rms}_1$, we use $f(v)$ to approximate $f(x)$ for all $x \in b_{\epsilon}(v)$; it follows that
  \begin{align}
    \hat{\pi}_f(\rms_1) \leq e^{2\epsilon}\pi_f(\overline{\rms}_1)\quad\text{and}\quad\pi_f(\overline{\rms}_1) \leq e^{2\epsilon}\hat{\pi}_f(\rms_1).
  \end{align}
  For $\rms_2$, we have
  \begin{align}
    \hat{\pi}_f(\rms_2) \leq 2e^{2\epsilon}\pi_f(\overline{\rms}_2 \cap \rmk),
  \end{align}
  so
  \begin{align}
    \hat{\pi}_f(\rms) &= \hat{\pi}_f(\rms_1) + \hat{\pi}_f(\rms_2) \\
                                &\leq e^{2\epsilon}\pi_f(\overline{\rms}_1) + 2e^{2\epsilon} \pi_f(\overline{\rms}_2\cap\rmk)\leq 3\pi_f(\rmk\cap\overline{\rms}). \label{eq:rel-pihat-pi}
  \end{align}

  Now we bound the numerator of the conductance: $\sum_{u \in \rms}\sum_{v \in \rmke \setminus \rms} P_{uv}\hat{\pi}_f(u)$. For $u \in \rms$ and $v \in \rmk \setminus \rms$, we consider four cases. First, when $b_{\epsilon}(u), b_{\epsilon}(v) \subseteq \rmk$, we have
  \begin{align}
    P_{uv}\hat{\pi}_f(u) &\geq \frac{(\sqrt{\epsilon})^n}{2e^{5+2\epsilon'+\epsilon}n^{n/4}}\int_{x \in b_{\epsilon}(u)}P'_u(b_{\sqrt{\epsilon}n^{1/4}}(v))\,\d\pi_f(x).
  \end{align}
  Second, when $b_{\epsilon}(u) \subseteq \rmk$ and $b_{\epsilon}(v) \not\subseteq \rmk$, we have
  \begin{align}
    P_{uv}\hat{\pi}_f(u) &\geq \frac{(\sqrt{\epsilon})^n}{2e^{5+2\epsilon'+\epsilon}n^{n/4}}\int_{x \in b_{\epsilon}(u)}P'_u(b_{\sqrt{\epsilon}n^{1/4}}(v))\,\d\pi_f(x) \\
                                   &\geq \frac{(\sqrt{\epsilon})^n}{2e^{5+2\epsilon'+\epsilon}n^{n/4}}\int_{x \in b_{\epsilon}(u)}P'_u(b_{\sqrt{\epsilon}n^{1/4}}(v)\cap\rmk)\,\d\pi_f(x).
  \end{align}
  Third, when $b_{\epsilon}(u) \not\subseteq \rmk$ and $b_{\epsilon}(v) \subseteq \rmk$, we have
  \begin{align}
    P_{uv}\hat{\pi}_f(u) &\geq \frac{(\sqrt{\epsilon})^n}{2e^{5+2\epsilon'+\epsilon}n^{n/4}}\int_{x \in b_{\epsilon}(u)}P'_u(b_{\sqrt{\epsilon}n^{1/4}}(v))\,\d\pi_f(x) \\
                                   &\geq \frac{(\sqrt{\epsilon})^n}{2e^{5+2\epsilon'+\epsilon}n^{n/4}}\int_{x \in b_{\epsilon}(u) \cap \rmk}P'_u(b_{\sqrt{\epsilon}n^{1/4}}(v))\,\d\pi_f(x).
  \end{align}
  Fourth, when $b_{\epsilon}(u) \not\subseteq \rmk$ and $b_{\epsilon}(v) \not\subseteq \rmk$, we have
  \begin{align}
    P_{uv}\hat{\pi}_f(u) &\geq \frac{(\sqrt{\epsilon})^n}{2e^{5+2\epsilon'+\epsilon}n^{n/4}}\int_{x \in b_{\epsilon}(u)}P'_u(b_{\sqrt{\epsilon}n^{1/4}}(v))\,\d\pi_f(x) \\
                                  &\geq \frac{(\sqrt{\epsilon})^n}{2e^{5+2\epsilon'+\epsilon}n^{n/4}}\int_{x \in b_{\epsilon}(u) \cap \rmk}P'_u(b_{\sqrt{\epsilon}n^{1/4}}(v)\cap\rmk)\,\d\pi_f(x).
  \end{align}
  We also need to consider the set $\rmk \setminus \overline{\rmk}_{\epsilon}$. There exists a small subset $\mathrm{E} \subseteq \rmk \setminus \overline{\rmk}_{\epsilon}$ such that $\pi_f(\mathrm{E}) \leq \epsilon \pi_f(\rms)$. We need to consider the transition from $\mathrm{E}$ to $\subseteq \rmk \setminus \overline{\rmk}_{\epsilon} \setminus \mathrm{E}$: we have $\int_{x \in \mathrm{E}} P'_x(\rmk \setminus \overline{\rmk}_{\epsilon}\setminus\mathrm{E})\,\d\pi_f(x) \leq \pi_f(\mathrm{E}) \leq \epsilon\pi_f(\rms)$. Putting everything together, we have
  \begin{align}
   & \sum_{u \in \rms}\sum_{v \in \rmke \setminus \rms} P_{uv} \hat{\pi}_f(u) + \int_{x \in \mathrm{E}\cap\rmk}P'_x(\rmk \setminus \overline{\rmk}_{\epsilon}\setminus\mathrm{E})\,\d\pi_f(x) \nonumber \\
   &\qquad\qquad\qquad \geq \frac{1}{2e^{5+2\epsilon'+\epsilon}} \int_{x \in \overline{\rms} \cap \rmk \cup \mathrm{E}} P'_x(\rmk \setminus (\overline{\rms} \cap \rmk \cup \mathrm{E}))\,\d\pi_f(x),
  \end{align}
  which further implies that
  \begin{align}
    \nonumber \sum_{u\in\rms}\sum_{v \in \rmke\setminus\rms}P_{uv}\hat{\pi}_f(u) &\geq \frac{1}{2e^{5+2\epsilon'+\epsilon}} \int_{x \in \overline{\rms} \cap \rmk \cup \mathrm{E}} P'_x(\rmk \setminus (\overline{\rms} \cap \rmk \cup \mathrm{E}))\,\d\pi_f(x) - \epsilon\pi_f(\rms) \\
                                                                       &\geq \frac{1}{2e^{5+2\epsilon'+\epsilon}} \int_{x \in \overline{\rms} \cap \rmk \cup \mathrm{E}} P'_x(\rmk \setminus (\overline{\rms} \cap \rmk \cup \mathrm{E}))\,\d\pi_f(x) - \epsilon e^{\epsilon}\hat{\pi}_f(\rms).
\end{align}
By \prop{LV06-Theorem6.9}, we have
\begin{align}
    \phi(\rms) &= \frac{\sum_{u \in \rms}\sum_{v \in \rmke\setminus\rms}P_{uv}\hat{\pi}_f(u)}{\hat{\pi}_f(\rms)} \\
               &\geq \frac{1}{2e^{5+2\epsilon'+\epsilon}}\frac{\int_{x \in \overline{\rms} \cap \rmk \cup \mathrm{E}} P'_x(\rmk \setminus (\overline{\rms} \cap \rmk \cup \mathrm{E}))\,\d\pi_f(x)}{\hat{\pi}_f(\rms)} - \frac{\epsilon}{2e^{4+2\epsilon'}} \\
               &\geq \frac{1}{6e^{5+2\epsilon'+\epsilon}}\frac{\int_{x \in \overline{\rms} \cap \rmk \cup \mathrm{E}} P'_x(\rmk \setminus (\overline{\rms} \cap \rmk \cup \mathrm{E}))\,\d\pi_f(x)}{\pi_f(\overline{\rms}\cap\rmk) + \pi_f(\mathrm{E})} - \frac{\epsilon}{2e^{5+2\epsilon'}} \\
               \label{eq:conductance-proof-ineq}
               &\geq \frac{1}{10^{14}e^{5+2\epsilon'+\epsilon}n\sqrt{n}\ln(\frac{n\sqrt{n}}{\pi_f(\overline{\rms}\cap\rmk)})} - \frac{\epsilon}{2e^{5+2\epsilon'}} \\
               &\geq \frac{1}{10^{14}e^{5+2\epsilon'+\epsilon}n\sqrt{n}\ln(\frac{n\sqrt{n}}{(1-e^{-\epsilon}/2)e^{\epsilon}\hat{\pi}_f(\rms)})} - \frac{\epsilon}{2e^{5+2\epsilon'}},
\end{align}
where the third inequality follows from \eq{rel-pihat-pi}. The above inequality can then be simplified to
  \begin{align}
    \phi(\rms) \geq \frac{1}{10^{16}n\sqrt{n}\ln(\frac{2n\sqrt{n}}{\hat{\pi}_f(\rms)})} - \epsilon,
  \end{align}
which is exactly the claim in \thm{conductance}.
\end{proof}

The mixing time for the discrete hit-and-run walk can be bounded by the following corollary.
\begin{corollary}
  Let $\rmke$ be the discretization of convex body $\K$ that contains a unit ball and is contained in a ball with radius $R \leq \sqrt{n}$. Let the density function be $f(x) = e^{-a^Tx}$ having support $\K$ where $a = (1, 0, \ldots, 0)$. Let $\epsilon' \leq \sqrt{\epsilon}n^{-3/4}$. Let the initial distribution be $\sigma$ and the distribution after $m$ steps be $\sigma^m$. If $\sum_{x \in \rmke}\frac{\sigma(x)}{\hat{\pi}_f(x)}\sigma(x) \leq M$ then, after
  \begin{align}
    m \geq 10^{33}n^3\ln^2\frac{Mn\sqrt{n}}{\epsilon}\ln\frac{M}{\epsilon}
  \end{align}
  steps, we have $d_{\mathrm{TV}}(\sigma^m, \hat{\pi}_f) \leq \epsilon$.
\end{corollary}
\begin{proof}
  First note that, since $\sum_{x \in \rmke}\frac{\sigma(x)}{\hat{\pi}_f(x)}\sigma(x) \leq M$, the set $\rms = \{x: \frac{\sigma(x)}{\hat{\pi}_f(X)} > \frac{2M}{\epsilon}\}$ has measure $\sigma(\rms) \leq \epsilon/2$. Then a random point in $\rmke$ can be thought of as being generated with probability $1-\epsilon/2$ from a distribution $\sigma'$ satisfying $\frac{\sigma'(\rms')}{\hat{\pi}_f(\rms')} \leq 2M/\epsilon$ for any subset $\rms' \subseteq \rmke$ and with probability $\epsilon/2$ from some other distribution. As a consequence of~\thm{conductance}, for any such subset $\rms'$ with $\hat{\pi}_f(\rms') = p$, the conductance of $\rms'$ is at least
  \begin{align}
    \Phi_p = \frac{1}{10^{16}n\sqrt{n}\ln(2n\sqrt{n}/p)} - \epsilon.
  \end{align}
  For the purpose of analysis, we use $p = \frac{\epsilon^2}{8M}$. When $\epsilon$ is reasonably small (say, $\epsilon \leq \frac{1}{2\cdot 10^{16}n\sqrt{n}\ln(Mn\sqrt{n}/\epsilon)}$), the $\epsilon$ term in the conductance bound can be ignored with an additional $1/2$ factor. Then we have $\Phi_p \geq \frac{1}{2\cdot10^{16}n\sqrt{n}\ln(2n\sqrt{n}/p)}$. By the condition that $\sigma'(\rms') \leq (2M/\epsilon) \hat{\pi}_f(\rms')$, as well as the way a random point in $\rmke$ is generated, \prop{mixing-phip} implies that
  \begin{align}
    d_{\text{TV}}(\sigma^{(m)}, \hat{\pi}_f) \leq \frac{\epsilon}{2} + \left(1-\frac{\epsilon}{2}\right)\left(\frac{\epsilon}{2} + \frac{4M}{\epsilon}\left(1-\frac{\Phi_p^2}{2}\right)^m\right).
  \end{align}
Therefore, after the claimed number of steps, the total variation distance is at most $\epsilon$.
\end{proof}

As the uniform distribution is a special case of a log-concave distribution, the proof of~\thm{conductance} also applies to this case. More specifically, we use \prop{conductance-uniform} in \eq{conductance-proof-ineq}, which yields the following stronger corollary.
\begin{corollary}\label{cor:conductance-uniform}
  Let $\rmke$ be the discretization of a convex body $\rmk$ that contains a unit ball and is contained in a ball with radius $R \leq \sqrt{n}$. Let $\epsilon' \leq \sqrt{\epsilon}n^{-3/4}$. The conductance of the hit-and-run walk in $\rmke$ with uniform distribution satisfies
  \begin{align}\label{eq:conductance-uniform}
    \phi \geq \frac{1}{2^{26}n\sqrt{n}} - \epsilon.
  \end{align}
\end{corollary}
\noindent
Note that \cor{conductance-uniform} is stronger than \thm{conductance} because \eq{conductance-uniform} is independent of $\rms \subseteq \rmke$. This corollary is informative and is not used in this paper.

\subsection{Implementing the quantum walk operators}\label{sec:impl-quant-walk}

We now describe how to implement the discretized quantum walk.
Following \eqn{ball-relationship}, consider a convex body $\K$ such that $\rmb_2(0,r) \subseteq \K \subseteq \rmb_2(0,R)$. Each stage of the volume estimation algorithm involves a hit-and-run walk over the convex body with target density $e^{-a x_0}$. In order to use techniques from \cite{wocjan2008speedup} to obtain a speedup in mixing time, we implement the quantum walk operator $W$ corresponding to an $\epsilon$-discretized version of this walk \algo{hit-and-run-1step}.

Let $|x\>$ be the register for the state of the walk, and $U$ be a unitary that satisfies $U|x\>|0\> = |x\>|p_x\>$ for all $|x\>$ (recall that $|p_x\> = \sum_{y \in \K_\epsilon}\sqrt{p_{x \to y}}|y\>$ where $p_{x \to y}$ is the probability of a transition from $x$ to $y$). Since the state of the hit-and-run walk is given by points on an $\epsilon$-grid that can be restricted to $\rmb_2(0,R)$, there are $\left(\frac{2R}{\epsilon}\right)^n$  possible values of $x$ and thus $|x\>$ can be represented using $n \log\left(\frac{2R}{\epsilon}\right)$ qubits. In the rest of the section, we abuse notation by letting $x$ refer to both a point on the grid and its corresponding bit representation. Then the quantum walk operator \cite{wocjan2008speedup} can be realized as
  \begin{equation}
    W' = U^\dagger S U R_{\mathcal{A}} U^\dagger S U R_{\mathcal{A}}
  \end{equation}
  where $R_{\mathcal{A}}$ is the reflection around the subspace $\mathcal{A} = \spn \{|x\>|0\> \mid x \in \K_\epsilon \}$ and $S$ is the swap operator. It thus remains to implement the operator $U$.

\paragraph{Continuous case}
  We first explain a continuous version of the implementation before explaining how it can be discretized.
  Given an input $|x\>$, consider $n$ real ancilla registers, each in the state $\int_0^1 |z\> \,\d z$. Given a pair of uniformly distributed random variables $\xi_1,\xi_2$, the Box-Muller transform
  \begin{align}
    \label{eq:box-muller}
    \phi_1 &= \sqrt{-2\delta^2 \ln{\xi_1}}\cos{2\pi \xi_2} \\
    \phi_2 &= \sqrt{-2\delta^2 \ln{\xi_1}}\sin{2\pi \xi_2}
  \end{align}
  yields two variables $\phi_1,\phi_2$ that are distributed according to a univariate normal distribution with mean $0$ and variance $\delta^2$. Thus applying the unitary mapping
  \begin{align}
    |\xi_1\>|\xi_2\> \mapsto \big|\sqrt{-4 \ln{\xi_1}}\cos{2\pi \xi_2}\big\>\big|\sqrt{-4 \ln{\xi_1}}\sin{2\pi \xi_2}\big\>
  \end{align}
  to $\int_0^1|z\>\,\d z \otimes \int_0^1|z\>\,\d z$ yields the state $\int_\R \frac{1}{\sqrt{4\pi}}e^{-z^2/4}|z\>\,\d z \otimes \int_\R \frac{1}{\sqrt{4\pi}}e^{-z^2/4}|z\>\,\d z$. With $n$ such registers, we have the state
  \begin{align}
    \int_{\R^n} \frac{1}{\sqrt{4\pi}}e^{-(\sum_{i=1}^n z_{i}^{2}/4)}|z\>\,\d z.
  \end{align}
  We now compute the unit vector (direction) corresponding to each $x$ in a different ancilla register, and uncompute the Gaussian registers. Since $\frac{1}{\sqrt{4\pi}}e^{-(\sum_{i=1}^{n}x_{i}^{2}/4)}$ is independent of the direction of the vector $z$, we obtain a uniform distribution over all the directions on the $n$-dimensional sphere $\S^n$ given by
  \begin{align}
    \sqrt{\frac{n\pi^{n/2}}{\Gamma\left(n + \frac{1}{2}\right)}}\int_{\S^n} |u\>\,\d u.
  \end{align}
  Corresponding to each direction $u$, the line $\{x + tu : t \in \R\}$ intersects the convex body $\K$ at two points with parameters $t_1,t_2$. These points as well as the length $l(u) = |t_1 - t_2|$  can be determined within error $\epsilon$ using $O(\log\frac{1}{\epsilon})$ calls to the membership oracle. We must now map each direction $|u\>$ to a superposition proportional to $\int_{t_1}^{t_2} e^{a^T(x + tu)/2} |x + tu\>\,\d t = \int_{t_1}^{t_2} e^{a_0(x_0 + tu_0)/2} |x + tu\>\,\d t$. Since the exponential distribution is efficiently integrable, this can be easily effected by making a variable change starting from the state $\int_0^1 |z\>\,\d z$. The normalization factor is
\begin{align}
A := \sqrt{\frac{a_0u_0}{e^{-a_0x_0}(e^{-a_0t_1}-e^{-a_0t_2})}}.
\end{align}
Consider the variable change $f\colon[0,1]\to[t_1,t_2]$ such that $\frac{\d f^{-1}(t)}{\d t} = Ae^{a_0(x_0 + tu_0)/2}$, $f(0) = t_1$, $f(1) = t_2$. Applying $f$ to $\int_0^1|z\>\,\d z$ produces $\int_{t_1}^{t_2} Ae^{a_0(x_0 + tu_0)/2} |t\>\,\d t$, which can be transformed to $\int_{t_1}^{t_2} e^{a_0(x_0 + tu_0)/2} |x + tu\>\,\d t$ with an operation controlled on the input register $x$. This produces the appropriate superposition over points corresponding to each direction.

  \paragraph{Discrete case}
  The operator $U$ can be implemented in a discrete setting using a similar process to the continuous case with two main changes:
  \begin{itemize}
  \item Instead of a continuous uniform variable $\int_0^1 |z\>\,\d z$ we use a discrete uniform distribution. We can create a uniform distribution on a grid with spacing $\epsilon$ as follows. We take $n$ sets of ancilla registers, each consisting of $\log\left(1/\epsilon\right)$ registers initialized to the state $0$. We apply Hadamard gates to each of these registers, giving the superposition $\bigotimes_{i=1}^n \sqrt{\epsilon} \sum_{z_i=0}^{1/\epsilon -1}|z_i\>$. Each $|z\>$ can be mapped to $|z\epsilon\>$, producing the required uniform distribution over the grid.
  \item Applying a bijective mapping to a discrete uniform distribution simply relabels the states, so the change of variable methods used in the continuous setting cannot be used to construct the Gaussian and exponential superpositions. We use instead the Grover-Rudolph method \cite{grover2002creating} that prepares states with amplitudes corresponding to efficiently integrable probability distributions. Exponential distributions can be analytically integrated, and an $n$-dimensional Gaussian variable is a product of $n$ univariate standard normal distributions, each of which can be efficiently integrated by Monte Carlo methods.
  \end{itemize}
  Given a point $u \in K_\epsilon$ and a line $l(u,\epsilon)$ to be approximately uniformly sampled, we determine the range of points in $l(\overline{\K}_\epsilon,u,\epsilon')$ using binary search with the membership oracle and prepare an exponential superposition as described above. We apply a unitary mapping to compute the closest point $v'' \in \K_{\sqrt{\epsilon}n^{1/4}}$. Finally, corresponding to each point $v''$, we generate a uniform distribution over an $\epsilon$ grid in $b_{\sqrt{\epsilon}n^{1/4}} \cap (\R_n)_\epsilon$ by applying the Hadamard transform to $\log(n^{1/4}/\sqrt{\epsilon})$ qubits.

Overall, this implementation of the discretized quantum hit-and-run walk operator gives the following.
\begin{theorem}\label{thm:quantum-hit-and-run-implementation}
The gate complexity of implementing an operator $\tilde{U}$ such that $\lVert \tilde{U} - U \rVert=O(\epsilon)$ where $U|x\>|0\> = |x\>\sum_{y \in \K_\epsilon}\sqrt{p_{x \to y}}|y\>$ is $\tilde{O}\left(n \log\left(\frac{1}{\epsilon}\right)\right)$. The correspondsing quantum walk operator $W$ can be implemented using a constant number of calls to $U$.
\end{theorem}


\section{Quantum lower bounds for volume estimation}\label{sec:quantum-lower}
\subsection{A quantum lower bound in $n$}\label{sec:quantum-lower-n}
In this subsection, we prove the following quantum query lower bound in $n$ for volume estimation:
\begin{theorem}\label{thm:lower-bound}
Suppose $0<\epsilon<\sqrt{2}-1$. Estimating the volume of $\K$ with multiplicative precision $\epsilon$ requires $\Omega(\sqrt{n})$ quantum queries to the membership oracle $O_{\K}$ defined in \eqn{oracle-defn}.
\end{theorem}

\begin{proof}
We prove \thm{lower-bound} by reduction from the Hamming weight problem. In~\cite{nayak1999quantum} by Nayak and Wu, it is shown that if we are given an oracle $O_s\colon \ket{i,b}\mapsto\ket{i,b\oplus s_i}$ for an input $n$-bit string $s=(s_1,\ldots,s_n)\in\{0,1\}^n$, and given the promise that the Hamming weight of $s$ is either 0 or 1, it takes $\Omega(\sqrt{n})$ quantum queries to decide which is the case.

To establish an $\Omega(\sqrt{n})$ lower bound for volume estimation,
for an $n$-bit string $s\in\{0,1\}^n$ with Hamming weight $|s|_{\text{Ham}}\leq 1$, we consider the convex body
$\K=\bigtimes_{i=1}^n [0,2^{s_i}]$. The volume of $\K$ is $2^{|s|_{\text{Ham}}}\in\{1,2\}$, and membership in $\K$ is determined by the function
\begin{align}
\textsc{mem}_s(x):=\left\{\begin{array}{ll}
    1 & \text{if for each } i\in\range{n}, 0\leq x_i\leq 2^{s_i}, \\
    0 & \text{otherwise}.
    \end{array}\right.
\end{align}
The corresponding membership oracle $O_{\K}$ (defined in \eqn{oracle-defn}) can be simulated by querying $O_s$ using \algo{lower-bound}.

\begin{algorithm}[htbp]
\KwInput{A vector $x=(x_1,\ldots,x_n)\in\mathbb{R}^n$.}
\KwOutput{$\textsc{mem}_s(x)$.}
\For{$i=1,\ldots,n$}{\label{lin:lb1}
\If{$x_i>2$ \emph{or} $x_i<0$}{Return 0\;}
Set $y_i=1$ if $x_i>1$ and 0 otherwise\;
}
\eIf{$|y|_{\emph{Ham}}>1$}{Return 0\;}{
\eIf{$|y|_{\emph{Ham}}=1$}{Find $i$ such that $y_i=1$. Return $O_s(i)$\;\label{lin:lb2}}{Return 1\;}
}
\caption{Simulating $\textsc{mem}_s$ with one query to $O_s$.}
\label{algo:lower-bound}
\end{algorithm}

We now prove that for any positive integer $k$ and $s\in\{0,1\}^n$ with $|s|_{\text{Ham}}\leq 1$, if there is a $k$-query algorithm that computes the volume with access to $\textsc{mem}_s$, then there is a $k$-query algorithm for deciding whether $|s|_{\text{Ham}}>0$ with access to $O_s$. We first show that \algo{lower-bound} simulates the oracle $\textsc{mem}_s$. In the for loop of \lin{lb1}, we know that $y_i=1$ if and only if $1<x_i\leq 2$, which is inside the convex body if $s_i=1$. The case $|y|_{\text{Ham}}>1$ implies that there exist two distinct coordinates $i,j$ such that $x_i,x_j>1$, which implies that $x$ lies outside the convex body. Now we are left with the cases $|y|_{\text{Ham}}=1$ or 0. In \lin{lb2}, $y_i=1$ implies $1<x_i\leq 2$, which lies in the convex body if and only if $s_i=O_s(i)=1$. Also, $|y|=0$ implies that for every coordinate $i$, $0\leq x_i\leq 1$, which lies in the body for all $s$.

Finally, if there is a $k$-query algorithm that computes an estimate $\widetilde{\vol(\K)}$ of the volume of $\K$ up to multiplicative precision $0<\epsilon<\sqrt{2}-1$, then $s=\lceil\log_{2}\widetilde{\vol(\K)}\rfloor$ where $\lceil\cdot\rfloor$ returns the nearest integer. This immediately gives a $k$-query algorithm that decides whether $|s|_{\text{Ham}}=0$ or 1. Since there is an $\Omega(\sqrt{n})$ quantum query lower bound for this task, the $\Omega(\sqrt{n})$ lower bound on volume estimation follows.
\end{proof}

\begin{remark}
The proof of \thm{lower-bound} has similarity to \cite[Section 5]{vanApeldoorn2018convex}.
\end{remark}

\subsection{An optimal quantum lower bound in $1/\epsilon$}\label{sec:quantum-lower-eps}
In this subsection, we prove:
\begin{theorem}\label{thm:lower-bound-eps}
Suppose $1/n\leq\epsilon\leq 1/3$. Estimating the volume of $\K$ with multiplicative precision $\epsilon$ requires $\Omega(1/\epsilon)$ quantum queries to the membership oracle $O_{\K}$ defined in \eqn{oracle-defn}.
\end{theorem}

Comparing with \thm{main}, this shows that our quantum algorithm for volume estimation is optimal in $1/\epsilon$ up to poly-logarithmic factors.

The proof constructs a convex body whose volume encodes the Hamming weight of a string. A membership oracle for this convex body can be implemented by querying the bits of the string. Then the tight lower bound of Nayak and Wu on the quantum query complextiy of approximating the Hamming weight~\cite{nayak1999quantum} implies a lower bound on the query complexity of volume estimation.

We construct the convex body by attaching hyperpyramids to the faces of the $n$-dimensional unit hypercube. The axis of each hyperpyramid is aligned with the axis of the face of the hypercube it corresponds to, and the height of the hyperpyramid is $1/2$. More concretely, if the unit hypercube is $H_{n}:=[-1/2,1/2]^{n}$, then the two hyperpyramids on the face perpendicular to the $i^{\text{th}}$ axis are
\begin{align}
P_{i,+}&:=\big\{x: x_{i}\geq 1/2, |x_{k}|+|x_{i}|\leq 1~\forall\,k\in\range{n}/\{i\}\big\}; \label{eqn:lower-pyramid-1} \\
P_{i,-}&:=\big\{x: x_{i}\leq -1/2, |x_{k}|+|x_{i}|\leq 1~\forall\,k\in\range{n}/\{i\}\big\}. \label{eqn:lower-pyramid-2}
\end{align}
We denote the convex body with all hyperpyramids attached by
\begin{align}
\cvx_{n}:=H_{n}\cup\Big(\bigcup_{i=1}^{n}(P_{i,+}\cup P_{i,-})\Big).
\end{align}
For illustration, the $3$-dimensional convex body $\cvx_{3}$ is shown in \fig{lower-eps}.

\begin{figure}
\centering
\includegraphics[width=0.4\textwidth]{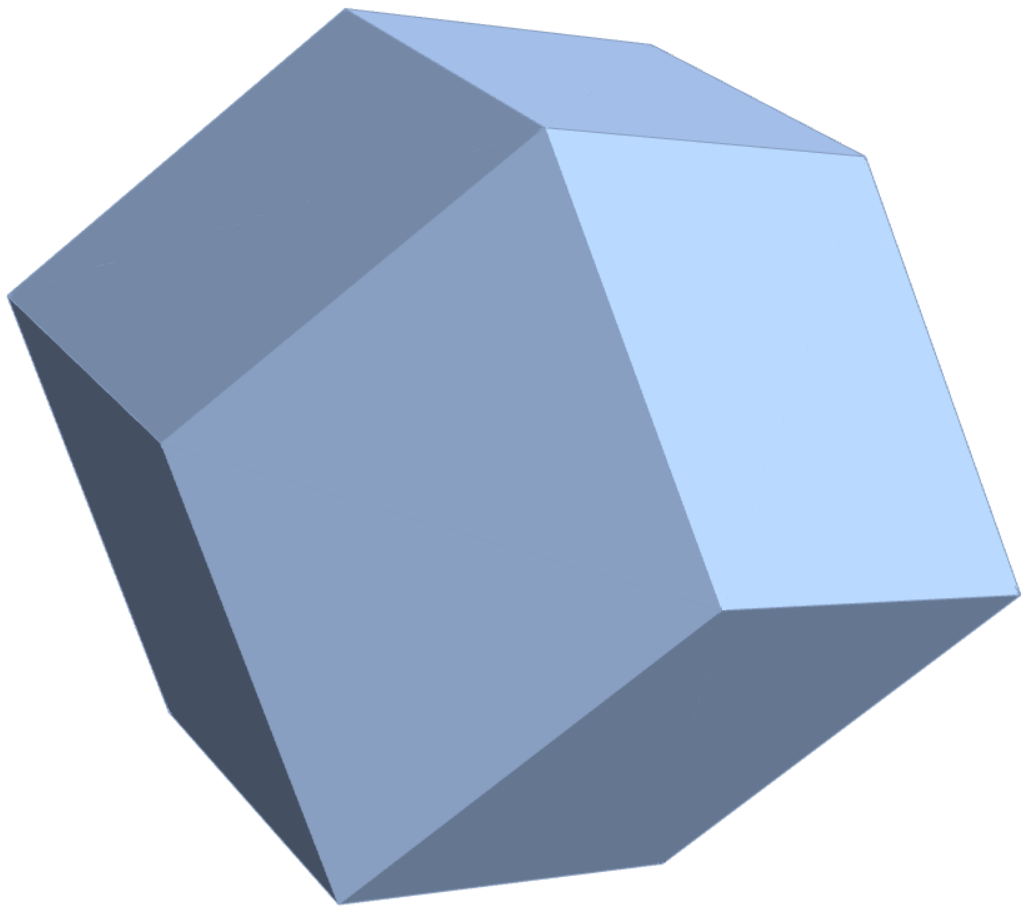}
\caption{The convex body $\cvx_{3}$.
}
\label{fig:lower-eps}
\end{figure}

We first prove:
\begin{lemma}\label{lem:lower-convexity}
$\cvx_{n}$ is convex for all $n\in\N$.
\end{lemma}

\begin{proof}
It suffices to show that if $x,y\in\cvx_{n}$ and $\alpha\in[0,1]$, then $\alpha x+(1-\alpha)y\in\cvx_{n}$. We consider three cases:

\paragraph{Case 1: $x,y\in H_{n}$}
This case is straightforward as $H_{n}$ is convex, hence $\alpha x+(1-\alpha)y\in H_{n}\subset\cvx_{n}$.

\paragraph{Case 2: $x\in \bigcup_{i=1}^{n}(P_{i,+}\cup P_{i,-}), y\in H_{n}$}
Let $i^* \in \range{n}$ such that $x\in P_{i^{*},+}\cup P_{i^{*},-}$. Then by \eqn{lower-pyramid-1} and \eqn{lower-pyramid-2}, $|x_{i^{*}}|\geq 1/2$ and $|x_{i}|+|x_{i^{*}}|\leq 1\ \forall i\in\range{n}\setminus\{i^{*}\}$, which implies $|x_{i}|\leq 1/2\ \forall i\in\range{n}\setminus\{i^{*}\}$. Also note that $y\in H_{n}$ implies $|y_{i}|\leq 1/2\ \forall i\in\range{n}$. Therefore,
\begin{align}
|\alpha x_{i}+(1-\alpha)y_{i}|\leq \alpha|x_{i}|+(1-\alpha)|y_{i}|\leq\frac{\alpha}{2}+\frac{1-\alpha}{2}=\frac{1}{2}\quad\forall\,i\in\range{n}/\{i^{*}\}.
\end{align}
If $|\alpha x_{i^{*}}+(1-\alpha)y_{i^{*}}|\leq 1/2$, then $\alpha x+(1-\alpha)y\in H_{n}\subseteq\cvx_{n}$. If $|\alpha x_{i^{*}}+(1-\alpha)y_{i^{*}}|>1/2$, then
\begin{align}
|\alpha x_{i^{*}}+(1-\alpha)y_{i^{*}}|+|\alpha x_{i}+(1-\alpha)y_{i}|&\leq \alpha(|x_{i^{*}}|+|x_{i}|)+(1-\alpha)(|y_{i^{*}}|+|y_{i}|) \\
&\leq\alpha+(1-\alpha)=1\quad\forall\,i\in\range{n}\setminus\{i^{*}\}.
\end{align}
Therefore, by \eqn{lower-pyramid-1} and \eqn{lower-pyramid-2} we have $\alpha x+(1-\alpha)y\in P_{i^{*},+}\cup P_{i^{*},-}\subset\cvx_{n}$. In any case, we always have $\alpha x+(1-\alpha)y\in\cvx_{n}$.

\paragraph{Case 3: $x,y\in \bigcup_{i=1}^{n}(P_{i,+}\cup P_{i,-})$}
Let $i^*,j^* \in \range{n}$ such that $x\in P_{i^{*},+}\cup P_{i^{*},-}$ and $y\in P_{j^{*},+}\cup P_{j^{*},-}$. If $i^{*}=j^{*}$, the proof is identical to that of Case 2 and we omit the details here. It remains to consider the case $i^{*}\neq j^{*}$. Then we have $|x_{i}|,|y_{i}|\leq 1/2\ \forall i\in\range{n}\setminus\{i^{*},j^{*}\}$. In addition,
\begin{align}
\nonumber |\alpha x_{i^{*}}+(1-\alpha)y_{i^{*}}|+|\alpha x_{j^{*}}+(1-\alpha)y_{j^{*}}|&\leq\alpha(|x_{i^{*}}|+|x_{j^{*}}|)+(1-\alpha)(|y_{j^{*}}|+|y_{i^{*}}|) \\
&\leq\alpha+(1-\alpha)=1
\end{align}
by \eqn{lower-pyramid-1} and \eqn{lower-pyramid-2}. This means that at most one of $|\alpha x_{i^{*}}+(1-\alpha)y_{i^{*}}|$ and $|\alpha x_{j^{*}}+(1-\alpha)y_{j^{*}}|$ can be more than $1/2$. If neither of them is more than $1/2$, then $\alpha x+(1-\alpha)y\in H_{n}\subset\cvx_{n}$. If exactly one of them is more than $1/2$, say $|\alpha x_{i^{*}}+(1-\alpha)y_{i^{*}}|>1/2$ and $|\alpha x_{j^{*}}+(1-\alpha)y_{j^{*}}|\leq 1/2$, then $\alpha x+(1-\alpha)y\in P_{i^{*},+}\cup P_{i^{*},-}\subset\cvx_{n}$. In any case, we always have $\alpha x+(1-\alpha)y\in\cvx_{n}$.
\end{proof}

We use the following lower bound on the quantum query complexity of approximating the Hamming weight:

\begin{proposition}[\cite{nayak1999quantum}]\label{prop:Nayak-Wu}
Suppose we are given the quantum oracle $O_{s}|i\>|0\>=|i\>|s_{i}\>\ \forall i\in\range{n}$ for some $s\in\{0,1\}^{n}$. Let $0\leq l<l'\leq n$ be two integers, $\Delta=|l-l'|$, and $m\in\{l,l'\}$ such that $|\frac{n}{2}-m|$ is maximized. Then the quantum query complexity of determining whether $s$ has Hamming weight at most $l$ or at least $l'$ is $\Theta(\sqrt{n/\Delta}+\sqrt{m(n-m)}/\Delta)$.
\end{proposition}

Now we can prove \thm{lower-bound-eps}.
\begin{proof}
Given a binary string $s\in\{0,1\}^{n}$, we consider the convex body
\begin{align}\label{eqn:cvx-s}
\cvx_{s}:=H_{n}\cup\Big(\bigcup_{i\colon s_{i}=1}(P_{i,+}\cup P_{i,-})\Big).
\end{align}
By \lem{lower-convexity} and the fact that each hyperpyramid is the intersection of $\cvx_{n}$ and the convex spaces $\{x:x_{i}\geq 1/2\}$ or $\{x:x_{i}\geq 1/2\}$, $\cvx_{s}$ is also convex. Furthermore, a query to the membership oracle in \eqn{oracle-defn} for $\cvx_{s}$ can be implemented using one query to the binary string oracle $O_{s}$: queries to points outside $\cvx_{n}$ or inside $H_{n}$ are trivially answered with 0 and 1, respectively, whereas queries to points in $P_{i,+}\cup P_{i,-}$ should return $s_{i}$. Also note that for each $i\in\range{n}$, the volume of the hyperpyramid $P_{i,+}$ is
\begin{align}
\vol(P_{i,+})=\int_{0}^{1/2}(1-2t)^{n-1}\d t=\frac{1}{2n}
\end{align}
since the intersection of $P_{i,+}$ and $\{x:x_{i}=1/2+t\}$ is an $(n-1)$-dimensional hypercube with side-length $1-2t$ and hence volume $(1-2t)^{n-1}$. By symmetry, we also have $\vol(P_{i,-})=\frac{1}{2n}$. Therefore
\begin{align}
\vol(\cvx_{s})&=\vol(H_{n})+\sum_{i:s_{i}=1}\big(\vol(P_{i,+})+\vol(P_{i,-})\big) \\
&=1+|s|_{\text{Ham}}\cdot \frac{2}{2n}=1+\frac{|s|_{\text{Ham}}}{n}.
\end{align}
In other words, estimating the volume of $\cvx_{s}$ with multiplicative error $\epsilon$ is equivalent to the Hamming distance problem with $\Delta=4\epsilon n$. Taking $m=\frac{n}{2}+\Delta$ in \prop{Nayak-Wu}, we find that the quantum query complexity of estimating the volume of $\cvx_{s}$ is at least
\begin{align}
\Omega\Big(\sqrt{\frac{n}{\epsilon n}}+\frac{\sqrt{n^{2}/4-\epsilon^{2}n^{2}}}{\epsilon n}\Big)=\Omega\Big(\frac{1}{\epsilon}\Big)
\end{align}
for any $1/n\leq\epsilon\leq 1/3$.
\end{proof}

\begin{remark}
The same proof strategy implies a classical lower bound of $\Omega(1/\epsilon^{2})$ for volume estimation if we replace \prop{Nayak-Wu} by its folklore classical counterpart. In particular, this shows that our quantum algorithm in \thm{main} achieves a provable quadratic quantum speedup in $1/\epsilon$.
\end{remark}

\begin{remark}
Although the proofs of both theorems consider well-rounded convex bodies, this assumption can be simply waived by assuming known multiplicative rescaling factors $c_{1},\ldots,c_{n}$ along all the $n$ directions. The proofs follow from the same arguments.
\end{remark}


\section*{Acknowledgements}
We thank anonymous reviewers for helpful suggestions on earlier versions of this paper, including one of them that pointed out a minor technical issue in the nondestructive amplitude estimation circuit (\fig{SA-block-nondes}). TL and CW thank Nai-Hui Chia for helpful discussions. SC, AMC, TL, and XW were supported in part by the U.S. Department of Energy, Office of Science, Office of Advanced Scientific Computing Research, Quantum Algorithms Teams program. AMC also received support from the Army Research Office (MURI award W911NF-16-1-0349), the Canadian Institute for Advanced Research, and the National Science Foundation (grant CCF-1813814). SHH received support from the U.S. Department of Energy, Office of Science, Office of Advanced Scientific Computing Research, Quantum Testbed Pathfinder program under Award Number {DE-SC0019040}. TL also received support from an IBM Ph.D.\ Fellowship and an NSF QISE-NET Triplet Award (grant DMR-1747426). CW was supported by Scott Aaronson's Vannevar Bush Faculty Fellowship from the US Department of Defense. XW also received support from the National Science Foundation (grants CCF-1755800 and CCF-1816695).


\end{document}